\documentclass[a4paper,11pt,openright,oneside]{book}

\usepackage[utf8]{inputenc}
\usepackage[T1]{fontenc}
\usepackage{lmodern}
\usepackage[x11names]{xcolor}
\usepackage{amssymb}
\usepackage{graphicx}
\usepackage[colorlinks,citecolor=blue ,urlcolor=blue ,linkcolor=blue]{hyperref}
\usepackage{amsmath}
\usepackage{amsthm}
\usepackage{braket}
\usepackage{mathrsfs}
\usepackage{dsfont}
\usepackage{enumerate}
\usepackage[titletoc]{appendix}
\usepackage{datetime}
\usepackage[margin=2.5cm]{geometry} 
\usepackage{soul} 

\newcommand{\bracket}[2]{\ensuremath{\langle#1 \vphantom{#2}| #2\vphantom{#1}\rangle}}

\newcommand{\ketbra}[2]{\ensuremath{|#1 \vphantom{#2}\rangle \langle #2\vphantom{#1}|}}

\newcommand{\norm}[1]{\left|\left|#1\right|\right|}
\newcommand{\unit}{\mathds{1}}

\newcommand{\tr}[1]{\mathrm{Tr}\left(#1\right)}

\newcommand{\trp}[2]{\mathrm{Tr}_\mathrm{#1}\left(#2\right)}
\newcommand{\abs}[1]{\left| #1 \right|} 

\newcommand{\im}{\mathrm{i}}

\newcommand{\com}[1]{}

\newcommand\scalemath[2]{\scalebox{#1}{\mbox{\ensuremath{\displaystyle #2}}}}

\usepackage[format=plain,labelsep=quad,labelformat=simple,
labelfont={sf,bf},textfont=small, font=small]{caption}
\DeclareCaptionLabelFormat{dodo}{\quad #1 #2}
\usepackage[numbers,comma,sort&compress,sectionbib]{natbib}
\usepackage{doi}
\usepackage{hyperref}

\usepackage{amsthm}

\usepackage{mathtools}
\let\oldproof\proof
\renewcommand{\proof}{\color{darkgray}\oldproof}

\theoremstyle{plain}
\newtheorem{postulate}{Postulate}
\newtheorem{theorem}{Theorem}[chapter]
\newtheorem{definition}{Definition}[chapter]

\theoremstyle{definition}
\newtheorem*{example}{Example}
\newtheorem*{remarks}{Remarks}

\usepackage[skins,breakable]{tcolorbox}

\usepackage{lmodern,adjustbox}
\usepackage{tikz}
\usetikzlibrary{quantikz}
\tcbuselibrary{listings}
\tcbuselibrary{breakable}
\begin{document}
\frontmatter	  
\thispagestyle{empty}
\begin{center}
    {\Huge\bfseries Quantum information theory\par}
    \vspace{0.5cm}
   {\LARGE \bfseries Lecture notes\par}
    \vspace{4cm}
    {\Large Christoph Dittel \\ Albert-Ludwigs-Universität Freiburg \par}
    \vspace{2cm}
	{\Large Summer semester 2022 \\ \vspace{0.5cm} (Version from November 2022)}
\end{center}
\clearpage

\addtocounter{page}{-1}
\chapter*{Preamble}
This lecture provides an introduction to \textit{quantum information} and \textit{quantum computation}, which are strongly related disciplines and subject of intense research. The lecture notes contain only a small selection of topics in these disciplines, with the aim of providing you with an overview and a basic introduction. Please do \textit{not} see these lecture notes as a replacement of a proper textbook. Instead, I recommend to read any textbooks of your choice alongside. Given that most parts of these notes are based on the textbook \textit{Quantum Computation and Quantum Information} by M. A. Nielsen and I. L. Chuang, I am sure that this book will help you deepening the discussed topics. However, since some topics in this lecture are not properly discussed there, please also consult other literature. 

The version from 2022 is the first version of my lecture notes on quantum information theory. If you find any errors or have suggestions for improvement, please don't hesitate to contact me, either in person or by email: \href{mailto:christoph.dittel@physik.uni-freiburg.de}{christoph.dittel@physik.uni-freiburg.de}. \par 
\vspace{0.5cm}

\begin{flushright}
Christoph Dittel, April 2022
\end{flushright}

{\hypersetup{linkcolor=black}
\tableofcontents
}

\mainmatter	  
\chapter*{Introduction}
\addcontentsline{toc}{chapter}{Introduction}
\chaptermark{Introduction}
There is plenty of literature about quantum information and quantum computation. Most of the topics discussed in this lecture can be found in the following references: \nocite{Nielsen-QC-2011} \nocite{Benenti-PQ-2007} \nocite{Preskill-QI-1998} \nocite{Mintert-BC-2009}

%


\vspace{0.5cm}
Let's start with asking what we need in order to deal with information. We need to encode it into something physical! Why? Think of the hard drive of your computer, or an old-fashioned music tape. If you want to store some information, you need a physical carrier. For example a set of capacitors, which can be charged or uncharged, or a magnetic tape, where the magnetization can be aligned. Next, we also know that our world is well described by quantum mechanics. Hence, if we study information theory, we better consider it from a quantum mechanical point of view. This is exactly what we will do in this lecture. 

What is then the difference between information (e.g. a music track) on a magnetic tape and information stored in a set of capacitors? In the end nothing. After decoding, the stored information must be independent on the physical carrier. In principle it doesn't matter whether you store your music track on tape, on vinyl, or on a SSD. In the end you want to perform a measurement and get back the information of your track as good as possible. Classically, this can be illustrated as follows: 
\begin{center}
\includegraphics[width=0.8\linewidth]{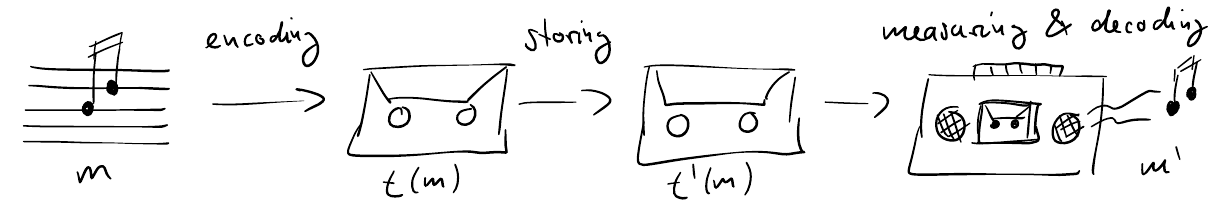}
\end{center}
If your original information $m$ is disrupted while being stored or processed (e.g. through scratches on the tape), then you get back $m'$ instead of $m$. The challenge (of our industry) is to find a convenient information carrier such that $m'$ is as close to $m$ as possible. 

Quantum mechanically we can formulate a similar schema:
\begin{center}
\centering
\includegraphics[width=0.8\linewidth]{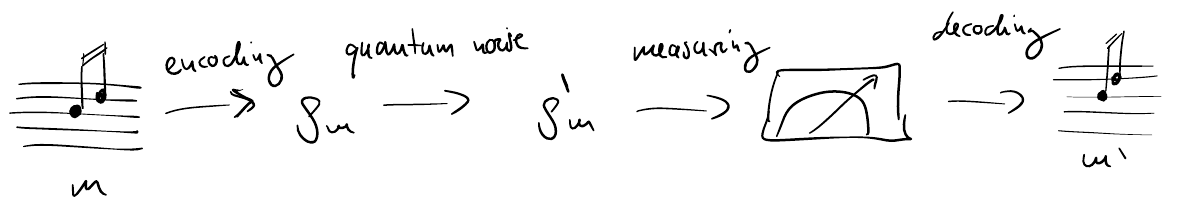}
\end{center}
The difference is that we now encode our information $m$ on a quantum state $\rho_m$, which experiences some noise, such that the measurement returns $m'$ instead of $m$. Similar as in the classical case, it does (in principle) not matter on which physical system (e.g. trapped ions, superconducting circuits, optomechanical systems, photons, etc.) we encode our information. That is, in quantum information theory we study quantum states $\rho_m$ independently on the exact physical system they describe. Thus, it can be seen as a theory about the quantum mechanical state space. Typical questions are then for example as follows: Given some noise, $\rho_m\rightarrow\rho'_{m}$, what is the best way to encode the information $m$ on a quantum state $\rho_m$, such that measuring $\rho'_{m}$ yields with high probability the original information $m$. 

In summary, in quantum information theory we study \emph{the general structure of quantum states} and their potential and limitations for information processing, such as computational tasks, communication, or error correction. This lecture provides you with an introduction into some (but not all) of these topics: 
\begin{center}
\includegraphics[width=\linewidth]{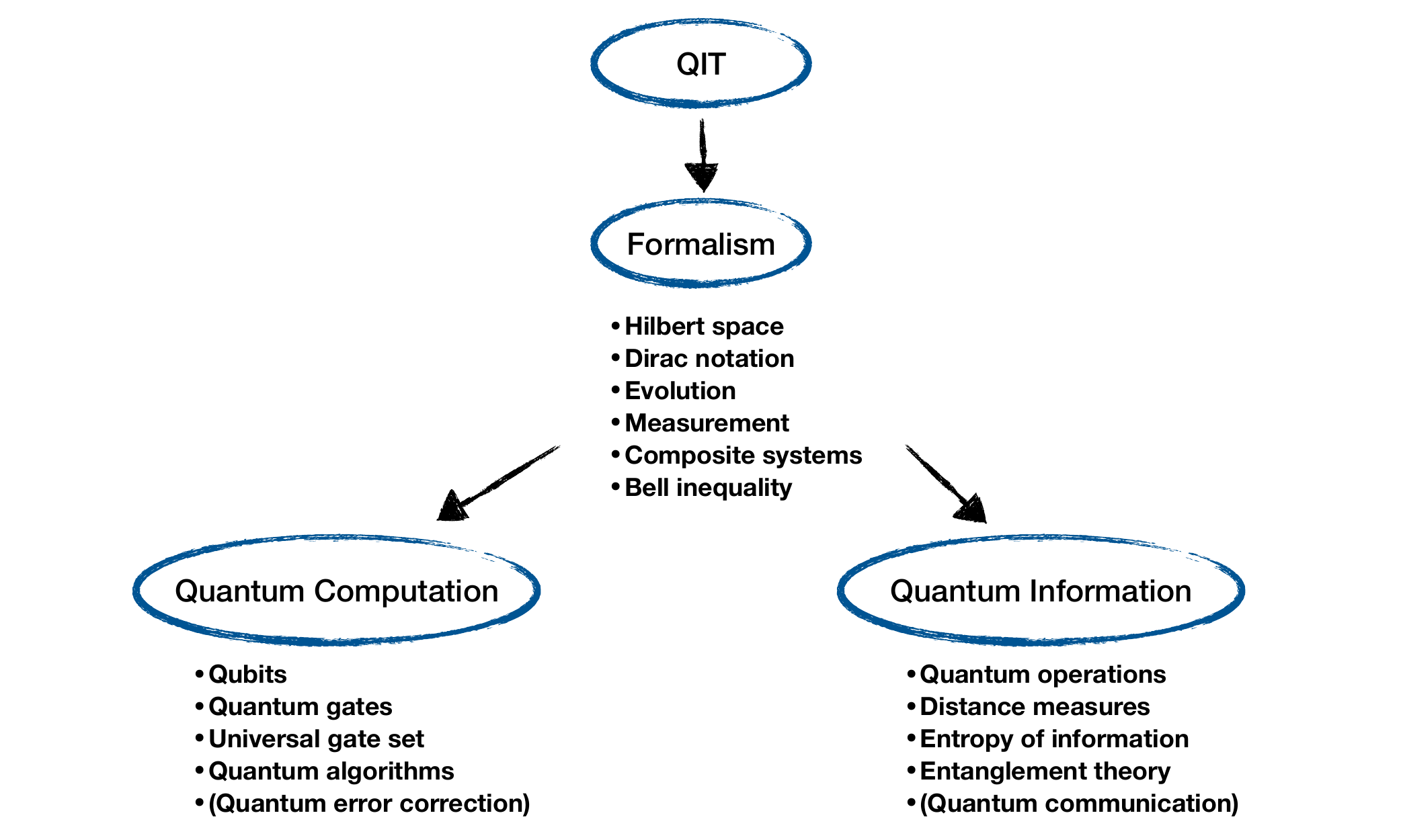}
\end{center}

\chapter{The formalism of quantum information} 

\section{Hilbert space and Dirac notation}

\begin{tcolorbox}
\begin{postulate}
Every isolated quantum system can be associate with a Hilbert space $\mathcal{H}$ (aka state space), which is a complex vector space together with an inner product. The system is completely described by its state vector $\psi\in \mathcal{H}$ (density operator $\rho$ on $\mathcal{H}$).
\end{postulate}
\end{tcolorbox}
In this lecture we restrict to finite dimensional Hilbert spaces $\mathcal{H}=\mathbb{C}^n$. The dual space of $\mathcal{H}$, denoted by $\mathcal{H}^*$, is the set of linear maps $\phi: \mathcal{H}\rightarrow \mathbb{C}$. Since the Hilbert space $\mathcal{H}$ is finite dimensional, the spaces $\mathcal{H}$ and $\mathcal{H}^*$ have the same dimension, we can associate a dual vector $\psi \in \mathcal{H}^*$ with each vector $\psi \in \mathcal{H}$. In particular, the dual vector $\phi \in \mathcal{H}^*$, which is a linear map $\phi: \mathcal{H}\rightarrow \mathbb{C}$, acts as $\psi \mapsto \left( \phi, \psi \right)$, where $\left( \phi, \psi \right)$ denotes the inner product of $\psi,\phi\in\mathcal{H}$. Recall from linear algebra that a function $(.,.) : \mathcal{H}\times \mathcal{H} \rightarrow \mathbb{C}$ is an inner product if, for all $a_j \in \mathbb{C}$ and $\psi,\psi_j,\phi \in \mathcal{H}$,
\begin{enumerate}[i)]
\item $\left( \psi, \sum_j a_j \psi_j\right)=\sum_j a_j \left( \psi,\psi_j \right)$
\item $\left( \psi,\phi\right) = \left( \phi,\psi\right)^*$
\item $\left( \psi,\psi\right)\geq 0$, with equality if and only if $\psi=0$. 
\end{enumerate}
The inner product induces the norm $\norm{\psi}=\sqrt{(\psi,\psi)}$. 

In the following, we mainly use Dirac's notation and write $\psi\in\mathcal{H}$ as $\ket{\psi}$, and its dual vector $\psi\in\mathcal{H}^*$ as $\bra{\psi}$. The inner product of $\ket{\psi},\ket{\phi} \in\mathcal{H}$ is then denoted by $(\phi,\psi)=(\ket{\phi}, \ket{\psi})=\bracket{\phi}{\psi}$, and the norm of $\ket{\psi} \in \mathcal{H}$ becomes $\norm{\psi}=\sqrt{\bracket{\psi}{\psi}}$. Note that the state vector $\ket{\psi} \in \mathcal{H}$ is sometimes also called \underline{state}, \underline{vector}, or simply \underline{ket}.

\begin{definition}
A \underline{basis} of a $n$-dimensional Hilbert space $\mathcal{H}$ is a set of linearly independent state vectors $\{ \ket{v_1},\dots,\ket{v_n}\}$, such that any state vector $\ket{v}\in\mathcal{H}$ can be written as
\begin{align}
\ket{v}=\sum_{j=1}^n a_j \ket{v_j},
\end{align}
where $a_j\in\mathbb{C}$. If, additionally, $\bracket{v_j}{v_k}=\delta_{j,k}$ for all $j,k$, then $\{ \ket{v_1},\dots,\ket{v_n}\}$ is an \underline{orthonormal basis} (ON-basis) of $\mathcal{H}$, and $\norm{v}=1$ if and only if $\sum_{j=1}^n|a_j|^2=1$.
\end{definition}

Given an ON-basis $\{\ket{v_j}\}_j$, we can represent the state vectors $\ket{v_j}\in\mathcal{H}$ by complex valued matrices (i.e. complex vectors), e.g. $\ket{v_1}=(1,0,0,\dots,0)^\top,\ket{v_2}=(0,1,0,\dots,0)^\top,\dots$. Hence, any state vector $\ket{\psi} =\sum_{j=1}^n \psi_j \ket{v_j} $ can be associated with a \ul{matrix representation}
\begin{align}
\ket{\psi}=\begin{pmatrix}\psi_1\\ \vdots \\ \psi_n\end{pmatrix}, \quad \bra{\psi}=\begin{pmatrix}\psi_1^*, \dots, \psi_n^*\end{pmatrix},
\end{align}
such that the inner product is
\begin{align}
\bracket{\phi}{\psi}=(\phi_1^*,\dots,\phi_n^*)\begin{pmatrix}\psi_1\\ \vdots \\ \psi_n\end{pmatrix}=\sum_{j=1}^n \phi^*_j\psi_j,
\end{align}
and the \underline{outer product}
\begin{align}
\ketbra{\psi}{\phi}=\begin{pmatrix}
\psi_1\phi_1^*& \cdots & \psi_1\phi^*_n \\ \vdots& & \vdots \\ \psi_n\phi_1^*& \cdots & \psi_n\phi^*_n
\end{pmatrix}.
\end{align}

\begin{example}
For the Hilbert space $\mathcal{H}=\mathbb{C}^2$ of a qubit (we will soon discuss the qubit in more detail) we choose the orthonormal basis 
\begin{align}
\ket{0}=\ket{v_1}=\begin{pmatrix}1\\0\end{pmatrix}, \quad \ket{1}=\ket{v_2}=\begin{pmatrix}0\\1\end{pmatrix},
\end{align}
such that any normalized vector $\ket{v} \in \mathcal{H}$ can be written as $\ket{v}=a_0\ket{0}+a_1\ket{1}$, with $|a_0|^2+|a_1|^2=1$. Moreover, we have 
\begin{align}
\begin{split}
\ketbra{0}{0}=\begin{pmatrix}1&0\\0&0\end{pmatrix} , \quad \ketbra{0}{1}=\begin{pmatrix}0&1\\0&0\end{pmatrix}\\
\ketbra{1}{0}=\begin{pmatrix}0&0\\1&0\end{pmatrix}\quad \ketbra{1}{1}=\begin{pmatrix}0&0\\0&1\end{pmatrix}.
\end{split}
\end{align}
Note that the vectors
\begin{align}
\ket{+}=\frac{1}{\sqrt{2}} (\ket{0}+\ket{1})=\frac{1}{\sqrt{2}} \begin{pmatrix}1\\1\end{pmatrix},\quad\ket{-}=\frac{1}{\sqrt{2}}(\ket{0}-\ket{1})=\frac{1}{\sqrt{2}} \begin{pmatrix}1\\-1\end{pmatrix},
\end{align}
also form an orthonormal basis of $\mathcal{H}$. The Hilbert space $\mathcal{H}=\mathbb{C}^2$ corresponds, for example, to the polarization of a single photon,  the spin state of a spin-$1/2$ particle, or the occupation of an atom in a double-well potential.
\end{example}

\section{Operators on Hilbert space}
\subsection{Linear operators}
\begin{definition}
A linear operator $A:\mathcal{H}\rightarrow\mathcal{H}$ is defined to be linear in its inputs,
\begin{align}
A\left(\sum_j a_j \ket{v_j}\right) = \sum_j a_j A(\ket{v_j}).
\end{align}
\end{definition}

\begin{remarks}\leavevmode
\begin{enumerate}[1)]

\item Every linear operators has a matrix representation. To see this, suppose that $\{\ket{v_1},\dots,\ket{v_n}\}$ is an ON-basis of $\mathcal{H}$. Then, for $A$ a linear operator there exist numbers $A_{k,j}$, such that $A(\ket{v_j})=\sum_k A_{k,j}\ket{v_k}$. The matrix with entries $A_{k,j}$ is the matrix representation of $A$, and 
\begin{align}
\bra{v_k}A(\ket{v_j})=\bra{v_k}\sum_{k'} A_{k',j} \ket{v_{k'}} = \sum_{k'} A_{k',j} \bracket{v_k}{v_{k'}} = A_{k,j}. 
\end{align}

\item A matrix $A$ with components $A_{k,j}$ is a liner operator, which satisfies $A\left(\sum_j a_j \ket{v_j}\right) = \sum_j a_j A\ket{v_j}$.

\item Given $A$ on $\mathcal{H}$ and $\ket{v}\in\mathcal{H}$, by remarks 1) and 2), we write $A(\ket{v})=A\ket{v}$.

\item Given the entries $A_{k,j}$ and the ON-basis $\{ \ket{v_1},\dots,\ket{v_n}\}$, the linear operator $A$ can be written in the \underline{outer product} representation
\begin{align}\label{eq:outerprod}
A=\sum_{k,j} A_{k,j} \ket{v_k}\bra{v_j},
\end{align}
such that $A \ket{v_j} = \sum_{k,j'} A_{k,j'} \ket{v_k}\bracket{v_{j'}}{v_j}=\sum_k A_{k,j} \ket{v_k}$, as required. 

\item From Eq.~\eqref{eq:outerprod} it follows that for any ON-basis $\{ \ket{v_1},\dots,\ket{v_n}\}$ of $\mathcal{H}$, the identity matrix can be written as
\begin{align}\label{eq:identity}
\unit = \sum_j \ketbra{v_j}{v_j},
\end{align}
which is called the \underline{completeness relation}.
\end{enumerate}
\end{remarks}
In the following we only consider linear operators. 

\begin{example}
Consider the linear operator $X$ on $\mathcal{H}=\mathbb{C}^2$, which acts as $X\ket{0}=\ket{1}$, and $X\ket{1}=\ket{0}$. It can be written as $X=\ketbra{0}{1}+\ketbra{1}{0}$, and, in the basis $\ket{0}=(1,0)^\top, \ket{1}=(0,1)^\top$, it has the matrix representation $X= \begin{pmatrix}0&1\\1&0\end{pmatrix}$. To check this, we calculate
\begin{align}
\begin{split}
X\ket{0}&=\begin{pmatrix}0&1\\1&0\end{pmatrix} \begin{pmatrix}1\\0\end{pmatrix} = \begin{pmatrix}0\\1\end{pmatrix}=\ket{1},\\
X\ket{1}&=\begin{pmatrix}0&1\\1&0\end{pmatrix} \begin{pmatrix}0\\1\end{pmatrix} = \begin{pmatrix}1\\0\end{pmatrix}=\ket{0},
\end{split}
\end{align}
and
\begin{align}
\begin{split}
X\ket{0}&=\left( \ketbra{0}{1}+\ketbra{1}{0}\right)\ket{0} =\ket{0} \bracket{1}{0} + \ket{1} \bracket{0}{0}=\ket{1},\\
X\ket{1}&=\left( \ketbra{0}{1}+\ketbra{1}{0}\right)\ket{1} =\ket{0} \bracket{1}{1} + \ket{1} \bracket{0}{1}=\ket{0}.
\end{split}
\end{align}
\end{example}

\subsection{Adjoints and Hermitian operators}
\begin{definition}
Let $A$ be a linear operator on $\mathcal{H}$. Its \underline{adjoint} (or \underline{Hermitian conjugate}) operator $A^\dagger$ on $\mathcal{H}$ is defined such that for all $\ket{v},\ket{w}\in\mathcal{H}$
\begin{align}\label{eq:adjoint}
\left( \ket{v},A\ket{w} \right)=\left( A^\dagger \ket{v}, \ket{w} \right),
\end{align}
where $(\cdot,\cdot)$ is the inner product.
\end{definition}

From Eq.~\eqref{eq:adjoint} we get $(\ket{v},AB\ket{w})= (\ket{v},A(B\ket{w}))=(A^\dagger\ket{v},B\ket{w})=(B^\dagger A^\dagger\ket{v},\ket{w})$, and, thus
\begin{align}
\left(AB\right)^\dagger = B^\dagger A^\dagger.
\end{align}
By defining $\ket{v}^\dagger=\bra{v}$, we have $(\ketbra{w}{v})^\dagger=\ketbra{v}{w}$ and $(A^\dagger \ket{v})^\dagger = \bra{v}A$. With this, one can identify both sides in Eq.~\eqref{eq:adjoint} with $\bra{v}A\ket{w}$.

\begin{definition}
An operator $A$ on $\mathcal{H}$ with adjoint $A^\dagger=A$ is called \underline{Hermitian} (or \underline{self-adjoint}).\footnote{Note that all self-adjoint operators are Hermitian, but in infinite dimensional Hilbert spaces Hermitian operators are not necessarily self-adjoint (this has to do with the domain of the linear operator and the domain of its adjoint -- for a detailed discussion consult the literature). However, since here we merely consider finite dimensional Hilbert spaces, any Hermitian operator is also self-adjoint, and, thus, we can use the terms Hermitian and self-adjoint interchangeably.}
\end{definition}

\begin{definition}
An operator $A$ on $\mathcal{H}$ satisfying $A^\dagger A = A A^\dagger$ is called \underline{normal}.
\end{definition}

\begin{definition}
An operator $U$ on $\mathcal{H}$ satisfying  $U^\dagger U=UU^\dagger =\unit$ is called \underline{unitary}.
\end{definition}

\begin{definition}
An operator $A$ on $\mathcal{H}$ is called \underline{positive} (or \underline{positive semi-definite}) if $\bra{v}A\ket{v} \geq 0$ for all $\ket{v} \in \mathcal{H}$. If $A$ is positive, we write $A\geq 0$. 
\end{definition}

\begin{definition}
An operator $A$ on $\mathcal{H}$ is called \underline{positive definite} if $\bra{v}A\ket{v} >0$ for all $\ket{v} \in \mathcal{H}$, $\ket{v} \neq 0$.  
\end{definition}

\begin{remarks}\leavevmode
\begin{enumerate}[1)]
\item Hermitian operators have real eigenvalues (Exercise).
\item Unitary operators preserve the inner product, since 
\begin{align}
(U \ket{v}, U\ket{w})=\bra{v}U^\dagger U \ket{w} = \bra{v}\unit\ket{w} = \bracket{v}{w}=( \ket{v}, \ket{w}).
\end{align}

\item Given an operator $A=\sum_{k,j} A_{k,j} \ket{v_k}\bra{v_j}$ in the outer product representation~\eqref{eq:outerprod} with matrix representation $A_{k,j}$, its adjoint  $A^\dagger$ is obtained by transposition and complex conjugation of the matrix representation, i.e. $A^\dagger=\sum_{k,j} A_{j,k}^* \ket{v_k}\bra{v_j}$.
\end{enumerate}
\end{remarks}

An important class of operators are so-called projectors: 

\begin{definition}
The \underline{projector} on a $k$-dimensional subspace $W \subset \mathcal{H}$ with ON-basis $\{\ket{w_1},\dots,\ket{w_k}\}$ is defined as
\begin{align}\label{eq:projector}
P=\sum_{j=1}^k \ketbra{w_j}{w_j}.
\end{align}
\end{definition}

\begin{remarks}\leavevmode
\begin{enumerate}[1)]
\item Projectors are Hermitian, $P^\dagger = \left(\sum_{j=1}^k \ketbra{w_j}{w_j}\right)^\dagger = \sum_{j=1}^k\left( \ketbra{w_j}{w_j}\right)^\dagger=\sum_{j=1}^k \ketbra{w_j}{w_j}=P$.
\item Hence, projectors are normal.
\item Projectors satisfy $P^2=P$, since $P^2= \sum_{j,j'}\ket{w_j}\bracket{w_j}{w_{j'}}\bra{w_{j'}}=\sum_{j=1}^k \ketbra{w_j}{w_j}=P$.
\item Projectors are positive, since $\bra{v}P\ket{v}=\sum_{j=1}^k \bracket{v}{w_j}\bracket{w_j}{v}=\sum_{j=1}^k |\bracket{v}{w_j}|^2\geq 0$. 
\end{enumerate}
\end{remarks}

\subsection{Spectral, polar, and singular value decomposition}
\begin{theorem}[\textbf{Spectral decomposition}]
Any normal operator $A$ on $\mathcal{H}$ is diagonal with respect to some ON-basis of $\mathcal{H}$. Conversely, any unitarily diagonalizable operator is normal.
\end{theorem}
\begin{proof}
See page 72 in \cite{Nielsen-QC-2011}.
\end{proof}

By the spectral decomposition, if $A$ is normal, it can be written as $A=VDV^\dagger$, with $V=\sum_j \ketbra{\lambda_j}{v_j}$ unitary, $D=\sum_j \lambda_j \ketbra{v_j}{v_j}$, and $\{\ket{v_1},\dots,\ket{v_n}\}$ and $\{\ket{\lambda_1},\dots,\ket{\lambda_n}\}$ ON-bases of $\mathcal{H}$. With this, we have
\begin{align}
\begin{split}
A=VDV^\dagger&=\sum_j \ketbra{\lambda_j}{v_j}\sum_k \lambda_k \ketbra{v_k}{v_k}\sum_l \ketbra{v_l}{\lambda_l}\\
&=\sum_{j,k,l} \lambda_k \ket{\lambda_j} \bracket{v_j}{v_k}\bracket{v_k}{v_l} \bra{\lambda_l}\\
&=\sum_k \lambda_k \ketbra{\lambda_k}{\lambda_k}.
\end{split}
\end{align}
This outer product form directly shows that $\lambda_k$ are the eigenvalues of $A$ with corresponding eigenvectors $\ket{\lambda_k}$. It is called the \underline{spectral decomposition} of $A$.

\begin{theorem}[\textbf{Polar decomposition}]
For a linear operator $A$ on $\mathcal{H}$ there exists a unitary $U$, such that
\begin{align}\label{eq:polardecomposition}
A=UJ=KU,
\end{align}
with positive operators $J=\sqrt{A^\dagger A}=|A|$ and $K=\sqrt{AA^\dagger}$. If $A$ is invertible, then $U$ is unique.
\end{theorem}
\begin{proof}
See page 78 in \cite{Nielsen-QC-2011}.
\end{proof}

\begin{theorem}[\textbf{Singular value decomposition}]
Let $A$ be a square matrix. Then there exist unitary matrices $U$ and $V$, and a diagonal matrix $D$ with non-negative entries, such that
\begin{align}
A=UDV.
\end{align}
The diagonal entries of $D$ are called \underline{singular values} of $A$. 
\end{theorem}
\begin{proof}
By the polar decomposition, $A=SJ$, with $S$ unitary and $J$ positive. Now, note that: $J$ is positive $\Rightarrow$ $J$ is Hermitian $\Rightarrow$ $J$ is normal. Hence, $J$ has a spectral decomposition  $J=TDT^\dagger$, with $T$ unitary and $D$ diagonal with non-negative entries. Then $A=STDT^\dagger=UDV$, with $U=ST$ and $V=T^\dagger$. 
\end{proof}

\subsection{Trace of operators}
\begin{definition}
The trace of an operator $A$ on $\mathcal{H}$ is defined as the trace of any matrix representation of $A$. That is, for any ON-basis $\{\ket{v_1},\dots,\ket{v_n}\}$ of $\mathcal{H}$, we have
\begin{align}\label{eq:trace}
\tr{A}=\sum_l \bra{v_l}A\ket{v_l}.
\end{align}
\end{definition}
To see that this definition of the trace of an operator $A$ gives the trace of any matrix representation of $A$, let us write $A$ in the outer product form $A=\sum_{k,j} A_{k,j} \ketbra{w_k}{w_j}$ with the ON-basis $\{\ket{w_1},\dots,\ket{w_n}\}$ of $\mathcal{H}$. By Eq.~\eqref{eq:trace} we then have
\begin{align}
\begin{split}
\tr{A}&=\sum_l \sum_{k,j} A_{k,j} \bracket{v_l}{w_k}\bracket{w_j}{v_l}\\
&=\sum_{k,j} \sum_l A_{k,j} \bracket{w_j}{v_l}\bracket{v_l}{w_k}\\
&=\sum_{k,j}A_{k,j}  \bra{w_j} \left(   \sum_l \ketbra{v_l}{v_l} \right) \ket{w_k}\\
&\overset{\eqref{eq:identity}}{=} \sum_j A_{j,j},
\end{split}
\end{align}
which is the trace of the matrix representation of $A$. Accordingly, all properties of the matrix trace also hold for the trace of linear operators. In particular, for $A,B$ linear operators and $z\in\mathbb{C}$, we have
\begin{subequations}
\begin{align}
\quad \tr{AB}&=\tr{BA} \label{eq:Trprop01}\\
\quad \tr{A+B}&=\tr{A}+\tr{B}\\
\quad \tr{zA}&=z\ \tr{A}.
\end{align}
\end{subequations}
Note that from~\eqref{eq:Trprop01} it follows that the trace is invariant under unitary transformations, i.e., for unitary $U$ we have 
\begin{align}
\tr{U A U^\dagger}=\tr{U^\dagger U A }=\tr{ A}.
\end{align}

\subsection{Density operators}
So far we described the state of a quantum system with associated Hilbert space $\mathcal{H}$ by a state vector $\ket{\psi}\in\mathcal{H}$. However, it can happen that the quantum system is with ``classical'' probability $p_1$ in the state $\ket{\psi_1}$, and with probability $p_2$ in $\ket{\psi_2}$, and so on. That is, we deal with an \underline{ensemble} of pure state $\{p_j,\ket{\psi_j}\}_j$, and we know that the quantum system is with probability $p_j$ in the state $\ket{\psi_j}$. For example, suppose Alice throws a coin, and prepares $\ket{\psi_\mathrm{H}}$ if it shows heads and $\ket{\psi_\mathrm{T}}$ if it shows tails. If she sends us the state, we must describe it by the ensemble $\{(1/2,\ket{\psi_\mathrm{H}}),(1/2,\ket{\psi_\mathrm{T}})\}$. This is \textit{not} a coherent superposition such as the single state vector $(\ket{\psi_\mathrm{H}}+\ket{\psi_\mathrm{T}})/\sqrt{2}$. Indeed, the ensemble \textit{cannot} be described by a single state vector. Instead, we must describe it via an incoherent sum of state vectors $\ket{\psi_j}$, which appear with probability $p_j$. This can be done via the \underline{density operator} (or \underline{density matrix})
\begin{align}
\rho=\sum_j p_j \ketbra{\psi_j}{\psi_j}.
\end{align}

\begin{theorem}[\textbf{Characterization of density operators}]
An operator $\rho$ is the density operator associated to some ensemble $\{p_j,\ket{\psi_j}\}_j$ if and only if it satisfies the following conditions:
\begin{enumerate}[i)]
\item $\rho$ has trace equal to one, i.e., $\tr{\rho}=1$
\item $\rho$ is positive (and, thus, Hermitian), i.e., $\rho\geq 0$.
\end{enumerate}
\end{theorem}
\begin{proof}
Exercise.
\end{proof}

\begin{theorem}[\textbf{Pure density operators}]
The density operator describes a pure state (i.e., a single state vector) if and only if $\tr{\rho^2}=1$.
\end{theorem}
\begin{proof}
\begin{align}
\begin{split}
\tr{\rho^2}&=\tr{\sum_{j,k} p_j p_k \ket{\psi_j}\bracket{\psi_j}{\psi_k}\bra{\psi_k}}\\
&=\sum_{j,k} p_j p_k\tr{\ket{\psi_j}\bracket{\psi_j}{\psi_k}\bra{\psi_k}}  \\
&=\sum_{j,k} p_j p_k \underbrace{|\bracket{\psi_j}{\psi_k}|^2}_{\leq 1} \\
&\leq \sum_{j,k} p_j p_k\\
&=1.
\end{split}
\end{align}
Here the equality holds if and only if $|\bracket{\psi_j}{\psi_k}|=1$ for all $j,k$, which is the case if and only if $\rho$ is pure. 
\end{proof}
\begin{remarks}\leavevmode
\begin{enumerate}[1)]
\item We denote the space of density operators on $\mathcal{H}$ by $\mathcal{D}(\mathcal{H})$.
\item We call $\tr{\rho^2}$ the \underline{purity} of $\rho$. It satisfies $1/\mathrm{dim}(\mathcal{H})\leq \tr{\rho^2} \leq 1$. 
\item If $\rho$ is not pure, i.e., $\tr{\rho^2}\neq 1$, then we say that $\rho$ is \underline{mixed}. 
\item If $\tr{\rho^2}=1/\mathrm{dim}(\mathcal{H})$, we say that $\rho$ is \underline{maximally mixed}.
\end{enumerate}
\end{remarks}

Note that different ensembles can give rise to the same density operator. For example, with $\ket{a}=\sqrt{3/4} \ket{0} + \sqrt{1/4} \ket{1}$ and $\ket{b}=\sqrt{3/4} \ket{0} - \sqrt{1/4} \ket{1}$, the ensemble $\{(1/2,\ket{a}),(1/2,\ket{b})\}$ gives rise to the same density operator as $\{(3/4,\ket{0}),(1/4,\ket{1})\}$, since
\begin{align}
\begin{split}
\rho&=\frac{1}{2} \ketbra{a}{a}+\frac{1}{2} \ketbra{b}{b}\\
&=\frac{1}{2}\left(\frac{3}{4} \ketbra{0}{0} + \frac{\sqrt{3}}{4} \ketbra{0}{1} +\frac{\sqrt{3}}{4} \ketbra{1}{0} +  \frac{1}{4} \ketbra{1}{1} \right) \\
&+ \frac{1}{2}\left(\frac{3}{4} \ketbra{0}{0} - \frac{\sqrt{3}}{4} \ketbra{0}{1} -\frac{\sqrt{3}}{4} \ketbra{1}{0} +  \frac{1}{4} \ketbra{1}{1} \right) \\
&=\frac{3}{4} \ketbra{0}{0} +  \frac{1}{4} \ketbra{1}{1} .
\end{split}
\end{align}

\begin{theorem}[\textbf{Unitary freedom in the ensemble of density operators}]\label{th:unitaryfreedomdensity}
For normalized states $\{\ket{\psi_j}\}_j$ and $\{\ket{\phi_k}\}_k$, and probability distributions $\{p_j\}_j$ and $\{q_k\}_k$, we have $\rho=\sum_j p_j \ketbra{\psi_j}{\psi_j}=\sum_k q_k \ketbra{\phi_k}{\phi_k}$ if and only if 
\begin{align}\label{eq:Ufreedom}
\sqrt{p_j} \ket{\psi_j} = \sum_k U_{j,k} \sqrt{q_k} \ket{\phi_k}
\end{align}
for some unitary $U_{j,k}$. If the ensembles are of different size, we expand the smaller ensemble with entries having probability zero.
\end{theorem}
\begin{proof}
See page 104 in \cite{Nielsen-QC-2011}.
\end{proof}

\section{Evolution and measurement}
\subsection{Evolution of states}\label{sec:EvoOfStates}
\begin{tcolorbox}
\begin{postulate}\label{p:evolution}
The evolution of a closed quantum system is described by a unitary transformation. The state vectors $\ket{\psi}\in\mathcal{H}$ and density operators $\rho \in \mathcal{D}(\mathcal{H})$ transform according to
\begin{align}\label{eq:Uevolve}
\ket{\psi}\rightarrow \ket{\psi'}=U\ket{\psi}, \quad \rho \rightarrow \rho'=U \rho U^\dagger,
\end{align}
with $U$ unitary. 
\end{postulate}
\end{tcolorbox}

In this lecture we usually don't care about how such unitary transformations can be realized through the properties of specific quantum systems. However, for completeness, let us state that if $H$ is the Hamiltonian of the quantum system, then the time evolution is governed by the Schrödinger equation
\begin{align}
\im \hbar \frac{\mathrm{d}}{\mathrm{d}t} \ket{\psi} = H \ket{\psi}
\end{align}
and von Neumann equation 
\begin{align}
\frac{\mathrm{d}}{\mathrm{d}t} \rho = -\frac{\im}{\hbar}[H,\rho],
\end{align}
which give rise to the unitary time evolution operator $U(t,t_0)=\exp(-\im H(t-t_0)/\hbar)$. Since the Hamiltonian $H$ must be Hermitian (its eigenvalues correspond to the eigenenergies, and, hence, must be real), we can use its spectral decomposition $H=\sum_j E_j \ketbra{E_j}{E_j}$, and find $U(t,t_0)=\sum_j \exp(-\im E_j(t-t_0)/\hbar)\ketbra{E_j}{E_j}$.

\subsection{Quantum measurement}

\begin{tcolorbox}
\begin{postulate}\label{p:measurement}
Quantum measurements are described by a collection $\{M_m\}_m$ of \underline{measurement operators} on $\mathcal{H}$, which satisfy the \underline{completeness equation} $\sum_m M_m^\dagger M_m=\unit$. The index $m$ refers to the measurement outcomes that may occur in the experiment. If the state of the system is $\rho$ immediately before the measurement then the probability that result $m$ occurs is given by
\begin{align}
p(m)=\tr{M_m^\dagger M_m \rho},
\end{align}
and the state of the system after the measurement is 
\begin{align}\label{eq:measurementstate}
\frac{M_m \rho M_m^\dagger}{\tr{M_m^\dagger M_m \rho}}.
\end{align}
\end{postulate}
\end{tcolorbox}

\begin{remarks}\leavevmode
\begin{enumerate}[1)]
\item The completeness equation expresses the fact that the probabilities sum to one,
\begin{align}
\sum_m p_m = \sum_m \tr{M_m^\dagger M_m \rho}= \tr{\sum_m M_m^\dagger M_m\rho}=\tr{\rho}=1.
\end{align}
\item If $\rho$ is pure, $\rho=\ketbra{\psi}{\psi}$, then $p(m)=\bra{\psi}M_m^\dagger M_m \ket{\psi}$, and the post-measurement state is $M_m\ket{\psi}/\sqrt{p(m)} $.
\end{enumerate}
\end{remarks}

\begin{example}
The measurement of a two-level system $\mathcal{H}=\mathbb{C}^2$ (qubit) in the computational basis $\{ \ket{0},\ket{1}\}$ is described by the measurement operators $M_0=\ketbra{0}{0}$ and $M_1=\ketbra{1}{1}$. These are projectors [see Eq.~\eqref{eq:projector}] satisfying $M_j^\dagger=M_j, M_j^2=M_j$. The completeness relation holds, since $M_0^\dagger M_0+M_1^\dagger M_1=M_0+M_1=\ketbra{0}{0}+\ketbra{1}{1}=\unit$, where we used Eq.~\eqref{eq:identity} in the last step. 

If the state being measured is $\ket{\psi}=a\ket{0}+b\ket{1}$, with $a,b\in \mathbb{C}$, then $p(0)=\bra{\psi} M_0^\dagger M_0 \ket{\psi}=\bra{\psi}M_0\ket{\psi}= \bracket{\psi}{0}\bracket{0}{\psi}=|\bracket{\psi}{0}|^2=|a|^2$, and, similarly, $p(1)=|b|^2$. If the outcome is $0$, then the state after the measurement is $a\ket{0}/|a|$. Writing $a=|a|\exp(\im \varphi_a)$ yields $a\ket{0}/|a|=\exp(\im \varphi_a)\ket{0}$. That is, up to a negligible global phase, we get the post-measurement state $\ket{0}$. Similarly, if the outcome is $1$, we get $\ket{1}$.
\end{example}

\subsubsection{Projective measurement}

\begin{definition}
A \underline{projective measurement} (or \underline{projective-valued measure} PVM, or \ul{von Neumann measurement}) is described by the \underline{observable} $M$, which is a Hermitian operator on $\mathcal{H}$ with spectral decomposition 
\begin{align}
M=\sum_m m\ P_m.
\end{align}
The projectors $P_m$ [see Eq.~(\ref{eq:projector})] onto the eigenspace of $M$ with eigenvalue $m$ constitute the measurement operators of the projective measurement. 
\end{definition}
The projective measurement is a special case of postulate~\ref{p:measurement}, where the measurement operators are projectors, $M_m=P_m$, satisfying $P_mP_{m'}=\delta_{m,m'} P_m$. The possible outcomes of the measurement correspond to the eigenvalues $m$ of $M$. Measuring the outcome $m$ appears with probability $p(m)=\tr{P_m\rho}$, and results in the post-measurement state $P_m\rho P_m/\tr{P_m\rho}$.

\begin{example}
The observable $\sigma_z=\ketbra{0}{0}-\ketbra{1}{1}$ of a two-level system (qubit) was already discussed above. Here the outcome $m=1$ corresponds to the measurement operator $M_0=\ketbra{0}{0}$, and $m=-1$ to $M_1=\ketbra{1}{1}$.
\end{example}

\subsubsection{POVM measurement}\label{sec:POVM}
If the post-measurement state is not of interest (e.g. the quantum system gets destroyed or the experiment is concluded after a single measurement) we can use the particularly useful formalism of \underline{positive-operator-valued measures} (POVMs). To this end, consider postulate~\ref{p:measurement} and define the positive operator $E_m=M_m^\dagger M_m$. Then we have a set of positive operators $\{E_m\}_m$ such that $\sum_m E_m=\unit$ and $p(m)=\tr{E_m \rho}$. 

\begin{definition}
A POVM is a set of positive operators $\{E_m\}_m$ (aka \ul{POVM elements}), such that $\sum_m E_m=\unit$. 
\end{definition}

\begin{remarks}\leavevmode
\begin{enumerate}[1)]
\item For every POVM $\{E_m\}_m$ there exists a set of measurement operators $\{M_m\}_m$. \begin{proof} Define $M_m=\sqrt{E_m}$, then $M_m^\dagger=M_m$, and we have $\sum_m M_m^\dagger M_m = \sum_m E_m=\unit$. \end{proof}
\item Projective measurements are POVM measurements whose POVM elements $E_m$ are projectors.
\item POVMs are the most general kind of measurements in quantum mechanics. 
\end{enumerate}
\end{remarks}

\begin{theorem}[\textbf{Neumark's theorem}]
Any POVM can be realized by extending the Hilbert space to a larger space, and performing a projective measurement.
\end{theorem}
\begin{proof}
Exercise. 
\end{proof}

\section{Composite systems}
Suppose we have $n\geq 2$ quantum systems and we want to describe their combined system.

\begin{tcolorbox}
\begin{postulate}
The state space $\mathcal{H}$ of a composite physical system is the tensor product $\mathcal{H}_1 \otimes \mathcal{H}_2 \otimes \dots \otimes \mathcal{H}_n$ of the state spaces $\mathcal{H}_j$ of the component physical systems. If system $j$ is separately prepared in state $\rho_j$, then the joint state of the total system is $\rho_1 \otimes \rho_2 \otimes \dots \otimes \rho_n$. 
\end{postulate}
\end{tcolorbox}

\subsection{Tensor product}

\begin{definition}
The \underline{tensor product} $\otimes: \mathcal{H}_1 \times \mathcal{H}_2 \rightarrow \mathcal{H}_1 \otimes \mathcal{H}_2$ maps $\ket{v}\in \mathcal{H}_1$ and $\ket{w}\in \mathcal{H}_2$ as $(\ket{v},\ket{w}) \mapsto \ket{v}\otimes\ket{w}$, where $\mathrm{dim}(\mathcal{H}_1 \otimes \mathcal{H}_2)=\dim(\mathcal{H}_1)\dim(\mathcal{H}_2)$. For all $\ket{v},\ket{v_1},\ket{v_2}\in \mathcal{H}_1$, $\ket{w},\ket{w_1},\ket{w_2}\in \mathcal{H}_2$, and $z\in\mathbb{C}$, it satisfies
\begin{subequations}
\begin{align}
z(\ket{v}\otimes \ket{w}) &= (z\ket{v})\otimes\ket{w}=\ket{v}\otimes(z\ket{w})=z\ket{v}\otimes\ket{w}\\
(\ket{v_1}+\ket{v_2})\otimes\ket{w} &= \ket{v_1}\otimes\ket{w}+\ket{v_2}\otimes\ket{w}\\
\ket{v} \otimes(\ket{w_1}+\ket{w_2})&= \ket{v} \otimes\ket{w_1}+\ket{v} \otimes\ket{w_2}
\end{align}
\end{subequations}
\end{definition}

\begin{definition}
For linear operators $A$ on $\mathcal{H}_1$ and $B$ on $\mathcal{H}_2$ we have
\begin{align}\label{eq:opacttensor}
(A\otimes B)\left(\sum_j a_j \ket{v_j} \otimes \ket{w_j}  \right) = \sum_j a_j A\ket{v_j} \otimes B\ket{w_j}
\end{align}
for all $\ket{v_j}\in\mathcal{H}_1$, $\ket{w_j}\in\mathcal{H}_2$, and $a_j\in\mathbb{C}$. 
\end{definition}

\begin{definition}
The inner product on $\mathcal{H}=\mathcal{H}_1\otimes\mathcal{H}_2$ is defined by
\begin{align}
\begin{split}
\left( \sum_j a_j \ket{v_j}\otimes\ket{w_j} , \sum_k b_k \ket{v_k'} \otimes \ket{w_k'}\right)&=\sum_{j,k} a_j^* b_k (\ket{v_j},\ket{v_k'})\  (\ket{w_j},\ket{w_k'})\\
&=\sum_{j,k} a_j^* b_k \bracket{v_j}{v_k'} \bracket{w_j}{w_k'}
\end{split}
\end{align}
for all $\ket{v_j},\ket{v_k'}\in\mathcal{H}_1$, $\ket{w_j},\ket{w_k'}\in\mathcal{H}_2$, and $a_j,b_k\in\mathbb{C}$.
\end{definition}

\begin{remarks}\leavevmode
\begin{enumerate}[1)]
\item If $\mathrm{dim}(\mathcal{H}_1)=d_1$ and $\mathrm{dim}(\mathcal{H}_2)=d_2$ then $\mathrm{dim}(\mathcal{H}_1 \otimes \mathcal{H}_2)=d_1d_2$. For example, for $n$ two level systems (qubits), we have $\mathcal{H}=(\mathbb{C}^2)^{\otimes n}$, and, hence, $\mathrm{dim}(\mathcal{H})=\mathrm{dim}(\mathbb{C}^2 \otimes\mathbb{C}^2 \otimes \dots \otimes \mathbb{C}^2 )=2^n$, which grows exponentially in $n$. 
\item The elements of $\mathcal{H}_1\otimes \mathcal{H}_2$ are linear combinations of tensor products $\ket{v}\otimes\ket{w}$, where $\ket{v}\in\mathcal{H}_1$ and $\ket{w}\in\mathcal{H}_2$.
\item If $\{\ket{v_j}\}_j$ is an ON-basis of $\mathcal{H}_1$ and $\{\ket{w_k}\}_k$ is an ON-basis of $\mathcal{H}_2$ then $\{ \ket{v_j}\otimes \ket{w_k}\}_{j,k}$ is an ON-basis of $\mathcal{H}_1 \otimes \mathcal{H}_2$.
\item We often use the abbreviation $\ket{v}\otimes\ket{w}=\ket{v,w}=\ket{vw}$.
\item An arbitrary linear operator $C$ on $\mathcal{H}_1 \otimes \mathcal{H}_2$ can be represented as a linear combination of tensor products of linear operators $A_j$ on $\mathcal{H}_1$ and $B_j$ on $\mathcal{H}_2$, 
\begin{align}\label{eq:OponH1H2}
C=\sum_j c_j A_j \otimes B_j,
\end{align}
with $c_j\in\mathbb{C}$. By Eq.~\eqref{eq:opacttensor} we then have
\begin{align}
\left( \sum_j c_j A_j\otimes B_j \right) \ket{v} \otimes \ket{w} = \sum_j c_j A_j \ket{v}\otimes B_j\ket{w}.
\end{align}

\item Given an ON-basis $\{\ket{v_j}\}_j$ of $\mathcal{H}_1$ and $\{\ket{w_k}\}_k$ of $\mathcal{H}_2$, we can write an arbitrary operator $A$ on $\mathcal{H}_1 \otimes \mathcal{H}_2$ as
\begin{align}
A=\sum_{j,j',k,k'} A_{k,k'}^{j,j'} \ketbra{v_j}{v_{j'}} \otimes \ketbra{w_k}{w_{k'}},
\end{align}
where $A_{k,k'}^{j,j'} \in\mathbb{C}$. 
\end{enumerate}

\end{remarks}

\subsection{Matrix representation}
In the matrix representation, the tensor product $A\otimes B$ between linear operators $A$ on $\mathcal{H}_1$ and $B$ on $\mathcal{H}_2$ can be calculated via the \underline{Kronecker product} of the matrix representation of $A$ and $B$. Suppose $A$ is represented by a $d_1\times d_1$ matrix and $B$ by a $d_2\times d_2$ matrix, then
\begin{align}
A\otimes B  =\begin{pmatrix}
A_{1,1}B& A_{1,2}B &\cdots & A_{1,d_1}B \\
A_{2,1}B&A_{2,2}B&\cdots &A_{2,d_1}B \\
\vdots&\vdots & \ddots& \vdots\\
A_{d_1,1}B & A_{d_1,2}B &\hdots & A_{d_1,d_1}B
\end{pmatrix},
\end{align}
where
\begin{align}
A_{j,k}B=\begin{pmatrix}
A_{j,k}B_{1,1} &\hdots& A_{j,k} B_{1,d_2}\\
\vdots & \ddots & \vdots \\
A_{j,k}B_{d_2,1}& \hdots & A_{j,k}B_{d_2,d_2}
\end{pmatrix}.
\end{align}
Hence, $A\otimes B$ is a $d_1d_2\times d_1d_2$ matrix. 

The Kronecker product of two vectors works similar. For $\ket{\psi}=(\psi_1,\dots,\psi_{d_1})^\top \in \mathcal{H}_1$ and $\ket{\phi}=(\phi_1,\dots,\phi_{d_2})^\top\in \mathcal{H}_2$, we have
\begin{align}
\begin{split}
\ket{\psi} \otimes \ket{\phi}&=\begin{pmatrix}
\psi_1 \begin{pmatrix} \phi_1\\ \vdots\\ \phi_{d_2}\end{pmatrix}\\
\vdots\\
\psi_{d_1} \begin{pmatrix} \phi_1\\ \vdots\\ \phi_{d_2}\end{pmatrix}
\end{pmatrix}=(\psi_1\phi_1,\dots,\psi_1\phi_{d_2},\psi_2\phi_1,\dots,\psi_2\phi_{d_2},\dots,\psi_{d_1}\phi_1,\dots,\psi_{d_1}\phi_{d_2})^\top.
\end{split}
\end{align}
\begin{example}
The Kronecker product of $\ket{\psi}=(1,2)^\top/\sqrt{5}$ and $\ket{\phi}=(3,4)^\top/5$ is
\begin{align}
\ket{\psi} \otimes \ket{\phi}=\frac{1}{5\sqrt{5}}\begin{pmatrix}
1 \begin{pmatrix} 3\\4\end{pmatrix} \\
2 \begin{pmatrix} 3\\4\end{pmatrix} \\
\end{pmatrix}
=\frac{1}{5\sqrt{5}}\begin{pmatrix}3\\4\\6\\8\end{pmatrix}.
\end{align}
\end{example}

\subsection{Partial trace}
Suppose we have a state $\rho_\mathrm{AB}$ of a composite quantum system with state space $\mathcal{H}=\mathcal{H}_\mathrm{A}\otimes \mathcal{H}_\mathrm{B}$. Further suppose we have no access to system B, and we want to describe the state and measurement statistics of system A alone. That is, we want to describe the \underline{reduced density operator} $\rho_\mathrm{A}$ of system A. This is obtained from $\rho_\mathrm{AB}$ by taking the partial trace over system B, 
\begin{align}
\rho_\mathrm{A}=\trp{B}{\rho_\mathrm{AB}}.
\end{align}
\begin{definition}
The \ul{partial trace} is linear in its inputs and defined by 
\begin{align}
\trp{B}{\ketbra{v_1}{v_2}\otimes \ketbra{w_1}{w_2}} =\tr{\ketbra{w_1}{w_2}} \ \ketbra{v_1}{v_2} =\bracket{w_2}{w_1} \ \ketbra{v_1}{v_2},
\end{align}
where $\ket{v_1},\ket{v_2} \in \mathcal{H}_\mathrm{A}$, and $\ket{w_1},\ket{w_2} \in \mathcal{H}_\mathrm{B}$.
\end{definition}

\begin{remarks}\leavevmode
\begin{enumerate}[1)]
\item It's not obvious but this definition provides the correct measurement statistics for measurements made on system A. 
\item To be precise, the partial trace is the unique function with the property that $\tr{M \rho_\mathrm{A}}=\tr{(M\otimes \unit) \rho_\mathrm{AB}}$ for all linear operators $M$ (for a proof, see page 107 in \cite{Nielsen-QC-2011}).
\end{enumerate}
\end{remarks}

\begin{example}
Consider a product state (i.e., an uncorrelated state) $\rho_\mathrm{AB}=\rho\otimes \sigma \in \mathcal{D}(\mathcal{H}_\mathrm{A}\otimes \mathcal{H}_\mathrm{B})$. Then $\trp{B}{\rho_\mathrm{AB}} = \trp{B}{\rho\otimes \sigma} =\tr{\sigma} \rho = \rho$, as intuitively expected. 
\end{example}
\begin{example}
Consider a correlated state $\rho_\mathrm{AB}=\ketbra{\psi}{\psi}$, with $\ket{\psi}=(\ket{00}+\ket{11})/\sqrt{2}$, where $\ket{\psi} \in \mathcal{H}=\mathbb{C}^2 \otimes \mathbb{C}^2$. Then
\begin{align}
\begin{split}
\rho_\mathrm{AB}&=\frac{1}{\sqrt{2}} \left(\ket{00}+\ket{11} \right) \frac{1}{\sqrt{2}} \left(\bra{00}+\bra{11} \right)\\
&=\frac{1}{2}\left( \ketbra{00}{00}+\ketbra{00}{11}+\ketbra{11}{00}+\ketbra{11}{11} \right) .
\end{split}
\end{align}
Taking the partial trace over system B yields
\begin{align}
\begin{split}
\rho_\mathrm{A}&=\trp{B}{\rho_\mathrm{AB}}\\
&=\frac{1}{2}\left(\bracket{0}{0} \ \ketbra{0}{0} + \bracket{1}{0} \ \ketbra{0}{1} +\bracket{0}{1}\  \ketbra{1}{0} +\bracket{1}{1}\  \ketbra{1}{1} \right)\\
&=\frac{1}{2}\left(\ketbra{0}{0} + \ketbra{1}{1} \right)\\
&=\frac{1}{2}\unit.  
\end{split}
\end{align}
This state is maximally mixed since $\tr{\rho_\mathrm{A}^2} = 1/\mathrm{dim}(\mathbb{C}^2)=1/2$. Further, observe that
\begin{enumerate}[i)]
\item the state of the joint system is pure, and, thus, exactly known, while
\item the state of system A is maximally mixed, and, thus, completely undetermined.
\end{enumerate}
As a result, a measurement of $\{\ketbra{0}{0},\ketbra{1}{1}\}$ of system A is completely undetermined, while a subsequent measurement of $\{\ketbra{0}{0},\ketbra{1}{1}\}$ of system B is strictly correlated with the initial random outcome of the measurement of $A$. This holds for any orthogonal measurement basis, not only for $\{\ketbra{0}{0},\ketbra{1}{1}\}$ and is a manifestation of the state's entanglement.
\end{example}

\subsection{Bipartite pure-state entanglement}\label{sec:bipureent}
Composite quantum systems give rise to a unique quantum phenomenon: entanglement. In Ch.~\ref{ch:qinfo} we discuss entanglement in more detail. For the moment we restrict to a definition under which condition bipartite pure states must be entangled. 

\begin{definition}
A bipartite pure state $\ket{\psi}\in\mathcal{H}_\mathrm{A}\otimes \mathcal{H}_\mathrm{B}$ is called \underline{entangled} if it cannot be written as a product state $\ket{\psi_\mathrm{A}} \otimes \ket{\psi_\mathrm{B}}$ for any choices of states $\ket{\psi_\mathrm{A}} \in \mathcal{H}_\mathrm{A}$ and $\ket{\psi_\mathrm{B}} \in \mathcal{H}_\mathrm{B}$.
\end{definition}

\begin{example}
Consider the bipartition $\mathcal{H}=\mathbb{C}^2 \otimes \mathbb{C}^2$. The state $\ket{\phi}=(\ket{00}-\ket{01}+\ket{10}-\ket{11})/2$ is not entangled, since it can be written as a product state $\ket{\phi}=(\ket{0}+\ket{1})/\sqrt{2} \otimes (\ket{0}-\ket{1})/\sqrt{2} =\ket{+}\otimes \ket{-}$. On the other hand, the state $(\ket{00}+\ket{11})/\sqrt{2}$ [see the previous example] cannot be written as a product state, and, hence, is entangled. Indeed, it is one of the \underline{Bell states}
\begin{align}\label{eq:Bellstates}
\begin{split}
\ket{\psi^+}&=\frac{1}{\sqrt{2}}\left( \ket{01}+\ket{10}\right),  \quad \ket{\phi^+}=\frac{1}{\sqrt{2}}\left( \ket{00}+\ket{11}\right),\\
\ket{\psi^-}&=\frac{1}{\sqrt{2}}\left( \ket{01}-\ket{10}\right),  \quad \ket{\phi^-}=\frac{1}{\sqrt{2}}\left( \ket{00}-\ket{11}\right),
\end{split}
\end{align}
which are maximally entangled states in $\mathcal{H}=\mathbb{C}^2 \otimes \mathbb{C}^2$.
\end{example}

How can we find out whether a state is entangled or not?

\subsection{Schmidt decomposition}
\begin{theorem}[\textbf{Schmidt decomposition}]
Let $\ket{\psi}\in\mathcal{H}_\mathrm{A}\otimes \mathcal{H}_\mathrm{B}$ be a bipartite pure state, where $\mathrm{dim}(\mathcal{H}_\mathrm{A})=d_\mathrm{A}$ and $\mathrm{dim}(\mathcal{H}_\mathrm{B})=d_\mathrm{B}$. Then there exist ON-bases $\{\ket{v_j}\}_j$ of $\mathcal{H}_\mathrm{A}$ and $\{\ket{w_j}\}_j$ of $\mathcal{H}_\mathrm{B}$ such that
\begin{align}\label{eq:Schmidt}
\ket{\psi}= \sum_{j=1}^{d} \sqrt{\lambda_j} \ket{v_j} \otimes \ket{w_j},
\end{align}
where $\sqrt{\lambda_j}>0$ are strictly positive real numbers satisfying $\sum_j\lambda_j=1$, called \underline{Schmidt coefficients}, and $d\leq \min\{d_\mathrm{A},d_\mathrm{B}\}$ is the \underline{Schmidt rank} (or \underline{Schmidt number}), i.e., the number of non-vanishing Schmidt coefficients. 
\end{theorem}
\begin{proof}
Exercise. 
\end{proof}

\begin{remarks}\leavevmode
\begin{enumerate}[1)]
\item Given any ON-basis $\{\ket{x_j}\}_j$ of $\mathcal{H}_\mathrm{A}$ and $\{\ket{y_k}\}_k$ of $\mathcal{H}_\mathrm{B}$, we express $\ket{\psi}\in\mathcal{H}_\mathrm{A}\otimes \mathcal{H}_\mathrm{B}$ in general as $\ket{\psi}=\sum_{j,k} c_{j,k} \ket{x_j}\otimes \ket{y_k}$, where $c_{j,k}\in \mathbb{C}$. That is, the sum runs over two indizes, $j=1,\dots, \mathrm{dim}(\mathcal{H}_\mathrm{A})$ and $k=1,\dots, \mathrm{dim}(\mathcal{H}_\mathrm{B})$. The beauty of the Schmidt form in Eq.~\eqref{eq:Schmidt} is that the sum only runs over a single index $j=1,\dots,\min\{d_\mathrm{A},d_\mathrm{B}\}$. 

\item The squared Schmidt coefficients $\lambda_j$ are the non-vanishing eigenvalues of the reduced state $\rho_\mathrm{A}$ as well as of $\rho_\mathrm{B}$.
\begin{proof}
\begin{align}
\begin{split}
\rho_\mathrm{A}&=\trp{B}{\ketbra{\psi}{\psi}}\\
&=\sum_{j,j'} \sqrt{\lambda_j \lambda_{j'}} \trp{B}{\ketbra{v_j}{v_{j'}} \otimes \ketbra{w_j}{w_{j'}}} \\
&=\sum_{j,j'} \sqrt{\lambda_j \lambda_{j'}} \tr{\ketbra{w_j}{w_{j'}}} \ketbra{v_j}{v_j'}\\
&=\sum_j \lambda_j \ketbra{v_j}{v_j},
\end{split}
\end{align}
which is the spectral decomposition of $\rho_\mathrm{A}$. Similarly, we get 
\begin{align}\label{eq:Schmidteigenvaluesrhoa}
\rho_\mathrm{A}&=\sum_j \lambda_j \ketbra{w_j}{w_j}.
\end{align}
\end{proof}

\item A state $\ket{\psi}\in\mathcal{H}_\mathrm{A}\otimes \mathcal{H}_\mathrm{B}$ is a product state if and only if it has Schmidt rank $d=1$. 
\begin{proof}
Exercise
\end{proof}

\item Hence, a state $\ket{\psi}\in\mathcal{H}_\mathrm{A}\otimes \mathcal{H}_\mathrm{B}$ is entangled if and only if it has Schmidt rank $d>1$. 

\item The Schmidt decomposition cannot be generalized to more than two parties. 
\end{enumerate}
\end{remarks}

\subsection{Purification}
Given a mixed state $\rho_\mathrm{A}\in\mathcal{D}(\mathcal{H}_\mathrm{A})$, we can introduce a fictitious reference system $\mathcal{H}_\mathrm{R}$ to get a pure state $\ket{\psi} \in \mathcal{H}_\mathrm{A} \otimes \mathcal{H}_\mathrm{R}$, such that $\rho_\mathrm{A}=\trp{R}{\ketbra{\psi}{\psi}}$. This technique is called purification and is a purely mathematical procedure, which allows us to associate pure states with mixed states. 

\begin{definition}
A \underline{purification} of a density operator $\rho_\mathrm{A} \in \mathcal{D}(\mathcal{H}_\mathrm{A})$ is a pure bipartite state $\ket{\psi}\in \mathcal{H}_\mathrm{A} \otimes \mathcal{H}_\mathrm{R}$ on the composite system of $\mathcal{H}_\mathrm{A}$ and a reference system $\mathcal{H}_\mathrm{R}$, with the property that 
\begin{align}\label{eq:purification}
\rho_\mathrm{A}=\trp{R}{\ketbra{\psi}{\psi}}.
\end{align}
\end{definition}

\begin{remarks}\leavevmode
\begin{enumerate}[1)]
\item If $\rho_\mathrm{A}$ has the spectral decomposition $\rho_\mathrm{A}=\sum_j p_j \ketbra{p_j}{p_j}$, and $\{\ket{r_j}\}_j$ is a set of orthonormal vectors of the reference system, then
\begin{align}\label{eq:purifiedstate}
\ket{\psi}=\sum_j \sqrt{p_j} \ket{p_j}\otimes \ket{r_j}
\end{align}
is a purification of $\rho_\mathrm{A}$. 
\begin{proof}
\begin{align}
\trp{R}{\ketbra{\psi}{\psi}}&= \sum_{j,j'} \sqrt{p_j p_{j'}}\ \tr{\ketbra{r_j}{r_{j'}}} \ketbra{p_j}{p_{j'}}\\
&=\sum_j p_j \ketbra{p_j}{p_j}\\
&=\rho_\mathrm{A}
\end{align}
\end{proof}

\item Note the relation between the Schmidt decomposition and purification by Eqs.~\eqref{eq:Schmidt} and~\eqref{eq:purifiedstate}.

\end{enumerate}
\end{remarks}

\section{Bell inequality}

\subsection{History}
Consider the Bell state $\ket{\phi^+}=(\ket{00}+\ket{11})/\sqrt{2} \in \mathcal{H}_\mathrm{A} \otimes \mathcal{H}_\mathrm{B}$ from Eq.~\eqref{eq:Bellstates}. Since the system $A$ can be locally separated in spacetime from system $B$, the entanglement of $\ket{\phi^+}$ describes correlations between locally separated objects. Amongst others, Einstein, Podolsky, and Rosen (EPR) argued that such correlations must be unphysical. In particular, they argued that quantum mechanics is incomplete and \underline{hidden variables} have to be added in order to explain entanglement. In the sense of EPR [see Einstein, Podolsky, and Rosen, \textit{Phys. Rev.} \textbf{47}, 777 (1935)]:
\begin{description}
\item[Reality.] If we can, without disturbing a physical system, predict the value of a physical quantity with certainty, then there exists an element of reality associated with this physical quantity. 

\item[Locality.] Asserting the physical process occurring at one place should have no immediate effect on the element of reality at another location.

\item[Completeness.] A complete physical theory must have a theoretical object for each element of physical reality.
\end{description}

In 1964 John Bell derived an inequality for correlation measurements, which showed that the results for entangled states, which are predicted by quantum mechanics, could not be reproduced by a local realistic theory based on hidden variables. In 1969 Clauser, Horne, Shimony, and Holt (CHSH) presented a similar inequality, which is well suited to be tested experimentally. We discuss the latter inequality in the following.

\subsection{The CHSH Bell inequality}

\subsubsection{Local realistic interpretation of entanglement with hidden variables}
Suppose that two particles are entangled in a property for which a measurement of a single particle yields $+1$ or $-1$. 
\begin{center}
\includegraphics[width=0.7\linewidth]{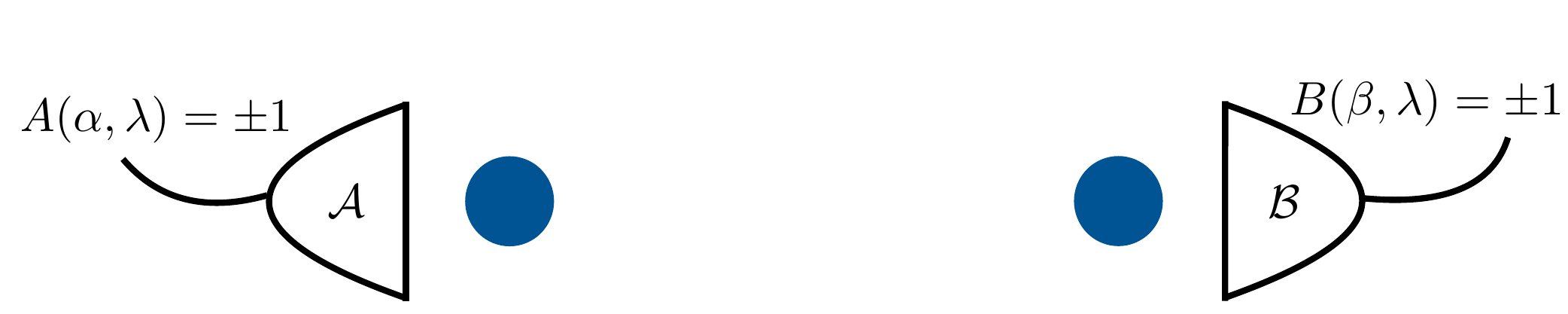}
\end{center}
In order to describe this scenario of entangled particles by means of a local realistic theory, we introduce local hidden variables. To this end, suppose that $\lambda$ is a single or a set of local hidden variables, and suppose that $\lambda$ is distributed according to a probability distribution 
\begin{align}
\rho(\lambda) \geq 0,\quad \int \mathrm{d}\lambda\ \rho(\lambda)=1.
\end{align}
For a theory to be local, measurements of the first particle at detector $\mathcal{A}$ must be completely determined by the setting $\alpha$ of $\mathcal{A}$ and the local hidden variables $\lambda$ (where $\alpha$ can be set independently on $\lambda$). Similar, measurements of the second particle are completely determined by the setting $\beta$ of detector $\mathcal{B}$ and by $\lambda$. The measurement results can be $A(\alpha,\lambda)=\pm 1$, $B(\beta,\lambda)=\pm 1$. 

The probability that $A=x$, where $x=\pm 1$, is 
\begin{align}
p_{A=x} = \int \mathrm{d}\lambda\ \rho(\lambda)\frac{1+x A(\alpha,\lambda)}{2},
\end{align}
and that $(A,B)=(x,y)$, where $x,y=\pm 1$, is
\begin{align}\label{eq:Bellpxy}
p_{xy} = \int \mathrm{d}\lambda\ \rho(\lambda)\frac{1+x A(\alpha,\lambda)}{2} \ \frac{1+y B(\beta,\lambda)}{2}.
\end{align}
Now, let us define the particle correlation measure 
\begin{align}\label{eq:Bellcorrelator}
E(\alpha,\beta)=p_{++}+p_{--}-p_{+-}-p_{-+},
\end{align}
which satisfies $-1 \leq E(\alpha,\beta) \leq 1$. The lower bound holds if the detectors always disagree, and the upper bound if they always agree. Plugging~\eqref{eq:Bellpxy} into~\eqref{eq:Bellcorrelator} and using the shorthand $A(\alpha,\lambda)\equiv A$, $B(\beta,\lambda)\equiv B$, yields
\begin{align}\label{eq:BellEint}
\begin{split}
E(\alpha,\beta)&=\int \mathrm{d}\lambda\ \rho(\lambda) \left( \frac{1+A}{2}\ \frac{1+B}{2}+ \frac{1-A}{2}\ \frac{1-B}{2} - \frac{1+A}{2}\ \frac{1-B}{2}- \frac{1-A}{2}\ \frac{1+B}{2}\right)\\
&=\int \mathrm{d}\lambda\ \rho(\lambda) \frac{1}{4} \left(1+A+B+AB+1-A-B+AB-1-A+B+AB-1-B+A+AB \right)\\
&=\int \mathrm{d}\lambda\ \rho(\lambda) A(\alpha,\lambda)B(\beta,\lambda).
\end{split}
\end{align}
Next, consider the detector settings $\alpha,\alpha',\beta,\beta'$, and define the quantity
\begin{align}\label{eq:BellShv}
S(\alpha,\alpha',\beta,\beta')=A(\alpha,\lambda)\left[ B(\beta,\lambda)-B(\beta',\lambda) \right] + A(\alpha',\lambda)\left[ B(\beta,\lambda)+B(\beta',\lambda) \right].
\end{align}
Note that, by definition, $S \in\{-2,2\}$. Using Eqs.~\eqref{eq:BellEint} and~\eqref{eq:BellShv}, the expectation value of $S$ becomes
\begin{align}
\begin{split}
\braket{S} &= \int \mathrm{d}\lambda\ \rho(\lambda)\  S(\alpha,\alpha',\beta,\beta')\\
&=E(\alpha,\beta)-E(\alpha,\beta')+E(\alpha',\beta)+E(\alpha',\beta'),
\end{split}
\end{align}
and we have
\begin{align}
\begin{split}
\abs{\braket{S}} &= \abs{\int \mathrm{d}\lambda\ \rho(\lambda)\ S(\alpha,\alpha',\beta,\beta')} \\
&\leq \int \mathrm{d}\lambda\ \rho(\lambda) \underbrace{\abs{S(\alpha,\alpha',\beta,\beta')}}_{= 2} \\
& \leq 2 \int \mathrm{d}\lambda\ \rho(\lambda) \\
&=2.
\end{split}
\end{align}
That is, we found that
\begin{align}\label{eq:BellSex}
\abs{\braket{S}} =\abs{E(\alpha,\beta)-E(\alpha,\beta')+E(\alpha',\beta)+E(\alpha',\beta')} \leq 2
\end{align}
holds for any local realistic theory. 

\subsubsection{Quantum mechanical description}
We now describe the situation of the two entangled particles by assuming that quantum mechanics is complete, i.e., that for the description of entanglement no additional hidden variables are needed. Recall that we assumed that the particles are entangled in a property for which a measurement of a single particle yields $+1$ or $-1$. This applies to the Bell states~\eqref{eq:Bellstates}. Hence, we can assume that the particles are in the Bell state $\ket{\psi^-}=(\ket{01}-\ket{10})/\sqrt{2} \in \mathcal{H}=\mathbb{C}^2 \otimes \mathbb{C}^2$. 
\begin{center}
\includegraphics[width=0.6\linewidth]{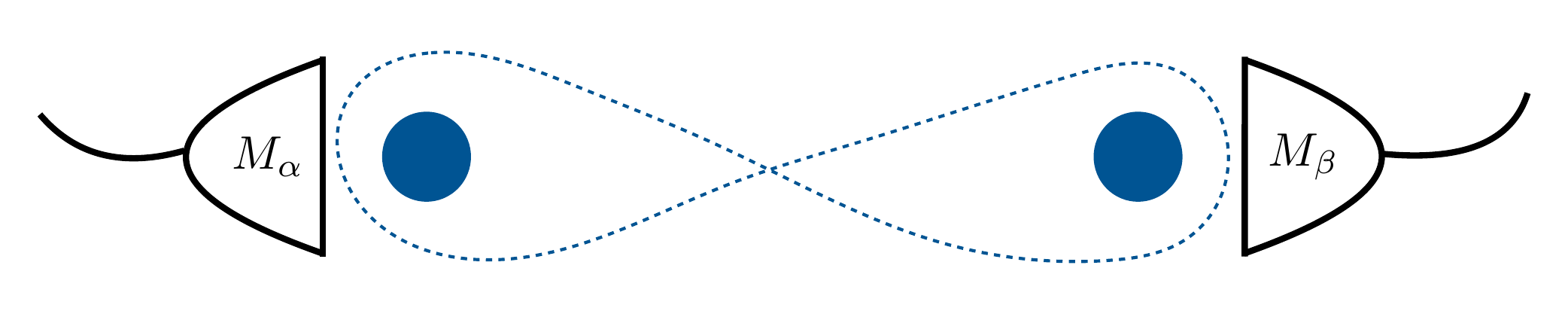}
\end{center}
For the measurement of each particle consider the projective measurement according to the observable $M_\alpha=P_\alpha-P_{\alpha+\pi/2}$. The projectors are
\begin{align}
P_\alpha= \ketbra{\alpha}{\alpha},\quad \ket{\alpha}=\cos(\alpha) \ket{1} -\sin(\alpha)\ket{0},
\end{align}
and satisfy the completeness equation $P_\alpha+P_{\alpha+\pi/2}=\unit$. The correlation measure~\eqref{eq:Bellcorrelator} is then given by
\begin{align}\label{eq:BellEqu}
E(\alpha,\beta)=p_{\alpha,\beta} + p_{\alpha+\pi/2,\beta+\pi/2} -p_{\alpha,\beta+\pi/2} -p_{\alpha+\pi/2,\beta},
\end{align}
where
\begin{align}
\begin{split}
p_{\alpha,\beta} &= \tr{P_\alpha \otimes P_\beta \ketbra{\psi^-}{\psi^-}}\\
&=\bra{\psi^-}P_\alpha \otimes P_\beta\ket{\psi^-}\\
&= \bracket{\psi^-}{\alpha\beta}\bracket{\alpha\beta}{\psi^-}\\
&=\abs{\bracket{\psi^-}{\alpha\beta}}^2\\
&=\frac{1}{2}\abs{\left( \bra{01}-\bra{10} \right) \ket{\alpha\beta}}^2\\
&=\frac{1}{2}\abs{ \bracket{0}{\alpha}\bracket{1}{\beta} - \bracket{1}{\alpha}\bracket{0}{\beta}}^2\\
&=\frac{1}{2}\abs{ -\sin(\alpha) \cos(\beta) + \cos(\alpha)\sin(\beta)}^2\\
&=\frac{1}{2}\abs{ - \sin (\alpha-\beta)}^2\\
&=\frac{1}{2} \sin^2(\alpha-\beta).
\end{split}
\end{align}
Using this in Eq.~\eqref{eq:BellEqu}, we get 
\begin{align}
\begin{split}
E(\alpha,\beta)&=\frac{1}{2}\Bigg[ \sin^2(\alpha-\beta)+\sin^2\left(\alpha+\frac{\pi}{2}-\beta-\frac{\pi}{2}\right)-\underbrace{\sin^2\left(\alpha-\beta-\frac{\pi}{2}\right)}_{\cos^2(\alpha-\beta)}-\underbrace{\sin^2\left(\alpha+\frac{\pi}{2}-\beta)\right)}_{\cos^2(\alpha-\beta)} \Bigg]\\
&=\sin^2(\alpha-\beta)-\cos^2(\alpha-\beta)\\
&=-\cos\left(2(\alpha-\beta)\right).
\end{split}
\end{align}
Therewith, we get for $\abs{\braket{S}}$ in Eq.~\eqref{eq:BellSex}
\begin{align}
\begin{split}
\abs{\braket{S}} &=\abs{E(\alpha,\beta)-E(\alpha,\beta')+E(\alpha',\beta)+E(\alpha',\beta')} \\
&=\abs{\cos\left( 2(\alpha-\beta)\right)-\cos\left( 2(\alpha-\beta')\right)+\cos\left( 2(\alpha'-\beta)\right)+\cos\left( 2(\alpha'-\beta')\right)}.
\end{split}
\end{align}
Now we are free to choose the angles $\alpha,\alpha',\beta,\beta'$. Note that we want to choose them such that $\abs{\braket{S}}$ is maximal, and check whether the maximum of $\abs{\braket{S}}$ beats the upper bound~\eqref{eq:BellSex} of any local realistic theory. It turns out that $\abs{\braket{S}}$ is miximal by choosing
\begin{align}
\alpha = 0+\varphi,\quad
\beta= \frac{\pi}{8} + \varphi,\quad
\alpha'= \frac{2\pi}{8} + \varphi,\quad
\beta'= \frac{3\pi}{8} + \varphi,
\end{align}
for any $\varphi \in \mathbb{R}$. With these angles, we get
\begin{align}
\begin{split}
\abs{\braket{S}} &=\abs{ \cos\left(-\frac{\pi}{4}\right)- \cos\left(-\frac{3\pi}{4}\right)+\cos\left(\frac{\pi}{4}\right)+\cos\left(-\frac{\pi}{4}\right)}\\
&=\abs{ \frac{1}{\sqrt{2}} - \left(  -\frac{1}{\sqrt{2}}\right)+ \frac{1}{\sqrt{2}}+ \frac{1}{\sqrt{2}}}\\
&= \frac{4}{\sqrt{2}}\\
&=2\sqrt{2}.
\end{split}
\end{align}
By assuming that quantum mechanics is complete, we found a situation where $\abs{\braket{S}}=2\sqrt{2}$. This is the quantum mechanical upper bound, i.e., in general we have  $\abs{\braket{S}}\leq 2\sqrt{2}$, which is called \underline{Tsirelson's bound}. Surprisingly, this bound beats the upper bound $\abs{\braket{S}}\leq 2$ [see Eq.~\eqref{eq:BellSex}] of any local realistic theory. That is, we found a way to experimentally check whether or not hidden variables are needed to make quantum mechanics a local realistic theory. There has been a huge experimental effort, with the outcomes speaking in favor of quantum mechanics. That is, nature is well described by quantum mechanics, and no additional hidden variables are required to turn quantum mechanics into a local realistic theory  (in the sense of EPR, see above).

\chapter{Quantum computation}  

\section{Qubits} 
\begin{definition} A \underline{qubit} is a two-level quantum system with Hilbert space $\mathcal{H}=\mathbb{C}^2$. The orthonormal basis states $\ket{0}$ and $\ket{1}$ of $\mathcal{H}=\mathbb{C}^2$ are called \underline{computational basis states} (the choice of this ON-basis is arbitrary).
\end{definition}

\begin{remarks}\leavevmode
\begin{enumerate}[1)]
\item A pure state of a qubit is a state vector 
\begin{align}\label{eq:QubitGeneral}
\ket{\psi} = a \ket{0} + b \ket{1},
\end{align}
with $a,b\in\mathbb{C}$ satisfying $|a|^2+|b|^2=1$.
\item In this chapter we will mainly work with pure states.
\item A possibly mixed state of a qubit is described by a density operator (i.e., Hermitian, positive, and trace one operator) $\rho\in\mathcal{D}(\mathbb{C}^2)$.
\end{enumerate}
\end{remarks}

\subsection{Pauli matrices} 
\begin{definition}
The Pauli matrices are defined as
\begin{align}\label{eq:Pauli}
X=\sigma_1=\begin{pmatrix}0&1\\ 1&0\end{pmatrix}, \quad Y=\sigma_2=\begin{pmatrix}0&-\im\\ \im&0\end{pmatrix}, \quad Z=\sigma_3=\begin{pmatrix}1&0\\ 0&-1\end{pmatrix}.
\end{align}
\end{definition}

\begin{remarks}\leavevmode
\begin{enumerate}[1)]
\item The Pauli matrices satisfy
\begin{align}\label{eq:sigmacom}
\sigma_j \sigma_k = \delta_{jk}\ \unit + \im \sum_{l=1}^3 \epsilon_{jkl} \ \sigma_l,
\end{align}
with $\delta_{jk}$ the Kronecker delta, $\epsilon_{jkl}$ the Levi-Civita symbol, and $j,k\in\{1,2,3\}$.

\item The Pauli matrices, together with the identity matrix
\begin{align}
\unit=\sigma_0 = \begin{pmatrix}1&0\\ 0&1\end{pmatrix}
\end{align}
form a basis for the real (hence, $v_0,v_1,v_2,v_3\in \mathbb{R}$ in Eq,~\eqref{eq:Hermdecomp}) vector space of $2\times 2$ Hermitian matrices. Hence, any $2\times 2$ Hermitian matrix $A$ can be written as
\begin{align}\label{eq:Hermdecomp}
A=\frac{1}{2} \left(v_0 \sigma_0 + v_1 \sigma_1 + v_2 \sigma_2 + v_3 \sigma_3 \right)=\frac{1}{2} \left(  v_0 \unit + \vec{v} \cdot \vec{\sigma}\right),
\end{align}
with $v_0,v_1,v_2,v_3\in \mathbb{R}$, $\vec{v}=(v_1,v_2,v_3) \in \mathbb{R}^3$, and $\vec{\sigma}=(\sigma_1,\sigma_2,\sigma_3)$.

\item If the Hermitian matrix $A$ from Eq.~\eqref{eq:Hermdecomp} is additionally positive and has unite trace, then $v_0=1$ and $\abs{\vec{v}} \leq 1$. Hence, any density operator of a qubit can be written as 
\begin{align}\label{eq:QubitRho}
\rho=\frac{1}{2} \left(  \unit + \vec{v} \cdot \vec{\sigma}\right),
\end{align}
with $\vec{v}\in\mathbb{R}^3$, and $\abs{\vec{v}} \leq 1$.

\item The density operator in~\eqref{eq:QubitRho} describes a pure state if and only if $\abs{\vec{v}} = 1$.

\end{enumerate}
\end{remarks}

\begin{proof}
1) - 4): Exercise. 
\end{proof}

\subsection{Bloch sphere} 
By Eq.~\eqref{eq:QubitRho}, we can illustrate the state of a qubit by a vector $\vec{v}\in\mathbb{R}^3$ within the \underline{Bloch sphere} (unite sphere):
\begin{center}
\includegraphics[width=0.5\linewidth]{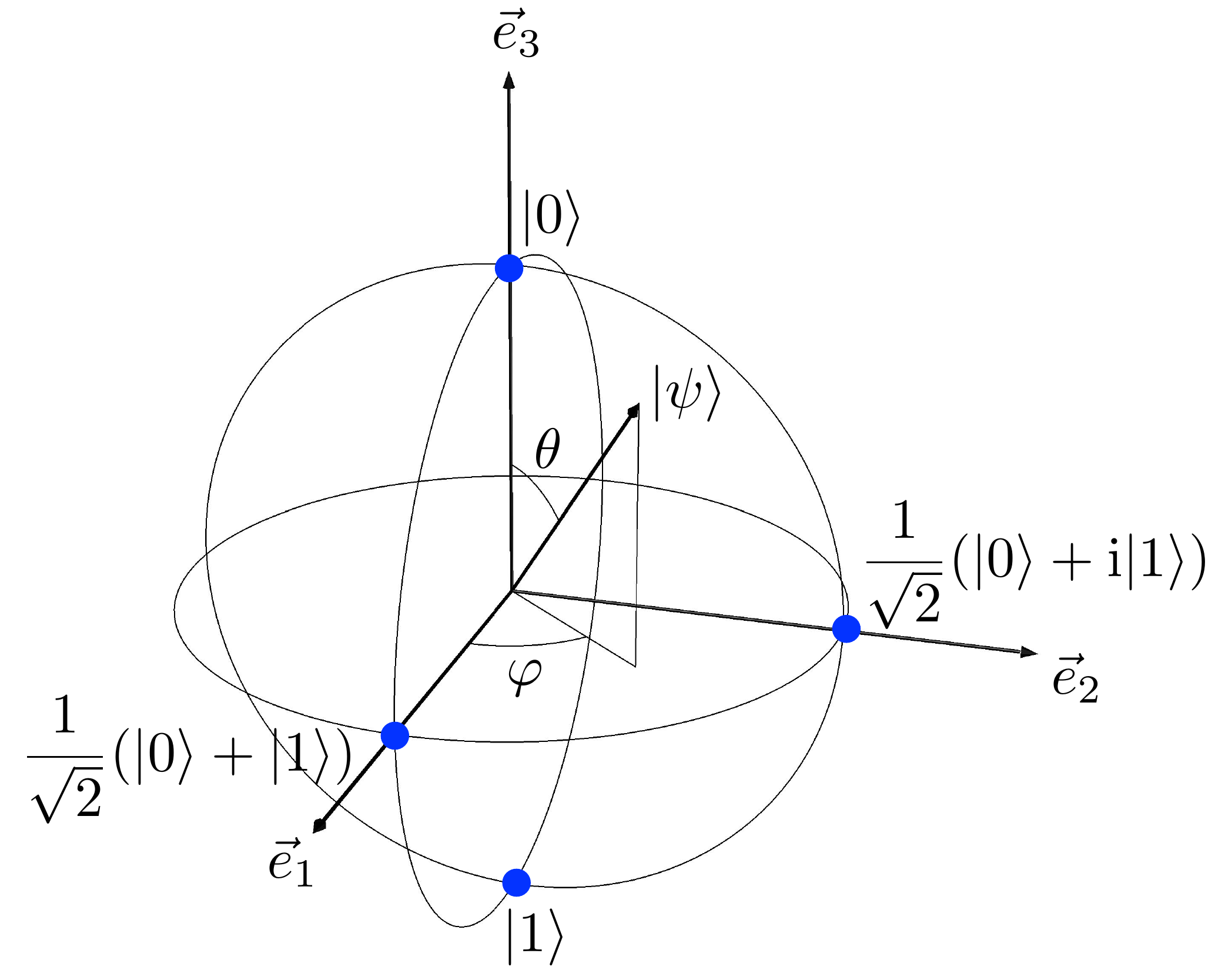}
\end{center}
In spherical coordinates we have 
\begin{align}
\vec{v}=\abs{\vec{v}} \begin{pmatrix} \sin \theta \cos\varphi \\ \sin \theta \sin\varphi \\ \cos\theta \end{pmatrix},
\end{align}
where $\abs{\vec{v}}\leq1$, $0\leq\theta\leq \pi$, and $0\leq\varphi<2\pi$. Pure states lie on the surface of the sphere and can be written as
\begin{align}\label{eq:QubitSphere1}
\ket{\psi}= \cos\left(\frac{\theta}{2} \right) \ket{0} + e^{\im \varphi} \sin\left(\frac{\theta}{2} \right) \ket{1},
\end{align}
and mixed states lie within the unit sphere. Note that Eq.~\eqref{eq:QubitSphere1} results from Eq.~\eqref{eq:QubitGeneral} by setting $a=\cos(\theta/2)e^{\im\gamma}$ and $b=\sin(\theta/2) e^{\im(\gamma+\varphi)}$ (which obviously satisfy $|a|^2+|b|^2=1$), such that $\ket{\psi}= e^{\im \gamma}[\cos(\theta/2) \ket{0} + e^{\im \varphi} \sin(\theta/2) \ket{1}]$, and neglecting the global phase $e^{\im \gamma}$.

\section{Quantum gates} 
 In general, quantum computation is nothing but the controlled execution of a unitary on a quantum register composed of qubits. A universal quantum computer would then allow to program arbitrary unitaries on such registers. In many ways, the faithful implementation of a quantum computation is therefore a specific application of coherent control techniques. 
 
 To mimic the structure of classical algorithms, quantum computation protocols are often formulated in terms of quantum circuits, where general unitaries on an $N$-qubit register are decomposed into elementary gates acting on subsets of qubits. The idea being that once one knows how to build reliably few-qubit gates one only needs to concatenate them properly to establish a $N$-qubit gate with $N\gg 1$. 

\subsection{Single qubit gates}
Operations on qubits must preserve the norm of pure states. This means that they must be unitary operations. In particular, single qubit gates correspond to arbitrary rotations on the Bloch sphere. Such rotations can be generated by the Pauli matrices from Eq.~\eqref{eq:Pauli}: A rotation by an angle $\vartheta$ around an unit vector
\begin{align}
\hat{n}=(n_x,n_y,n_z),\quad \abs{\hat{n}}=1
\end{align}
is mediated by the unitary
\begin{align}\label{eq:Rntheta}
R_{\hat{n}}(\vartheta) = e^{-\im \frac{\vartheta}{2} \hat{n}\cdot\vec{\sigma}},
\end{align}
where $\vec{\sigma}=(\sigma_1,\sigma_2,\sigma_3)=(X,Y,Z)$. Accordingly, an arbitrary single qubit unitary operator can be written as
\begin{align}\label{eq:Usq}
U=e^{\im \alpha}R_{\hat{n}}(\vartheta),
\end{align}
where $\alpha,\vartheta \in \mathbb{R}$ and $\hat{n} \in \mathbb{R}^3$. 

Besides the Pauli-$X$, $Y$, $Z$ gates, there are also other important single qubit gates:\\
\underline{$T$ gate (or $\pi/8$-gate)}
\begin{align}\label{eq:pi8gate}
T=e^{\im \pi/8} \begin{pmatrix} e^{-\im \pi/8} & 0 \\ 0 & e^{\im \pi/8}\end{pmatrix}=\begin{pmatrix} 1 & 0 \\ 0 & e^{\im \pi/4}\end{pmatrix}
\end{align}
\underline{phase gate}
\begin{align}
S=T^2=\begin{pmatrix} 1 & 0 \\ 0 & \im \end{pmatrix}
\end{align}
\underline{phase shift gate}
\begin{align}\label{eq:Pgate}
P_\varphi=\begin{pmatrix} 1 & 0 \\ 0 & e^{\im \varphi} \end{pmatrix}
\end{align}
\underline{Hadamard gate}
\begin{align}\label{eq:Hgate}
H=\frac{1}{\sqrt{2}}\left( X+Z\right)=\frac{1}{\sqrt{2}}\begin{pmatrix} 1 & 1 \\ 1 & -1 \end{pmatrix}=e^{\im \pi/2} R_{\hat{n}}(\pi),
\end{align}
with $\hat{n}=(1,0,1)/\sqrt{2}$.

\begin{remarks}\leavevmode
\begin{enumerate}[1)]
\item For fixed non-parallel unit vectors $\hat{n}$ and $\hat{m}$, any unitary single qubit gate $U$ [see Eq.~\eqref{eq:Usq}] can be decomposed as
\begin{align}\label{eq:Unm}
U=e^{\im\alpha} R_{\hat{n}}(\beta) R_{\hat{m}}(\gamma) R_{\hat{n}}(\delta),
\end{align}
where $\alpha,\beta,\gamma,\delta \in \mathbb{R}$.
\begin{proof}
Exercise.
\end{proof}

\item If $U$ is a unitary single qubit gate, then there exist unitary single qubit gates $A,B,C$, such that $ABC=\unit$ and 
\begin{align}\label{eq:ABC}
U=e^{\im\alpha} AXBXC.
\end{align}
\begin{proof}
Exercise. 
\end{proof}

\item The Pauli-$X$ gate is also know as \underline{quantum NOT} gate, since $X\ket{0}=\ket{1}$ and $X\ket{1}=\ket{0}$.

\item Quantum gates are often represented by \underline{quantum circuit diagrams}, where
\begin{enumerate}[i)]
\item time proceeds from left to right, and
\item wires represent qubits.
\end{enumerate}
The above single qubit gates are represented by

\begin{center}
\begin{minipage}[c]{0.4\textwidth}
\begin{align*}
\text{Pauli-X gate:}& 
\begin{tikzcd}
\lstick{} &\gate{X} &\qw
\end{tikzcd}\\
\text{Pauli-Y gate:}&\begin{tikzcd}
\lstick{} &\gate{Y} &\qw
\end{tikzcd}\\
\text{Pauli-Z gate:}&
\begin{tikzcd}
\lstick{} &\gate{Z} &\qw
\end{tikzcd}
\end{align*}
\end{minipage}
\begin{minipage}[c]{0.4\textwidth}
\begin{align*}
\text{$T$ gate:}&
\begin{tikzcd}
\lstick{} &\gate{T} &\qw
\end{tikzcd}\\
\text{phase gate:}&
\begin{tikzcd}
\lstick{} &\gate{S} &\qw
\end{tikzcd}\\
\text{phase shift gate:}&
\begin{tikzcd}
\lstick{} &\gate{P_\varphi} &\qw
\end{tikzcd}\\
\text{Hadamard gate:}&
\begin{tikzcd}
\lstick{} &\gate{H} &\qw
\end{tikzcd}
\end{align*}
\end{minipage}
\end{center}

\end{enumerate}
\end{remarks}

\subsection{CNOT-gate}
One of the most prominent and useful two-qubit gates is the \underline{controlled-$NOT$} or \underline{$CNOT$-gate}, with matrix and diagrammatic representation
\begin{align}
CNOT=\begin{pmatrix}
1&0&0&0\\
0&1&0&0\\
0&0&0&1\\
0&0&1&0
\end{pmatrix} = \begin{tikzcd} \lstick{} &\ctrl{1} &\qw\\ \lstick{} &\targ{} &\qw \end{tikzcd}=\begin{tikzcd} \lstick{} &\ctrl{1} &\qw\\
\lstick{} &\gate{X} &\qw \end{tikzcd}.
\end{align}
The first qubit is often called the \underline{control qubit}, and the second the \underline{target qubit}. The $CNOT$-gate is unitary and self-inverse, 
\begin{align}
(CNOT)^2=\unit,
\end{align}
and acts on computational basis states as
\begin{align}
CNOT\ket{c}\otimes \ket{t} =\ket{c} \otimes \ket{t \oplus c},
\end{align}
where $\oplus$ is the addition modulo $2$, and $c,t\in\{0,1\}$. That is, it acts as
\begin{align}
\begin{split}
CNOT\ket{00}&=\ket{00},\quad CNOT\ket{10}=\ket{11}\\
CNOT\ket{01}&=\ket{01},\quad CNOT\ket{11}=\ket{10}.
\end{split}
\end{align}

\begin{remarks}\leavevmode
\begin{enumerate}[1)]
\item In the following we often use the matrix representation 
\begin{align}
\ket{00}=\ket{0}=\begin{pmatrix}1\\0\\0\\0\end{pmatrix}, \quad \ket{01}=\ket{1}=\begin{pmatrix}0\\1\\0\\0\end{pmatrix}, \quad \ket{10}=\ket{2}=\begin{pmatrix}0\\0\\1\\0\end{pmatrix}, \quad \ket{11}=\ket{3}=\begin{pmatrix}0\\0\\0\\1\end{pmatrix},
\end{align}
and similar for more than two qubits. 

\item The $CNOT$ gate can generate entanglement. For example, consider
\begin{align}
CNOT\left( a \ket{0} + b\ket{1}\right) \otimes \ket{0} = a\ket{00} + b\ket{11} \neq \ket{\chi} \otimes \ket{\phi}
\end{align}
for all $a,b\neq 0$. 

\item The $CNOT$ in combination with the Hadamard gate allows to transform the computational basis states $\ket{00},\ket{01},\ket{10},\ket{11}$ into the Bell basis states from Eq.~\eqref{eq:Bellstates} by applying the circuit
\begin{center}
\begin{minipage}[c]{0.4\textwidth}
\begin{align*}
\begin{tikzcd} 
\lstick{} & \gate{H} &\ctrl{1} &\qw\\ 
\lstick{} & \qw &\targ{} &\qw 
\end{tikzcd}
\end{align*}
\end{minipage}
\begin{minipage}[c]{0.4\textwidth}
\begin{align}\label{eq:Bellcircuit}
\begin{split}
\ket{00}&\rightarrow \ket{\phi^+}\\
\ket{01}&\rightarrow \ket{\psi^+}\\
\ket{10}&\rightarrow \ket{\phi^-}\\
\ket{11}&\rightarrow \ket{\psi^-}
\end{split}
\end{align}
\end{minipage}
\end{center}

\begin{proof}
Exercise. 
\end{proof}

\item Since $(CNOT)^2=\unit$ and $H^2=\unit$, running the circuit from right to left inverts the transformation~\eqref{eq:Bellcircuit}. Hence, via the inverse circuit, the Bell states can unambiguously be identified via projective measurements in the computational basis. 

\item A $CNOT$ with control on the second qubit is realized by
\begin{align}
\begin{tikzcd} 
\lstick{} & \gate{H} &\ctrl{1} &\gate{H}&\qw\\ 
\lstick{} & \gate{H} &\targ{} &\gate{H}&\qw 
\end{tikzcd}
=
\begin{tikzcd} 
\lstick{} & \targ{} &\qw\\ 
\lstick{} & \ctrl{-1} &\qw 
\end{tikzcd}
\end{align}
\begin{proof}
Exercise.
\end{proof}

\item Concatenation of $CNOT$s allows to swap the computational basis states of two qubits,
\begin{align}\label{eq:SWAP}
SWAP=
\begin{tikzcd} 
\lstick{} & \swap{1} &\qw\\ 
\lstick{} & \targX{} &\qw 
\end{tikzcd}
=
\begin{tikzcd}
&  \gate[swap]{} & \qw  \\
&  &  \qw
\end{tikzcd}
=
\begin{tikzcd} 
\lstick{} & \ctrl{1} &\targ{}&\ctrl{1}&\qw\\ 
\lstick{} & \targ{} &\ctrl{-1}&\targ{}&\qw 
\end{tikzcd},
\end{align}
acting as
\begin{align}
SWAP \ket{a}\otimes \ket{b}=\ket{b} \otimes \ket{a},
\end{align}
for $a,b\in\{0,1\}$. 
\begin{proof}
Exercise. 
\end{proof}

\item A controlled-$NOT$ operation with control on $\ket{0}$ (instead of $\ket{1}$) is represented by 
\begin{align}
\begin{tikzcd} 
\lstick{} & \octrl{1} &\qw\\ 
\lstick{} & \targ{} &\qw 
\end{tikzcd}
=
\begin{tikzcd} 
\lstick{} & \gate{X}& \ctrl{1} &\gate{X}& \qw\\ 
\lstick{} & \qw & \targ{} &\qw & \qw 
\end{tikzcd}.
\end{align}
\end{enumerate}

\end{remarks}

\subsection{Controlled unitary gates}

\subsubsection{Single control}
For an arbitrary single qubit unitary operation $U$, the controlled-$U$ gate is represented by
\begin{align}
C(U)=\begin{pmatrix}1&0&0&0\\0&1&0&0\\0&0&U_{00}&U_{01}\\0&0&U_{10}&U_{11} \end{pmatrix}
=
\begin{tikzcd} 
\lstick{} & \ctrl{1} &\qw\\ 
\lstick{} & \gate{U} & \qw 
\end{tikzcd},
\end{align}
and acts as 
\begin{align}
C(U) \ket{c} \otimes \ket{t} = \ket{c} \otimes U^c\ket{t}
\end{align}
for $c,t\in\{0,1\}$. It can be decomposed into single qubit gates and $CNOT$ gates using the expression $U=e^{\im\alpha} AXBXC$ with $ABC=\unit$ from Eq.~\eqref{eq:ABC}. With this, it is easy to check the circuit identity
\begin{align}\label{cir:contrU}
\begin{tikzcd} 
\lstick{} & \ctrl{1} &\qw\\ 
\lstick{} & \gate{U} & \qw 
\end{tikzcd}
=
\begin{tikzcd} 
\lstick{} & \qw & \ctrl{1} &\qw & \ctrl{1}& \gate{P_\alpha} & \qw\\ 
\lstick{} & \gate{C} & \targ{} & \gate{B} & \targ{} & \gate{A}& \qw 
\end{tikzcd}.
\end{align}

\subsubsection{Multiple controls}
More generally, one also uses multi-qubit gates which condition the execution of a unitary $U$ on $k$ target qubits on the state of $n$ control qubits. Formally, this reads
\begin{align}
C^n(U)\ket{x_1x_2\dots x_n} \otimes \ket{\psi} = \ket{x_1x_2\dots x_n} \otimes U^{x_1x_2\cdots x_n}\ket{\psi}, 
\end{align}
with $x_j\in\{0,1\}$ and $\ket{\psi} \in (\mathbb{C}^2)^{\otimes k}$ a $k$-qubit state. The circuit diagram is
\begin{align}
\begin{tikzcd}
\lstick[wires=4]{$n$ control qubits} & \qw & \ctrl{1} & \qw & \qw\\
\lstick{} &\qw & \ctrl{1} & \qw & \qw\\
\lstick{} & \text{$\vdots$} &\text{$\vdots$} & \text{$\vdots$} & \\
\lstick{} &\qw & \ctrl{1} & \qw & \qw\\
\lstick[wires=4]{$k$ target qubits}  & \qw & \gate[4,nwires={3}]{U} & \qw & \qw\\
\lstick{}  & \qw &  & \qw & \qw \\
\lstick{}  & \text{$\vdots$}  & &\text{$\vdots$}  & \\
\lstick{}  & \qw &   & \qw & \qw
\end{tikzcd}.
\end{align}
A special incident thereof is the \underline{Toffoli gate} (or \underline{$C^2(NOT)$} gate or \ul{$CCNOT$ gate})
\begin{align}
C^2(NOT) =
\begin{tikzcd}
\lstick{} & \ctrl{2} & \qw \\
\lstick{} & \ctrl{1} & \qw \\
\lstick{} & \targ{} & \qw
\end{tikzcd},
\end{align}
acting as
\begin{align}
C^2(NOT) \ket{a}\otimes\ket{b}\otimes\ket{c} = \ket{a}\otimes\ket{b}\otimes\ket{c\oplus ab},
\end{align}
where $a,b,c\in\{0,1\}$, and $[C^2(NOT)]^2=\unit$. It can be decomposed into Hadamard, phase, $T$, and $CNOT$ gates via
\begin{align}\label{cir:Toffoli}
\begin{tikzcd}
\lstick{} & \ctrl{2} & \qw \\
\lstick{} & \ctrl{1} & \qw \\
\lstick{} & \targ{} & \qw
\end{tikzcd}
=
\begin{tikzcd}
\lstick{} & \qw & \qw & \qw & \ctrl{2} & \qw & \qw & \qw & \ctrl{2} & \qw & \ctrl{1} & \qw & \ctrl{1} & \gate{T} &  \qw \\
\lstick{} & \qw & \ctrl{1} & \qw & \qw & \qw & \ctrl{1} & \qw & \qw & \gate{T^\dagger} & \targ{} & \gate{T^\dagger} & \targ{} & \gate{S} &  \qw \\
\lstick{} & \gate{H} & \targ{} & \gate{T^\dagger} & \targ{} & \gate{T} & \targ{} & \gate{T^\dagger} & \targ{} & \gate{T} & \gate{H} & \qw & \qw & \qw &  \qw 
\end{tikzcd}
\end{align}

\begin{proof}
Exercise.
\end{proof}

The Toffoli gate allows for a simple construction of a $C^n(U)$-gate, where $n-1$ \underline{working} (or \underline{ancillary}) qubits are used in addition to the $n$ control qubits and the target qubit(s). For $n=5$, the circuit is
\begin{align}\label{cir:CnU}
\begin{tikzcd}
\lstick[wires=5]{$n$ control qubits} & & \lstick{$\ket{c_1}$}& \ctrl{5} & \qw & \qw & \qw & \qw & \qw & \qw & \qw & \ctrl{5} & \qw \rstick{$\ket{c_1}$} \\
\lstick{} & & \lstick{$\ket{c_2}$}& \ctrl{4} & \qw & \qw & \qw & \qw & \qw & \qw & \qw & \ctrl{4} & \qw\rstick{$\ket{c_2}$}  \\
\lstick{} &  &\lstick{$\ket{c_3}$}& \qw & \ctrl{4} & \qw & \qw & \qw & \qw & \qw & \ctrl{4} & \qw & \qw \rstick{$\ket{c_3}$} \\
\lstick{} &  &\lstick{$\ket{c_4}$}& \qw & \qw & \ctrl{4} & \qw & \qw & \qw & \ctrl{4} & \qw & \qw & \qw \rstick{$\ket{c_4}$} \\
\lstick{} & & \lstick{$\ket{c_5}$}& \qw & \qw & \qw & \ctrl{4} & \qw & \ctrl{4} & \qw & \qw & \qw & \qw \rstick{$\ket{c_5}$} \\
\lstick[wires=4]{$n-1$ working qubits} &  &\lstick{$\ket{0}$}& \targ{} & \ctrl{1} & \qw & \qw & \qw &\qw &  \qw & \ctrl{1} & \targ{} & \qw \rstick{$\ket{0}$} \\
\lstick{} & & \lstick{$\ket{0}$}& \qw & \targ{} & \ctrl{1} & \qw & \qw & \qw & \ctrl{1} & \targ{} & \qw & \qw \rstick{$\ket{0}$} \\
\lstick{} & & \lstick{$\ket{0}$}& \qw & \qw & \targ{} & \ctrl{1} & \qw & \ctrl{1} & \targ{} & \qw & \qw & \qw \rstick{$\ket{0}$} \\
\lstick{} & & \lstick{$\ket{0}$}& \qw & \qw & \qw & \targ{} & \ctrl{1} & \targ{} & \qw & \qw & \qw & \qw \rstick{$\ket{0}$} \\
\lstick{target qubit(s)} & &   & \qw & \qw & \qw & \qw  & \gate{U} & \qw & \qw & \qw & \qw & \qw 
\end{tikzcd}.
\end{align}
Via the first $n-1$ Toffoli gates, this circuit maps the last working qubit from initial state $\ket{0}$ to state $\ket{c_1c_2\cdots c_n}$ on which $U$ is controlled. The last $n-1$ Toffoli gates rest the working qubits to their initial states using the self-inverse property of $C^2(NOT)$.

\subsection{Measurements}
A final element used in quantum circuits is measurement. A projective measurement of $\ket{\psi}$ in the computational basis with measurement operators (i.e., projectors) $\{\ketbra{0}{0}, \ketbra{1}{1}\}$ is illustrated by
\begin{align}
\begin{tikzcd}
\lstick{$\ket{\psi}$} &\meter{}
\end{tikzcd}.
\end{align}
In the theory of quantum circuits it is conventional to not use any special symbols to denote more general measurements. Why?

\begin{remarks}\leavevmode
\begin{enumerate}[1)]
\item Any measurement can be represented by unitary transformations with ancilla qubits followed by projective measurements.
\begin{proof}
(This is essentially Neumark's theorem [see Sec.~\ref{sec:POVM}]) Suppose we have $\ket{\psi}$ and want to perform a measurement with measurement operators $\{M_m\}_m$. Then introduce ancilla qubits in $\ket{000\dots}=\ket{0}$, and implement a unitary $U$ on $\ket{\psi}\ket{0}$ such that
\begin{align}
U\ket{\psi}\ket{0} = \sum_m M_m \ket{\psi}\ket{m}.
\end{align}
A projective measurement of the ancilla qubits with projectors $P_m=\unit\otimes\ketbra{m}{m}$ then yields the probability
\begin{align}
p(m)&=\tr{P_m U \ket{\psi}\ket{0}\bra{\psi}\bra{0} U^\dagger}\\
&=\bra{\psi}\bra{0} U^\dagger P_m U \ket{\psi}\ket{0}\\
&= \sum_{j,k} \bra{\psi} M_k^\dagger \bra{k} \left( \unit \otimes \ketbra{m}{m} \right) M_j \ket{\psi}\ket{j}\\
&=\sum_{j,k} \bra{\psi} M_k^\dagger M_j \ket{\psi} \bracket{k}{m}\bracket{m}{j}\\
&=\bra{\psi} M_m^\dagger M_m \ket{\psi}\\
&=\tr{M_m^\dagger M_m \ketbra{\psi}{\psi}}.
\end{align}
\end{proof}
The circuit diagram of this measurement is
\begin{align}
\begin{tikzcd}
\lstick{$\ket{\psi}$} & \gate[wires=2]{U} \qwbundle{} & \qwbundle{} & \qw \rstick{$\frac{M_m \ket{\psi}}{\sqrt{\bra{\psi} M_m^\dagger M_m \ket{\psi}}}$}\\
\lstick{$\ket{0}$} & \qwbundle{}  &\qwbundle{} & \meter{$m$} 
\end{tikzcd},
\end{align}
where $\begin{tikzcd} &\qwbundle{} \end{tikzcd}$ denotes a bundle of qubits. Note that we will regularly experience diagrams of this form. 

\item Measurement is generally considered to be an irreversible operation, destroying quantum information and replacing it with classical information.

\item Measurements commute with controls, that is,
\begin{align}
\begin{tikzcd}
\lstick{} & \ctrl{1} & \meter{} \\
& \gate{U} & \qw
\end{tikzcd}
\quad =\quad 
\begin{tikzcd}
\lstick{} & \meter{} & \cwbend{1} & \\
& \qw & \gate{U} & \qw
\end{tikzcd}
\quad =\quad 
\begin{tikzcd}
\lstick{} & \meter{} & \\
& \gate{U} \vcw{-1} & \qw
\end{tikzcd},
\end{align}
where the double line $\begin{tikzcd} & \cw \end{tikzcd}$ respresents a classical bit. 
\begin{proof}
Exercise. 
\end{proof}

\item Any unterminated quantum wires at the end of a quantum circuit may be assumed to be measured. 
\end{enumerate}
\end{remarks}

\section{Universal gate sets} 
The aim of the following is to show that any unitary operation on a $n$-qubit quantum register can be approximated to arbitrary accuracy by a combination of $CNOT$s and a finite set of single-qubit gates. ($CNOT$, Hadarmard, $T$-gate). We proceed in three steps, and show that
\begin{enumerate}[1)]
\item any unitary can be decomposed into a product of two-level unitaries,

\item any two-level unitary can be implemented by combinations of $CNOT$s and single-qubit gates,

\item any single-qubit gate can be approximated to arbitrary accuracy by a concatenation Hadamard and $T$-gates.
\end{enumerate}

\subsection{Decomposition into products of two-level unitaries}
The essential idea behind the decomposition can be understood by the example of a $3\times 3$ unitary
\begin{align}\label{eq:DU}
U=\begin{pmatrix}
a&d&g \\ 
b&e&h \\ 
c&f&j 
\end{pmatrix}.
\end{align}
We want to find two-level unitaries $U_1,U_2,U_3$, such that $U_3U_2U_1 U=\unit$, and, hence, $U=U_1^\dagger U_2^\dagger U_3^\dagger$. First we eliminate $b$ in~\eqref{eq:DU}. To this end, we choose the two-level unitary
\begin{align}
U_1=\begin{pmatrix}
\frac{a^*}{\sqrt{|a|^2+|b|^2}} & \frac{b^*}{\sqrt{|a|^2+|b|^2}} & 0 \\
\frac{b}{\sqrt{|a|^2+|b|^2}} & \frac{-a}{\sqrt{|a|^2+|b|^2}} & 0 \\
0 & 0 & 1 
\end{pmatrix},
\end{align}
such that 
\begin{align}\label{eq:DU2}
U_1 U= \begin{pmatrix}
a' & d' & g' \\
0 & e' & h' \\
c' & f' & j'
\end{pmatrix}.
\end{align}
Note that $U_1^\dagger U_1=\unit$. Next we eliminate $c'$ in~\eqref{eq:DU2} in a similar way by choosing
\begin{align}
U_2=\begin{pmatrix}
\frac{a'^*}{\sqrt{|a'|^2+|c'|^2}} & 0 & \frac{c'^*}{\sqrt{|a'|^2+|c'|^2}}  \\
0 & 1 & 0 \\
\frac{c'}{\sqrt{|a'|^2+|c'|^2}} & 0 & \frac{-a'}{\sqrt{|a'|^2+|c'|^2}} 
\end{pmatrix},
\end{align}
such that
\begin{align}\label{eq:DU3}
U_2U_1 U= \begin{pmatrix}
1 & d'' & g'' \\
0 & e'' & h'' \\
0 & f'' & j''
\end{pmatrix}.
\end{align}
Since $U_2,U_1,U$ are unitary, $U_2U_1U$ is also unitary. Hence, the first row in~\eqref{eq:DU3} must have norm $1$. Accordingly, $d''=g''=0$, hence
\begin{align}\label{eq:DU4}
U_2U_1 U= \begin{pmatrix}
1 & 0 & 0 \\
0 & e'' & h'' \\
0 & f'' & j''
\end{pmatrix}.
\end{align}
Finally, to eliminate $f''$ in the second row (and, thus, $h''$), we multiply by
\begin{align}
U_3= \begin{pmatrix}
1 & 0 & 0 \\
0 & e''^* & f''^* \\
0 & h''^* & j''^*
\end{pmatrix},
\end{align}
leading to
\begin{align}
U_3U_2U_1 U= \unit. 
\end{align}
In general, if $U$ is a $d\times d$ unitary, we similarly find two-level unitaries $U_1,\dots,U_{d-1}$ such that
\begin{align}
U_{d-1} \cdots U_1 U =     \left(
    \begin{array}{c | c c c} 
1 & 0 & 0 & \cdots \\ \hline
0 &   &   &  \\
0 &   & V & \\
\vdots &   &   & 
     \end{array} 
    \right).
\end{align}
Repeating this procedure for the $d-1\times d-1$ subunitary $V$ and all further subunitaries results in the decomposition
\begin{align}\label{eq:Udecomp}
U=U_1^\dagger U_2^\dagger \dots U_k^\dagger,
\end{align}
with 
\begin{align}
k \leq (d-1)+(d-2)+ \dots +1= \frac{d(d-1)}{2}
\end{align}
factors. Note that for an $n$-qubit register, we have $d=2^n$, such that
\begin{align}\label{eq:Unumberk}
k\leq \frac{2^n(2^n-1)}{2} = 2^{n-1}(2^n-1),
\end{align}
which grows exponentially in $n$, as $\mathcal{O}(4^n)$.

\subsection{Representation of two-level unitaries by single-qubit and CNOT gates}\label{sec:twolevelbyCNOT}
We now show that any of the two-level unitaries $U_j^\dagger$ in Eq.~\eqref{eq:Udecomp} can be implemented by a combination of $CNOT$ gates and a single qubit unitary. 

Consider the example of an $n=3$ qubit register with the $d=2^n=8$ computational basis states labeled $000,001,010,011,100,101,110,111$. Given, for example, the two-level unitary
\begin{align}
U=\begin{pmatrix}
a & 0 & 0 & 0 & 0 & 0 & 0 & c \\
0 & 1 & 0 & 0 & 0 & 0 & 0 & 0 \\
0 & 0 & 1 & 0 & 0 & 0 & 0 & 0 \\
0 & 0 & 0 & 1 & 0 & 0 & 0 & 0 \\
0 & 0 & 0 & 0 & 1 & 0 & 0 & 0 \\
0 & 0 & 0 & 0 & 0 & 1 & 0 & 0 \\
0 & 0 & 0 & 0 & 0 & 0 & 1 & 0 \\
b & 0 & 0 & 0 & 0 & 0 & 0 & d 
\end{pmatrix},
\end{align}
we see that $U$ only acts non-trivially on the basis states $\ket{000}$ and $\ket{111}$, with the embedded two-level unitary 
\begin{align}
V=\begin{pmatrix}a&c \\b&d \end{pmatrix}.
\end{align}
Since $V$ is essentially a single qubit unitary acting on the subspace $\{\ket{000},\ket{111}\}$ we aim at transforming this subspace to the subspace $\{\ket{0},\ket{1}\}$ spanned by the basis states of a single qubit, applying $V$, and then transforming back via the inverse transformation. That is, 
\begin{enumerate}[1)]
\item map $\ket{000}$ and $\ket{111}$ onto $\ket{0}\otimes\ket{11}$ and $\ket{1}\otimes\ket{11}$,

\item execute V on the first qubit, conditioned on the second and third qubit,

\item map $\ket{011}$ and $\ket{111}$ back to $\ket{000}$ and $\ket{111}$.
\end{enumerate}
 
 We use a so-called \underline{Gray-code}, which is a sequence of register states which differ by one qubit state only and connect the two states on which $U$ acts non-trivially,
 \begin{align}\label{eq:gray}
 \begin{array}{c c c c }
                & \text{A} & \text{B} & \text{C} \\
 \ket{000} & 0 & 0 & 0\\
                & 0 & 0 & 1\\
               & 0 & 1 & 1\\
 \ket{111} & 1 & 1 & 1\\
 \end{array}.
 \end{align}
From the Gray-code~\eqref{eq:gray} and steps 2) and 3), we find that $U$ is implemented by the circuit
\begin{align}
\begin{tikzcd}
\lstick{A} & \octrl{1} & \octrl{1} & \gate{V} & \octrl{1} & \octrl{1} & \qw \\
\lstick{B} & \octrl{1} & \targ{} & \ctrl{-1} & \targ{} & \octrl{1} & \qw \\
\lstick{C} & \targ{} & \ctrl{-1} &  \ctrl{-2} & \ctrl{-1} & \targ{} & \qw
\end{tikzcd}.
\end{align}
Here the first controlled-$NOT$ maps $\ket{000}$ and $\ket{111}$ to $\ket{001}$ and $\ket{111}$. The second controlled-$NOT$ maps $\ket{001}$ and $\ket{111}$ to $\ket{011}$ and $\ket{111}$. The unitary $V$ is then applied conditioned on the last two qubits, and, due to controlled-$NOT$ gates being self-inverse, the last two controls map $\ket{011}$ and $\ket{111}$ back to $\ket{000}$ and $\ket{111}$. Similarly, one can find such circuits for any two-level unitary. 

In conclusion, together with~\eqref{cir:contrU},~\eqref{cir:Toffoli}, and~\eqref{cir:CnU} we showed that single qubit and $CNOT$ gates are universal, since their concatenation allows the implementation of arbitrary $d\times d$ unitaries. 

\subsubsection*{How many gates are needed to implement $U$?}
\begin{itemize}
\item A two-level unitary $U_j^\dagger$ requires at most $2(n-1)$ controlled operations, as well as a controlled $V$ operation. Each of them requires $\mathcal{O}(n)$ $CNOT$s and single-qubit gates (not shown here). Hence, $U_j^\dagger$ requires $\mathcal{O}(n^2)$ gates. 

\item In Eq.~\eqref{eq:Unumberk} we found that $\mathcal{O}(2^{2n})=\mathcal{O}(4^n)$ two-level unitaries are needed for a $n$-qubit unitary $U$.
\end{itemize}

Altogether, an arbitrary unitary operation on $n$ qubits can be implemented using circuits containing $\mathcal{O}(n^2 4^n)$ single-qubit and $CNOT$ gates. Since this number is exponentially in $n$, it is difficult to find fast quantum algorithms using this universality approach.

\subsection{Approximation of arbitrary single-qubit gates by Hadamard and T-gates}
When contemplating the construction of a universal quantum computer from a set of standard constituents, one seeks to approximate an arbitrary single-qubit gate [see Eq.~\eqref{eq:Usq}] by concatenation of a discrete set of gates. In the following we show that the Hadamard and $\pi/8$-gate from Eqs.~\eqref{eq:pi8gate} and~\eqref{eq:Hgate} constitute such a set. First, however, let us quantify how well an unitary $V$ approximates the ideal unitary $U$. 

\begin{definition}
For $V$ a unitary which approximates the unitary $U$ on $\mathcal{H}$, the \underline{approximation error} is defined as
\begin{align}\label{eq:approxerror}
E(U,V)=\max_{\ket{\psi}\in\mathcal{H}} \norm{ (U-V) \ket{\psi}},
\end{align}
where the maximization is over all normalized states $\ket{\psi}\in\mathcal{H}$.
\end{definition}

\begin{remarks}\leavevmode
\begin{enumerate}[1)]
\item $E(U,V)$ gives an upper bound for the difference between the measurement statistics derived from measurements on $U\ket{\psi}$ and $V\ket{\psi}$, respectively, for arbitrary $\ket{\psi}\in\mathcal{H}$ and arbitrary measurement operators $M$. 
\begin{proof}
For measurements on $U\ket{\psi}$ and $V\ket{\psi}$, the output probability corresponding to the measurement operator $M$ is given by $p_U=\mathrm{Tr}(M^\dagger M U\ketbra{\psi}{\psi}U^\dagger)$ and $p_V=\mathrm{Tr}(M^\dagger M V\ketbra{\psi}{\psi}V^\dagger)$, respectively. Their difference is
\begin{align}
\begin{split}
\abs{p_U-p_V}&=\abs{\bra{\psi} U^\dagger M^\dagger MU\ket{\psi} -\bra{\psi} V^\dagger M^\dagger  MV\ket{\psi} }\\
&=\abs{\bra{\psi} U^\dagger M^\dagger M(U-V)\ket{\psi} +\bra{\psi} (U^\dagger -V^\dagger) M^\dagger MV\ket{\psi} }\\
&\leq \abs{\bra{\psi} U^\dagger M^\dagger M(U-V)\ket{\psi}} +\abs{\bra{\psi} (U^\dagger -V^\dagger) M^\dagger MV\ket{\psi} }.
\end{split}
\end{align}
Next we use the Cauchy Schwarz inequality $\abs{\bracket{x}{y}} \leq \norm{x}\norm{y}$, as well as $\norm{M^\dagger M(U-V)\ket{\psi}} \leq \norm{(U-V)\ket{\psi}}$ for all $M$, such that
\begin{align}
\begin{split}
\abs{p_U-p_V}&\leq \norm{(U-V)\ket{\psi}} +\norm{\bra{\psi} (U^\dagger -V^\dagger)  } \\
&\leq 2 E(U,V).
\end{split}
\end{align}
\end{proof}

\item For unitaries $U_j$ and $V_j$, with $j\in\{1,\dots,m\}$, the approximation error satisfies
\begin{align}\label{eq:ApproxErrorSum}
E(U_mU_{m-1}\dots U_m, V_mV_{m-1}\dots V_1)\leq \sum_{j=1}^m E(U_j,V_j).
\end{align}
\begin{proof}
We prove by induction and start with $m=2$:
\begin{align}
\begin{split}
E(U_2U_1,V_2V_1) &= \max_{ \ket{\psi}\in\mathcal{H}} \norm{ (U_2U_1-V_2V_1)\ket{\psi}}\\
&=\max_{ \ket{\psi}\in\mathcal{H}} \norm{ (U_2U_1-U_2V_1)\ket{\psi} +(U_2V_1-V_2V_1)\ket{\psi} }\\
&\leq \max_{ \ket{\psi}\in\mathcal{H}} \left( \norm{ U_2(U_1-V_1)\ket{\psi} } + \norm{(U_2-V_2)V_1\ket{\psi} } \right) \\
&\leq E(U_1,V_1) + E(U_2,V_2).
\end{split}
\end{align}
The case $m>2$ follows from induction. 
\end{proof}

\end{enumerate}
\end{remarks}

We are now set to show that arbitrary single-qubit gates can be approximated to arbitrary accuracy by concatenations of Hadamard and $T$-gates. Let us start with rewriting the $T$-gate from Eq.~\eqref{eq:pi8gate} using
\begin{align}\label{eq:ecossin}
e^{\im \varphi \ \vec{n} \cdot \vec{\sigma}} = \cos\varphi\ \unit + \im \sin\varphi \ \vec{n}\cdot \vec{\sigma},
\end{align}
such that
\begin{align}\label{eq:Tsc}
\begin{split}
T&=e^{\im \frac{\pi}{8}} \begin{pmatrix}e^{-\im \frac{\pi}{8}}  & 0 \\ 0 & e^{\im \frac{\pi}{8}}   \end{pmatrix}\\
&=e^{\im \frac{\pi}{8}}  e^{-\im \frac{\pi}{8} Z}  \\
&= e^{\im \frac{\pi}{8}}   \left[ \cos\left( \pi/8 \right) \unit - \im \sin\left( \pi/8 \right) Z\right].
\end{split}
\end{align}
Next, we use $H=\frac{1}{\sqrt{2}}(X+Z)$ from Eq.~\eqref{eq:Hgate} and calculate $HTH$ using [recall Eq.~\eqref{eq:sigmacom}]
\begin{align}
\sigma_j \sigma_k = \delta_{jk}\ \unit +\im \sum_{l=1}^3 \epsilon_{jkl} \ \sigma_l
\end{align}
for $j,k,l \in \{1,2,3\}$, which gives, in particular, 
\begin{align}
\begin{split}
ZX&=\im Y \quad XZ=-\im Y \\
XY&=\im Z \quad YX=-\im Z\\
YZ&=\im X \quad ZY=-\im X.
\end{split}
\end{align}
We get
\begin{align}\label{eq:HTH}
\begin{split}
HTH&=\frac{1}{2}e^{\im \frac{\pi}{8}} (X+Z) \left[ \cos\left( \pi/8 \right) \unit - \im \sin\left( \pi/8 \right) Z\right] (X+Z)\\
&=\frac{1}{2}e^{\im \frac{\pi}{8}} \left[ \cos\left( \pi/8 \right)X -\sin\left( \pi/8 \right) Y +\cos\left( \pi/8 \right)Z -\im  \sin\left( \pi/8 \right) \unit\right] (X+Z)\\
&=\frac{1}{2}e^{\im \frac{\pi}{8}} \Big[ \cos\left( \pi/8 \right) \unit -\im \cos\left( \pi/8 \right) Y + \im \sin\left( \pi/8 \right)  Z - \im \sin\left( \pi/8 \right) X \\
&+ \im \cos\left( \pi/8 \right) Y + \cos\left( \pi/8 \right)\unit - \im \sin\left( \pi/8 \right)X - \im \sin\left( \pi/8 \right)Z\Big] \\
&=e^{\im \frac{\pi}{8}} \left[\cos\left( \pi/8 \right) \unit - \im \sin\left( \pi/8 \right)X  \right] 
\end{split}
\end{align}
Using Eqs.~\eqref{eq:Tsc}, and~\eqref{eq:HTH}, concatenation of $T$ and $HTH$ gives
\begin{align}\label{eq:THTH}
\begin{split}
T(HTH)&=e^{\im \frac{\pi}{4}}\left[\cos\left( \pi/8 \right) \unit - \im \sin\left( \pi/8 \right)Z  \right]\left[\cos\left( \pi/8 \right) \unit - \im \sin\left( \pi/8 \right)X  \right]\\
&=e^{\im \frac{\pi}{4}} \Big[\cos^2\left( \pi/8 \right) \unit - \im \sin\left( \pi/8 \right)\cos\left( \pi/8 \right)X - \im \sin\left( \pi/8 \right)\cos\left( \pi/8 \right)Z-\im \sin^2\left( \pi/8 \right) Y\Big]\\
&=e^{\im \frac{\pi}{4}} \left[\cos^2\left( \pi/8 \right) \unit - \im \sin\left( \pi/8 \right) \left[\cos\left( \pi/8 \right), \sin\left( \pi/8 \right),\cos\left( \pi/8 \right) \right] \cdot \vec{\sigma}\right]\\
&=e^{\im \frac{\pi}{4}} \left[\cos^2\left( \pi/8 \right) \unit - \im \sin\left( \pi/8 \right) \sqrt{1+\cos^2\left(\pi/8\right)}\  \hat{n}\cdot\vec{\sigma} \right]
\end{split}
\end{align}
with
\begin{align}\label{eq:nrot}
\hat{n}=\frac{\left[\cos\left( \pi/8 \right), \sin\left( \pi/8 \right),\cos\left( \pi/8 \right) \right]}{\sqrt{1+\cos^2\left(\pi/8\right)}}.
\end{align}
Now we define $\theta$ such that
\begin{align}\label{eq:theta}
\cos(\theta/2)=\cos^2(\pi/8).
\end{align}
It can be shown that $\vartheta$ is an irrational multiple of $2\pi$. A proof can be found in the literature, e.g. see P. O. Boykin, et al., arXiv:quant-ph/9906054, pp. 4+10. 

With Eq.~\eqref{eq:theta} we have
\begin{align}\label{eq:thetaSin}
\begin{split}
\sin(\pi/8) \sqrt{1+\cos^2\left(\pi/8\right)} &=\sqrt{1-\cos\left(\theta/2\right)}\sqrt{1+\cos\left(\theta/2\right)}\\
&=\sqrt{1-\cos^2\left(\theta/2\right)}\\
&=\sin(\theta/2),
\end{split}
\end{align}
such that Eq.~\eqref{eq:THTH} becomes 
\begin{align}\label{eq:THTHRn}
\begin{split}
T(HTH)&= e^{\im \frac{\pi}{4}} \left[  \cos(\theta/2) \ \unit - \im \sin(\theta/2)\ \hat{n}\cdot\vec{\sigma}\right]\\
&= e^{\im \frac{\pi}{4}} R_{\hat{n}}(\theta).
\end{split}
\end{align}
Hence, $THTH$ represents a rotation by $\theta$ around the unit vector $\hat{n}$ defined in Eq.~\eqref{eq:nrot}. 

Next, we proceed similarly for $(HTH)T$ using Eqs.~\eqref{eq:Tsc} and~\eqref{eq:HTH}:
\begin{align}
\begin{split}
(HTH)T&=e^{\im \frac{\pi}{4}}\left[\cos\left( \pi/8 \right) \unit - \im \sin\left( \pi/8 \right)X  \right]\left[\cos\left( \pi/8 \right) \unit - \im \sin\left( \pi/8 \right)Z  \right]\\
&=e^{\im \frac{\pi}{4}} \Big[\cos^2\left( \pi/8 \right) \unit - \im \sin\left( \pi/8 \right)\cos\left( \pi/8 \right)Z - \im \sin\left( \pi/8 \right)\cos\left( \pi/8 \right)X+\im \sin^2\left( \pi/8 \right) Y\Big]\\
&=e^{\im \frac{\pi}{4}} \left[\cos^2\left( \pi/8 \right) \unit - \im \sin\left( \pi/8 \right) \sqrt{1+\cos^2\left(\pi/8\right)}\  \hat{m}\cdot\vec{\sigma} \right]
\end{split},
\end{align}
with
\begin{align}\label{eq:mrot}
\hat{m}=\frac{\left[\cos\left( \pi/8 \right), -\sin\left( \pi/8 \right),\cos\left( \pi/8 \right) \right]}{\sqrt{1+\cos^2\left(\pi/8\right)}}.
\end{align}
Again, with $\theta$ from Eq.~\eqref{eq:theta}, and using Eq.~\eqref{eq:thetaSin}, we find
\begin{align}\label{eq:HTHTRm}
\begin{split}
(HTH)T&= e^{\im \frac{\pi}{4}} \left[  \cos(\theta/2) \ \unit - \im \sin(\theta/2)\ \hat{m}\cdot\vec{\sigma}\right]\\
&= e^{\im \frac{\pi}{4}} R_{\hat{m}}(\theta).
\end{split}
\end{align}
This is a rotation by $\theta$ around the unit vector $\hat{m}$, which is non-parallel to $\hat{n}$ [cf. Eqs.~\eqref{eq:nrot} and~\eqref{eq:mrot}]. 

Due to the irrationality of $\theta$, for any angle $\alpha \in [0,2\pi)$ and any desired accuracy $\epsilon >0$,
\begin{align}
\exists\ l : \abs{(l\theta -\alpha) \mod 2\pi}<\epsilon,
\end{align}
since integer multiples of an irrational number fill the unit circle densely. Hence, 
\begin{align}
\left[ R_{\hat{n}}(\theta) \right]^l \approx R_{\hat{n}}(\alpha).
\end{align}
Consequently, for the approximation error~\eqref{eq:approxerror}, we find that
\begin{align}\label{eq:errorRot}
\forall \delta >0 \ \exists \ l\in\mathbb{N}\ : \ E\left( \left[ R_{\hat{n}}(\theta) \right]^l, R_{\hat{n}}(\alpha)\right) < \delta. 
\end{align}
Finally, by Eq.~\eqref{eq:Unm} we found that (up to an unimportant global phase) an arbitrary single-qubit gate $U$ can be written as $U=R_{\hat{n}}(\alpha)R_{\hat{m}}(\beta)R_{\hat{n}}(\gamma)$. The results~\eqref{eq:THTHRn},~\eqref{eq:HTHTRm}, and~\eqref{eq:errorRot} therefore imply that
\begin{align}
\begin{split}
&\forall \delta>0 \ \exists\ l_1,l_2,l_3 \in \mathbb{N}\ :\\
  &E\left( \left[ R_{\hat{n}}(\theta) \right]^{l_1} \left[R_{\hat{m}}(\theta) \right]^{l_2}\left[R_{\hat{n}}(\theta) \right]^{l_3}, R_{\hat{n}}(\alpha)R_{\hat{m}}(\beta)R_{\hat{n}}(\gamma)\right)\\
&\overset{\eqref{eq:ApproxErrorSum}}{\leq} E\left( \left[ R_{\hat{n}}(\theta) \right]^{l_1} , R_{\hat{n}}(\alpha)\right) + E\left( \left[R_{\hat{m}}(\theta) \right]^{l_2},R_{\hat{m}}(\beta)\right)+E\left( \left[R_{\hat{n}}(\theta) \right]^{l_3},R_{\hat{n}}(\gamma)\right)\\
&\overset{\eqref{eq:errorRot}}{<} 3 \delta.
\end{split}
\end{align}
That is, any single-qubit gate $U$ can be approximated to arbitrary accuracy by concatenation of Hadamard and $T$-gates.

\begin{remarks}\leavevmode
\begin{enumerate}[1)]
\item Scaling of Resources: 
\begin{itemize}
\item In Sec.~\ref{sec:twolevelbyCNOT} we found that $\mathcal{O}(n^24^n)$ single qubit and $CNOT$ gates are required to implement an arbitrary $U$.
\item By the \underline{Solovay Kitaev theorem} (see Appendix 3 in \cite{Nielsen-QC-2011}), an arbitrary single-qubit gate $U$ can be approximated to an accuracy $\epsilon$ using $\mathcal{O}([\log_2(1/\epsilon)]^c)$ gates from the discrete set (Hadamard, and $T$-gate), where $c\approx 2$. 
\end{itemize}
Altogether, approximately $n^24^n[\log_2(1/\epsilon)]^c$ gates from the discrete set are required. This is exponential in $n$. 

\item In general, an exponential overhead in $n$ has to be anticipated for implementations of arbitrary $n$-qubit unitaries. 

\item This does not address the problem which families of unitary operations can be computed efficiently in the quantum circuit model. 
\end{enumerate}
\end{remarks}

Before we discuss quantum algorithms, we need one more ingredient.

\section{The quantum Fourier transform} 
The quantum Fourier transform is the key ingredient for many interesting quantum algorithms. It enables an efficient realization of quantum phase estimation, which is the key for several other interesting problems including the order-finding problem and the factoring problem (i.e., Shor's algorithm). 

\begin{definition}
The \underline{quantum Fourier transform} (QFT) on an orthonormal basis $\{\ket{0},\dots,\ket{N-1}\}$ is defined to be a linear operator $\mathcal{F}$ acting as
\begin{align}\label{eq:F}
\mathcal{F}\ket{j}=\frac{1}{\sqrt{N}} \sum_{k=0}^{N-1} e^{\im 2\pi j k /N} \ket{k}.
\end{align} 
\end{definition}

\begin{remarks}\leavevmode
\begin{enumerate}[1)]
\item From Eq.~\eqref{eq:F} we get $\bra{k}\mathcal{F}\ket{j}=\frac{1}{\sqrt{N}} e^{\im 2\pi j k /N}$, such that the outer product representation [see Eq.~\eqref{eq:outerprod}] of $\mathcal{F}$ reads
\begin{align}
\mathcal{F}=\sum_{j,k=0}^{N-1}\frac{1}{\sqrt{N}} e^{\im 2\pi jk/N} \ketbra{k}{j}.
\end{align}
\item $\mathcal{F}$ is unitary, i.e., $\mathcal{F}^\dagger \mathcal{F}=\unit$, where 
\begin{align}\label{eq:Fdagger}
\mathcal{F}^\dagger \ket{j}=\frac{1}{\sqrt{N}} \sum_{k=0}^{N-1} e^{-\im 2\pi j k /N} \ket{k}.
\end{align}
\begin{proof}
Exercise.
\end{proof}
\item Since $\mathcal{F}$ is unitary, it represents a basis transformation, and we denote the new basis states by
\begin{align}
\ket{p_j}=\mathcal{F}\ket{j} = \frac{1}{\sqrt{N}} \sum_{k=0}^{N-1} e^{\im 2\pi kj/N} \ket{k}.
\end{align}
\end{enumerate}
\end{remarks}

In the following we consider the QFT on a register of $n$ qubits. Hence, $N=2^n$, and $\ket{0},\dots,\ket{2^n-1}$ are the computational basis states.

\subsection{Binary representation}
In order to rewrite $\mathcal{F}\ket{j}$ from Eq.~\eqref{eq:F}, we now introduce some notation: So far we already used that the basis state labels $j\in\{0,\dots,2^n-1\}$, i.e., 
\begin{align}
j=j_{n-1}2^{n-1}+j_{n-2}2^{n-2} + \dots + j_02^0=\sum_{l=0}^{n-1} j_l 2^l,
\end{align}
with $j_l\in\{0,1\}$, can be expressed in \underline{binary representation} as
\begin{align}\label{eq:binrep}
j=j_{n-1}j_{n-2}\dots j_0.
\end{align}
For example, for $n=3$ we have $6=1\times 2^2+1\times 2^1+0\times 2^0$, which has binary representation $110$. Now consider
\begin{align}\label{eq:binfrac}
\begin{split}
\frac{j}{2^l}&= \sum_{q=0}^{n-1} j_q \frac{2^q}{2^l} \\
&= \underbrace{\sum_{q=0}^{l-1} \frac{j_q}{2^{l-q}} }_{l>q}+ \underbrace{\sum_{q=l}^{n-	1} j_q 2^{q-l}}_{l\leq q}\\
&=\underbrace{\frac{j_0}{2^l}+\frac{j_1}{2^{l-1}}+\dots + \frac{j_{l-1}}{2}  }_{<1} + \underbrace{j_l 2^0 + j_{l+1} 2^1 + \dots + j_{n-1}2^{n-1-l}}_{\in\mathbb{N}} \\
&=: 0.j_{l-1}j_{l-2}\dots j_0 + j_{n-1}j_{n-2}\dots j_l\\
&=:j_{n-1}j_{n-2}\dots j_l.j_{l-1}j_{l-2}\dots j_0,
\end{split}
\end{align}
where we used the binary representation~\eqref{eq:binrep} and introduced the notation of the \underline{binary fraction} $0.j_{l-1}j_{l-2}\dots j_0$.

\subsection{Rewriting the quantum Fourier transform}
We now rewrite the quantum Fourier transform~\eqref{eq:F} using the binary representation:
\begin{align}\label{eq:Frewrite}
\begin{split}
\mathcal{F}\ket{j}&=\frac{1}{\sqrt{2^n}} \sum_{k=0}^{2^n-1} e^{\im 2\pi j k/2^n} \ket{k} \\
&=\frac{1}{\sqrt{2^n}} \sum_{k_0,\dots,k_{n-1}=0}^1 e^{\im 2\pi j \sum_{l=0}^{n-1} k_l/2^{n-l}} \ket{k_{n-1}k_{n-2}\dots k_0}\\
&=\frac{1}{\sqrt{2^n}}  \sum_{k_0,\dots,k_{n-1}=0}^1\left( \prod_{l=0}^{n-1} e^{\im 2\pi j k_l/2^{n-l}}  \right) \ket{k_{n-1}} \otimes \ket{k_{n-2}} \otimes \dots \otimes \ket{k_0}\\
&= \frac{1}{\sqrt{2^n}} \left(  \sum_{k_{n-1}=0}^1 e^{\im 2\pi j k_{n-1}/2^1} \ket{k_{n-1}}  \right) \otimes \dots \otimes \left(  \sum_{k_0=0}^1 e^{\im 2\pi j k_0/2^n} \ket{k_0}  \right)\\
&= \frac{1}{\sqrt{2^n}} \bigotimes_{l=1}^n \sum_{k_{n-l}=0}^1 e^{\im 2\pi j k_{n-l}/2^l} \ket{k_{n-l}}\\
&=\frac{1}{\sqrt{2^n}} \bigotimes_{l=1}^n \left(\ket{0} + e^{\im 2\pi j/2^l}  \ket{1} \right).
\end{split}
\end{align}
By inserting $j/2^l$ from Eq.~\eqref{eq:binfrac}, we see that $j_{n-1}j_{n-2}\dots j_l \in \mathbb{N}$ leads to integer multiples of $2\pi$ in~\eqref{eq:Frewrite}, hence can be dropped. We are thus left with
\begin{align}\label{eq:Fprod}
\begin{split}
\mathcal{F}\ket{j}&=\frac{1}{\sqrt{2^n}} \bigotimes_{l=1}^n \left(\ket{0} + e^{\im 2\pi 0.j_{l-1}j_{l-2}\dots j_0}  \ket{1} \right)\\
&=\frac{1}{\sqrt{2^n}} \left(\ket{0} + e^{\im 2\pi 0.j_0}  \ket{1} \right)\otimes \left(\ket{0} + e^{\im 2\pi 0.j_1j_0}  \ket{1} \right)\otimes \dots \otimes \left(\ket{0} + e^{\im 2\pi 0.j_{n-1}j_{n-2}\dots j_0}  \ket{1} \right).
\end{split}
\end{align}
Equation~\eqref{eq:Fprod} is a product representation (i.e. a separable state) of all $n$ qubits. Each single qubit is in a balanced superposition of $\ket{0}$ and $\ket{1}$, with an additional phase attached to $\ket{1}$, depending on the index of the qubit. Since a balanced superposition of a single qubit state is generated by the Hadamard gate from Eq.~\eqref{eq:Hgate}, and phase shifts on $\ket{1}$ by the phase shift gate from Eq.~\eqref{eq:Pgate}, Eq.~\eqref{eq:Fprod} suggests that there exists an efficient quantum circuit of the QFT in terms of these gates. 

\subsection{Circuit of the quantum Fourier transform}
We now use Eq.~\eqref{eq:Fprod} to find a quantum circuit for the QFT. To this end, let us define the phase gate 
\begin{align}\label{eq:RgateQFT}
R_k=P_{2\pi/2^k}=\begin{pmatrix}
1&0\\0&e^{\im 2\pi /2^k}
\end{pmatrix}.
\end{align}
The circuit implementing the QFT then looks as follows:

\begin{align}\label{cir:QFT}
\scalemath{0.7}{
\begin{tikzcd}
\lstick{$\ket{j_{n-1}}$} & \gate{H} \gategroup[wires=5,steps=5,style={dashed, rounded corners,inner sep=2pt}]{$(1)$}& \gate{R_2} &\ \ldots\ \qw & \gate{R_{n-1}} & \gate{R_n} & \qw \gategroup[wires=5,steps=8,style={dashed, rounded corners, inner sep=2pt}]{$(2)-(n)$} &\ \ldots\ \qw &\qw &\qw &\ \ldots\ \qw &\qw &\qw &\qw & \swap{4} \gategroup[wires=5,steps=3,style={dashed, rounded corners, inner sep=2pt}]{$(SWAP)$}& \qw &\ \ldots\ \qw \\
\lstick{$\ket{j_{n-2}}$} & \qw & \ctrl{-1} & \ \ldots\ \qw &  \qw & \qw & \gate{H} & \ \ldots\ \qw &\gate{R_{n-2}} & \gate{R_{n-1}} & \ \ldots\ \qw &\qw &\qw &\qw &\qw & \swap{2} & \ \ldots\ \qw\\
\lstick{\vdots}\wave &&&&&&&&&&&&&&&&\\
\lstick{$\ket{j_1}$} & \qw & \qw &\ \ldots\ \qw & \ctrl{-3} & \qw & \qw &\ \ldots\ \qw & \ctrl{-2} & \qw &\ \ldots\ \qw & \gate{H} &  \gate{R_2} & \qw &\qw &\targX{}&\ \ldots\ \qw\\
\lstick{$\ket{j_0}$} & \qw & \qw &\ \ldots\ \qw &  \qw & \ctrl{-4} &  \qw &\ \ldots\ \qw &\qw & \ctrl{-3} & \ \ldots\ \qw &\qw &\ctrl{-1} &\gate{H} & \targX{} & \qw &\ \ldots\ \qw
\end{tikzcd}.}
\end{align}

\subsubsection*{Part $(1)$}
Starting with $\ket{j}=\ket{j_{n-1}j_{n-2}\dots j_0}=\ket{j_{n-1}}\otimes \dots \otimes \ket{j_0}$ and acting with $H$ on the first qubit yields [recall $H\ket{0}=(\ket{0}+\ket{1})/\sqrt{2}$ and $H\ket{1}=(\ket{0}-\ket{1})/\sqrt{2}$]
\begin{align}
\begin{split}
H\ket{j_{n-1}} \otimes \ket{j_{n-2} j_{n-3} \dots j_0} &= \begin{cases} 
\frac{1}{\sqrt{2}}\left(\ket{0} + \ket{1} \right) \otimes \ket{j_{n-2} j_{n-3} \dots j_0} \quad &\text{if }\  j_{n-1}=0\\
\frac{1}{\sqrt{2}}\left(\ket{0} - \ket{1} \right) \otimes \ket{j_{n-2} j_{n-3} \dots j_0} \quad &\text{if }\  j_{n-1}=1
 \end{cases}\\
 &=\frac{1}{\sqrt{2}}\left(\ket{0} +e^{\im 2\pi 0.j_{n-1}} \ket{1} \right) \otimes \ket{j_{n-2} j_{n-3} \dots j_0} .
 \end{split}
\end{align}
Next, the controlled phase shift on the first qubit conditioned on the second quit, $C^{(1,2)}(R_2)$, results in 
\begin{align}
\begin{split}
&\left[ C^{(1,2)}(R_2)\frac{1}{\sqrt{2}}\left(\ket{0} +e^{\im 2\pi 0.j_{n-1}} \ket{1} \right) \otimes \ket{j_{n-2}} \right] \otimes \ket{j_{n-3} \dots j_0} \\
&=(R_2)^{j_{n-2}} \frac{1}{\sqrt{2}}\left(\ket{0} +e^{\im 2\pi 0.j_{n-1}} \ket{1} \right)  \otimes \ket{j_{n-2} j_{n-3} \dots j_0} \\
&\overset{\eqref{eq:RgateQFT}}{=} \frac{1}{\sqrt{2}}\left(\ket{0} +e^{\im 2\pi 0.j_{n-1}} e^{\im 2\pi j_{n-2}/2^2}\ket{1} \right) \otimes \ket{j_{n-2} j_{n-3} \dots j_0}  \\
&\overset{\eqref{eq:binfrac}}{=} \frac{1}{\sqrt{2}}\left(\ket{0} +e^{\im 2\pi 0.j_{n-1}j_{n-2}} \ket{1} \right)  \otimes \ket{j_{n-2} j_{n-3} \dots j_0} .
\end{split}
\end{align}
All other controlled phase gates act similar, resulting in
\begin{align}
\frac{1}{\sqrt{2}}\left(\ket{0} +e^{\im 2\pi 0.j_{n-1}j_{n-2}\dots j_0} \ket{1} \right) \otimes \ket{j_{n-2} j_{n-3} \dots j_0} .
\end{align}

\subsubsection*{Part $(2)$-$(n)$}
Repetition of the same procedure on quits $2$ to $n$ (i.e., first Hadamard, then controlled phase gates) produces the state
\begin{align}\label{eq:Part1n}
\frac{1}{\sqrt{2}}\left(\ket{0} +e^{\im 2\pi 0.j_{n-1}j_{n-2}\dots j_0} \ket{1} \right) \otimes \frac{1}{\sqrt{2}}\left(\ket{0} +e^{\im 2\pi 0.j_{n-2}\dots j_0} \ket{1} \right)\otimes \dots \otimes \frac{1}{\sqrt{2}}\left(\ket{0} +e^{\im 2\pi 0. j_0} \ket{1} \right).
\end{align}

\subsubsection*{Part $(SWAP)$}
By comparing~\eqref{eq:Part1n} with~\eqref{eq:Fprod}, we see that both states coincide up to the ordering of the qubits. Hence, by swapping the $k$-th with the $(n+1-k)$-th qubit (i.e., reversing the order of the qubits), we arrive at the desired output. 

\subsubsection*{How many gates does the circuit use?}
\begin{enumerate}[\hspace{60pt}]
\item[(1):] $1$ Hadamard + $(n-1)$ phase gates = $n$ gates
\item[(2):] $1$ Hadamard + $(n-2)$ phase gates = $n-1$ gates
\item[$\vdots$]
\item[($n$):] $1$ Hadamard = $1$ gate
\item[(SWAP):] at most $n/2$ gates. 
\end{enumerate}
In total we have $n+(n-1)+\dots+1+n/2=n(n+1)2+n/2=n(n+2)/2$ gates. Since by~\eqref{eq:SWAP} and~\eqref{cir:contrU} each $SWAP$ gate and each controlled unitary (with a single control) requires at most a constant number of gates, the circuit~\eqref{cir:QFT} of the QFT requires $\mathcal{O}(n^2)$ elementary gates.

\section{Some quantum algorithms} 
There are three approaches to study the difference between the capabilities of classical and quantum computers:
\begin{enumerate}[1)]
\item\textbf{Relativized exponential speedup:} Given a quantum black box (also called oracle), which performs an a priori unknown unitary transformation, there are quantum algorithms which can analyze what the black box does with (exponential) speedup compared to a classical analysis. Example: Simon's algorithm (see Sec.~\ref{sec:Simon}).

\item\textbf{Exponential speedup for apparently hard problems:} Quantum algorithms that solve a problem in polynomial time, where the problem appears to be hard classically. Example: Shor's factoring algorithm (see Sec.~\ref{sec:PhaseEstimation}).

\item\textbf{Nonexponential speedup:} Quantum algorithms that are demonstrably faster than the best classical algorithm, but not exponentially faster. Example: Grover's search algorithm (see Sec.~\ref{sec:Grover}).
\end{enumerate}

\subsection{Deutsch's algorithm}
Deutsch's algorithm is one of the simplest algorithms which demonstrates how quantum parallelism (i.e., the superposition principle) can outperform classical algorithms. Suppose we are given a black box (called \underline{oracle}), which evaluates a function $f:\{0,1\}\rightarrow \{0,1\}$ with a one bit domain and range according to the unitary
\begin{align}\label{eq:Deutschf}
C(U_f)\ket{x}\ket{y}=\ket{x}\ket{y\oplus f(x)}.
\end{align}
The task is to decide whether the function is constant (i.e., $f(0)=f(1)$) or balanced (i.e., $f(0)\neq f(1)$). 

\subsubsection*{Classical input}
If we are restricted to classical inputs, to get the answer we must access the box (i.e., call the oracle) twice, for $x=0$ and $x=1$.

\subsubsection*{Quantum input}
If we can input coherent superpositions, we only need to access the box (i.e., call the oracle) once. To see this, consider the circuit 
\begin{align}\label{cir:Deutsch}
\begin{tikzcd}
 \lstick{$\ket{0}$}  & \qw  & \gate{H}& \ctrl{1}& \gate{H}  & \meter{} \\
\lstick{$\ket{0}$} & \gate{X} & \gate{H} & \gate{U_f}& \qw & \qw
\end{tikzcd}
=
\begin{tikzcd}
 \lstick{$\ket{0}$}  & \qw  & \gate{H}& \gate[wires=2][2cm]{U_f} \gateinput{$x$} \gateoutput{$x$} & \gate{H}  & \meter{} \\
\lstick{$\ket{0}$} & \gate{X} & \gate{H} & \qw \gateinput{$y$} \gateoutput{$y\oplus f(x)$} & \qw & \qw
\end{tikzcd}.
\end{align}
The $X$ and the first Hadamard gates act on the input as
\begin{align}
\ket{0}\ket{0} \xrightarrow{ \unit \otimes X} \ket{0} X\ket{0} =\ket{0}\ket{1}
\end{align}
and
\begin{align}\label{eq:DeutschHH}
\ket{0}\ket{1} \xrightarrow{ H\otimes H} H\ket{0}H\ket{1}= \frac{1}{\sqrt{2}}(\ket{0}+\ket{1}) \frac{1}{\sqrt{2}}(\ket{0}-\ket{1}).
\end{align}
Next, note that the controlled unitary acts on $\ket{x}\frac{1}{\sqrt{2}}(\ket{0}-\ket{1})$ as 
\begin{align}
\begin{split}
C(U_f)\ket{x} \frac{1}{\sqrt{2}}\left(\ket{0}-\ket{1}\right)&\overset{\eqref{eq:Deutschf}}{=} \ket{x} \frac{1}{\sqrt{2}}\left(\ket{0\oplus f(x)}-\ket{1\oplus f(x)}\right)\\
&= \begin{cases}
\phantom{-}\ket{x} \frac{1}{\sqrt{2}}\left(\ket{0}-\ket{1}\right) \quad \text{if }\ f(x)=0\\
-\ket{x} \frac{1}{\sqrt{2}}\left(\ket{0}-\ket{1}\right) \quad \text{if }\ f(x)=1
\end{cases}\\
&=(-1)^{f(x)} \ket{x} \frac{1}{\sqrt{2}}\left(\ket{0}-\ket{1}\right).
\end{split}
\end{align}
Using this, Eq.~\eqref{eq:DeutschHH} transforms as
\begin{align}
\frac{1}{\sqrt{2}}\left(\ket{0}+\ket{1}\right)\frac{1}{\sqrt{2}}\left(\ket{0}-\ket{1}\right) \xrightarrow{C(U_f)} \frac{1}{\sqrt{2}}\left((-1)^{f(0)}\ket{0}+(-1)^{f(1)}\ket{1}\right)\frac{1}{\sqrt{2}}\left(\ket{0}-\ket{1}\right).
\end{align}
Now we can forget about the second qubit, and calculate the action of the final Hadamard on the first qubit alone, 
\begin{align}
\begin{split}
H\frac{1}{\sqrt{2}}\left((-1)^{f(0)}\ket{0}+(-1)^{f(1)}\ket{1}\right) &=\frac{1}{2}\left[ (-1)^{f(0)} \left( \ket{0}+\ket{1} \right)  + (-1)^{f(1)} \left( \ket{0}-\ket{1} \right) \right]\\
&=\frac{1}{2}\left[ \left( (-1)^{f(0)}+(-1)^{f(1)} \right) \ket{0} + \left( (-1)^{f(0)}-(-1)^{f(1)} \right) \ket{1} \right]\\
=&\begin{cases}
\ket{0} \quad \text{if }\ f(0)=f(1)\\
\ket{1} \quad \text{if }\ f(0)\neq f(1).
\end{cases}
\end{split}
\end{align}
Hence, we only used a single call of the oracle, and, by measuring the first qubit, we find with certainty whether $f$ is constant or balanced. 

This computational advantage is due to interference: The first Hadamard gates put the state into a superposition, such that $f(x)$ is evaluated in parallel for $x=0$ and $x=1$, and the recombination by the second Hadamard gate on the first qubit then leads to an interference of the alternatives associated with $f(0)$ and $f(1)$.

\subsection{Simon's algorithm}\label{sec:Simon}
Simon's algorithm provides an example for a raltivized separation between quantum and classical complexity, i.e., that quantum computers can solve specific problems exponentially faster than classical computers. 
Consider a binary function $f:\{0,1\}^n \rightarrow \{0,1\}^n$, which is evaluated by an oracle and promised to satisfy $f(x)=f(x\oplus s)$ for all $x\in\{0,1\}^n$ and an unknown bit string $s\in\{0,1\}^n$, with $s\neq 0^n$ and $f(x)\neq f(y)$ for $y\neq x\oplus s$. Here $\oplus$ denotes the bitwise addition modulo two. For example, if $x\oplus s=y$, then 
\begin{align}
y_j=(x_j+s_j)\mod 2
\end{align} 
for all $j\in\{0,\dots,n-1\}$. The problem is to find $s$. 
\begin{example}\ \\
\begin{center}
\begin{tabular}{c|c|c|c|c|c|c|c|c}
$x$&000&001&010&011&100&101&110&111\\ \hline
$f(x)$&101&010&000&110&000&110&101&010
\end{tabular}
\end{center}
Since $f(x)=f(x\oplus s)$, and $f(000)=f(110)$, we find that $s=110$. 
\end{example}

\subsubsection*{Classical input}
First note that by $s\neq 0$ and $x\oplus s \oplus s=x$, each outcome appears exactly twice. Hence, we must find two times the same output to deduce $s$. The maximum number of oracle calls after which we can find two equal outcomes is 
\begin{align}
2^n/2+1=2^{n-1}+1.
\end{align}
That is, we need an exponential number of oracle calls to find $s$. To see that this exponential scaling holds in general (not only for the maximum number of calls), suppose that we call the oracle $2^{n/4}$ times, such that there are less than $(2^{n/4})^2$ possible pairs of outcomes $(f(x),f(y))$. For each pair, the probability that $f(x)=f(y)$, i.e., $x\oplus y=s$, is $2^{-n}$. Hence, the probability of successfully finding $s$ is less than
\begin{align}
2^{-n}(2^{n/4})^2=2^{-n/2}.
\end{align}
That is, with an exponential number of calls, the probability to find $s$ is still exponentially small. Indeed, the best known algorithm requires at least $\mathcal{O}(2^{n/2})$ oracle calls. 

\subsubsection*{Quantum input}
Using a quantum algorithm, we now show that we can find $s$ by $\mathcal{O}(n)$ oracle calls. This is an exponential speed up compared to classical computation. Let's consider the quantum circuit
\begin{align}
\begin{tikzcd}
 \lstick{$\ket{0}$}  & \gate{H^{\otimes n}} \qwbundle{} & \gate[wires=2][2cm]{U_f} \gateinput{$x$} \gateoutput{$x$} & \qw & \gate{H^{\otimes n}}  & \meter{} \\
\lstick{$\ket{0}$} &  \qw \qwbundle{} & \qw \gateinput{$y$} \gateoutput{$y\oplus f(x)$} & \meter{} & &
\end{tikzcd}
\end{align}
where $U_f$ performs the function evaluation (quantum black box)
\begin{align}
U_f(\ket{x}\ket{0})=\ket{x}\ket{f(x)}.
\end{align}
After the first Hadamards, the state is
\begin{align}
\ket{0}\ket{0} \xrightarrow{H^{\otimes n}\otimes \unit} H^{\otimes n}\ket{0} \ket{0} = \frac{1}{\sqrt{2^n}} \sum_{x=0}^{2^n-1} \ket{x} \ket{0},
\end{align}
and after the oracle call (i.e., after applying $U_f$), we have
\begin{align}
\frac{1}{\sqrt{2^n}} \sum_{x=0}^{2^n-1} \ket{x} \ket{0} \xrightarrow{U_f} \frac{1}{\sqrt{2^n}} \sum_{x=0}^{2^n-1} \ket{x} \ket{f(x)}.
\end{align}
Now suppose that the measurement of the second register yields $f(w)=f(w \oplus s)$. Then the first register is in the state 
\begin{align}
\frac{1}{\sqrt{2}}\left(\ket{w} + \ket{w\oplus s} \right).
\end{align}
Performing again Hadamard gates in the first register then leads to 
\begin{align}\label{eq:SimonH}
\frac{1}{\sqrt{2}}\left(\ket{w} + \ket{w\oplus s} \right) \xrightarrow{H^{\otimes n}} \frac{1}{\sqrt{2}}\left(H^{\otimes n}\ket{w} + H^{\otimes n}\ket{w\oplus s} \right) .
\end{align}
Next, note that
\begin{align}
\begin{split}
H^{\otimes n}\ket{w}&=H\ket{w_{n-1}} \otimes \dots \otimes H\ket{w_0} \\
&=\bigotimes_{j=1}^n H\ket{w_{n-j}}\\
&=\bigotimes_{j=1}^n \frac{1}{\sqrt{2}} \left( \ket{0} + (-1)^{w_{n-j}} \ket{1} \right)\\
&= \frac{1}{\sqrt{2^n}} \bigotimes_{j=1}^n \sum_{x_{n-j}=0}^1 (-1)^{x_{n-j}w_{n-j}} \ket{x_{n-j}}\\
&=\frac{1}{\sqrt{2^n}} \sum_{x_{n-1},\dots,x_0=0}^1  (-1)^{x_{n-1}w_{n-1}+\dots+x_0w_0} \ket{x_{n-j}\dots x_0}\\
&=\frac{1}{\sqrt{2^n}} \sum_{x=0}^{2^n-1} (-1)^{x\cdot w} \ket{x},
\end{split}
\end{align}
where
\begin{align}
x\cdot w= x_{n-1}w_{n-1}+x_{n-2}w_{n-2}+\dots + x_0w_0.
\end{align}
Using this in Eq.~\eqref{eq:SimonH}, the output of the first register becomes
\begin{align}
\begin{split}
\frac{1}{\sqrt{2}}\left(H^{\otimes n}\ket{w} + H^{\otimes n}\ket{w\oplus s} \right) &=\frac{1}{\sqrt{2}}\left(\frac{1}{\sqrt{2^n}} \sum_{x=0}^{2^n-1} (-1)^{x\cdot w} \ket{x}+ \frac{1}{\sqrt{2^n}} \sum_{x=0}^{2^n-1} (-1)^{x\cdot (w\oplus s)} \ket{x} \right) \\
&=\frac{1}{\sqrt{2^{n+1}}} \sum_{x=0}^{2^n-1} \left[  (-1)^{x\cdot w} +  (-1)^{x\cdot (w\oplus s)}\right] \ket{x}.
\end{split}
\end{align}
If the measurement of the first register yields outcome $z$, we must have
\begin{align}
\begin{split}
(-1)^{z\cdot w} &= (-1)^{z\cdot (w\oplus s)}\\
\Leftrightarrow  (z\cdot w)\mod 2 &= \left( z \cdot (w \oplus s)\right)\mod 2\\
\Leftrightarrow (z\cdot w)\mod 2 &= \left( z \cdot w \oplus z \cdot s \right)\mod 2\\
\Leftrightarrow  0 &= \left(  z \cdot s \right)\mod 2.
\end{split}
\end{align}
That is, by calling the oracle one time, we find $z\in\{0,1\}^n$ such that $z\cdot s=0 \mod 2$. 

Now suppose we call the oracle more often, such that the $j$-th call returns $z_j$. Given $z_1\neq 0^n$, by $z_1\cdot s=0 \mod 2$, we can rule out half of all possibilities for $s\in\{0,\dots,2^n-1\}$. Hence, we are left with $2^n/2$ numbers for $s$. By measuring $z_2\neq z_1$ (again $z_2 \neq 0^n$), we can further rule out half of the remaining numbers for $s$, such that we are left with $2^n/2^2$ numbers. Proceeding similar, by measuring $(n-1)$ linearly independent $z_j\neq 0^n$, we are left with $2^n/2^{n-1}=2$ different possibilities for $s$. One of them is $s=0^n$. However, since we know that $s\neq 0^n$, we found $s$. 

\begin{example}
Let us return to the example of $n=3$ qubits as discussed in the beginning of this section. Suppose that we first measure $z_1=001\neq 0^n$. In this case, we find
\begin{align*}
(z_1 \cdot 000) \mod 2&=0 \hspace{2cm} (z_1 \cdot 100) \mod 2=0 \\
(z_1 \cdot 001) \mod 2&=1 \hspace{2cm} (z_1 \cdot 101) \mod 2=1 \\
(z_1 \cdot 010) \mod 2&=0  \hspace{2cm} (z_1 \cdot 110) \mod 2=0 \\
(z_1 \cdot 011) \mod 2&=1  \hspace{2cm} (z_1 \cdot 111) \mod 2=1 ,
\end{align*}
such that $s\in\{000,010,100,110\}$, i.e., we are left with $2^n/2=4$ possibilities for $s$. Next, suppose we measure $z_2=110 \neq 0^n$, for which we find
\begin{align*}
(z_2 \cdot 000) \mod 2&=0  \hspace{2cm} (z_2 \cdot 100) \mod 2=1 \\
(z_2 \cdot 010) \mod 2&=1 \hspace{2cm} (z_2 \cdot 110) \mod 2=0.
\end{align*}
The possibilities for $s$ then reduce to $2^n/2^2=2$, namely $s\in\{000,110\}$. However, since we know that $s\neq 0^n$, we found $s=110$. 
\end{example}

In summary, we must call the oracle as long as we found $n-1$ linearly independent outcomes $z\neq 0^n$. This can be accomplished with a probability exponential close to $1$ by $\mathcal{O}(n)$ oracle calls. Hence, given the above quantum oracle, we can solve Simon's problem with a polynomial number of oracle calls by exploiting quantum superpositions, while an exponential number of oracle calls is required in the classical case.

\subsection{Phase estimation}\label{sec:PhaseEstimation}
Phase estimation is the key for many quantum algorithms, such as solving the order finding problem or the factoring problem (via Shor's algorithm). The problem of phase estimation is as follows: Suppose a unitary $U$ has an eigenvector $\ket{u}$ with eigenvalue $e^{\im 2\pi \varphi}$, where $\varphi \in [0,1)$ is unknown. The task is to find $\varphi$. 

Suppose we have oracles (i.e., black boxes) capable of preparing $\ket{u}$ and performing the controlled $U^j$ operation for non-negative integers $j$, and consider the circuit
\begin{align}\label{cir:Phase1}
\scalemath{0.9}{
\begin{tikzcd}
\lstick[wires=5]{$n$ qubits} & & \lstick{$\ket{0}$}& \gate{H} & \qw & \qw &  \qw & \ \ldots\ \qw & \ctrl{5} &\qw\\
&&\wave & & & & & & & &\\
& & \lstick{$\ket{0}$}&\gate{H} & \qw & \qw & \ctrl{3} & \ \ldots\ \qw &\qw &\qw\\
& & \lstick{$\ket{0}$}&\gate{H} & \qw & \ctrl{2} &  \qw & \ \ldots\ \qw &\qw &\qw\\
& & \lstick{$\ket{0}$}&\gate{H} & \ctrl{1} &\qw &  \qw & \ \ldots\ \qw &\qw &\qw\\
& & \lstick{$\ket{u}$}& \qwbundle{} &\gate{U^{2^0}} & \gate{U^{2^1}} & \gate{U^{2^2}} &\ \ldots\ \qw &\gate{U^{2^{n-1}}} & \qw
\end{tikzcd}}
=
\begin{tikzcd}
\lstick{$\ket{0}$}&\qw \qwbundle{n} &\gate{H^{\otimes n}} & \ctrl{1} & \qw \\
 \lstick{$\ket{u}$}&\qw \qwbundle{} &\qw & \gate{U^j} & \qw
\end{tikzcd}.
\end{align}
The Hadamards act as
\begin{align}
\ket{0}\ket{u} \xrightarrow{H^{\otimes n}\otimes \unit} H^{\otimes n}\ket{0}\ket{u} = \frac{1}{\sqrt{2^n}} \sum_{k=0}^{2^n-1} \ket{k}\ket{u},
\end{align}
and the controlled $U^{2^j}$ gates as
\begin{align}\label{eq:Phasestep1}
\begin{split}
\frac{1}{\sqrt{2^n}} \sum_{k=0}^{2^n-1} \ket{k}\ket{u} \xrightarrow{C(U^{2^j})} &\frac{1}{\sqrt{2^n}} \sum_{k=0}^{2^n-1} \ket{k}\left[ U^{2^{n-1}}\right]^{k_{n-1}} \dots \left[ U^{2^1}\right]^{k_1} \left[ U^{2^0}\right]^{k_0} \ket{u}\\
=&\frac{1}{\sqrt{2^n}} \sum_{k=0}^{2^n-1} \ket{k} e^{\im 2\pi \varphi 2^{n-1}k_{n-1}}\dots e^{\im 2\pi \varphi 2^1k_1} e^{\im 2\pi \varphi 2^0k_0}\ket{u}\\
=&\frac{1}{\sqrt{2^n}} \sum_{k=0}^{2^n-1} \ket{k} e^{\im 2\pi \varphi \sum_{j=1}^n 2^{n-j}k_{n-j}} \ket{u}\\
=&\frac{1}{\sqrt{2^n}} \sum_{k=0}^{2^n-1}  e^{\im 2\pi \varphi k} \ket{k}\ket{u}.
\end{split}
\end{align}
Next we consider an inverse Fourier transform $\mathcal{F}^\dagger$ acting on the first register, followed by a measurement. Hence, the circuit~\eqref{cir:Phase1} extends to
\begin{align}
\begin{tikzcd}
\lstick{$\ket{0}$}&\qw \qwbundle{n} &\gate{H^{\otimes n}} & \ctrl{1} & \gate{\mathcal{F}^\dagger} &\meter{} \\
 \lstick{$\ket{u}$}&\qw \qwbundle{} &\qw & \gate{U^j} & \qw &\qw
\end{tikzcd}.
\end{align}
To calculate the outcome, let us drop the second register in~\eqref{eq:Phasestep1} and act with $\mathcal{F}^\dagger$ from Eq.~\eqref{eq:Fdagger},
\begin{align}
\begin{split}
\mathcal{F}^\dagger\frac{1}{\sqrt{2^n}} \sum_{k=0}^{2^n-1}  e^{\im 2\pi \varphi k} \ket{k}&=\frac{1}{\sqrt{2^n}} \sum_{k=0}^{2^n-1}  e^{\im 2\pi \varphi k} \mathcal{F}^\dagger\ket{k}\\
&\overset{\eqref{eq:Fdagger}}{=} \frac{1}{2^n} \sum_{j,k=0}^{2^n-1} e^{\im 2\pi  k(\varphi-j/2^n)} \ket{j}.
\end{split}
\end{align}
Finally, the outcome $\ket{2^n \tilde{\varphi}}$, where $\tilde{\varphi}=0.\tilde{\varphi}_{n-1}\dots \tilde{\varphi}_0$ and $2^n \tilde{\varphi}=\tilde{\varphi}_{n-1}\dots \tilde{\varphi}_0$, appears with probability
\begin{align}\label{eq:PhaseEstDer1}
\begin{split}
p(2^n \tilde{\varphi})&=\abs{ \frac{1}{2^n} \sum_{j,k=0}^{2^n-1} e^{\im 2\pi  k(\varphi-j/2^n)} \bracket{2^n\tilde{\varphi}}{j} }^2\\
&=\abs{ \frac{1}{2^n} \sum_{k=0}^{2^n-1} e^{\im 2\pi  k(\varphi-\tilde{\varphi})} }^2\\
&=\abs{ \frac{1}{2^n} \sum_{k=0}^{2^n-1} \left[ e^{\im 2\pi  (\varphi-\tilde{\varphi})} \right]^k}^2.
\end{split}
\end{align}
Using the geometric series $\sum_{k=0}^n x^k=(1-x^{n+1})/(1-x)$ yields
\begin{align}\label{eq:PhaseEstDer2}
\begin{split}
p(2^n \tilde{\varphi})&=\abs{ \frac{1}{2^n} \ \frac{1-e^{\im 2\pi 2^n (\varphi-\tilde{\varphi})} }{ 1-e^{\im 2\pi  (\varphi-\tilde{\varphi})}   }}^2\\
&=\abs{ \frac{1}{2^n} \ \frac{e^{\im \pi 2^n (\varphi-\tilde{\varphi})} \left( e^{-\im \pi 2^n (\varphi-\tilde{\varphi})} -e^{\im \pi 2^n (\varphi-\tilde{\varphi})} \right)}{ e^{\im \pi  (\varphi-\tilde{\varphi})} \left(  e^{-\im \pi  (\varphi-\tilde{\varphi})}- e^{\im \pi  (\varphi-\tilde{\varphi})}  \right) }}^2\\
&=\abs{  \frac{\sin\left( 2^n \pi  (\varphi-\tilde{\varphi}) \right)}{2^n \sin\left(\pi(\varphi-\tilde{\varphi})\right)}  }^2.
\end{split}
\end{align}
First note that if $\varphi$ can exactly be expressed via $n$ bits, then the second line in~\eqref{eq:PhaseEstDer1} implies for $\tilde{\varphi}=\varphi$ that $p(2^n \varphi)=\left|\frac{1}{2^n} \sum_{k=0}^{2^n-1} \right|^2=1$. That is, we find $\varphi$ with certainty. On the other hand, if $\varphi$ cannot be expressed via $n$ bits, we can only find an $n$-bit approximation of $\varphi$. To see that we find with high probability an $n$-bit approximation $\tilde{\varphi}$ which is close to the optimal $n$-bit approximation $\tilde{\varphi}_\text{opt}$, note that $|\varphi-\tilde{\varphi}|<1$, such that $|\sin(\pi(\varphi-\tilde{\varphi}))| \leq |\pi (\varphi-\tilde{\varphi})|$, and Eq.~\eqref{eq:PhaseEstDer2} becomes 
\begin{align}\label{eq:PhaseEstDer3}
p(2^n \tilde{\varphi})&\geq \abs{ \frac{\sin\left( 2^n \pi (\varphi-\tilde{\varphi}) \right)}{2^n \pi (\varphi-\tilde{\varphi})}}^2= \text{sinc}^2\left( 2^n \pi(\varphi-\tilde{\varphi})\right).
\end{align}
Accordingly, $\tilde{\varphi}$ close to $\tilde{\varphi}_\text{opt}$ appears with high probability since we know that $\lim_{x\rightarrow 0}\text{sinc}(x)=1$. For example, for $\tilde{\varphi}=\tilde{\varphi}_\text{opt}$, we have $\abs{\varphi-\tilde{\varphi}_\text{opt}} \leq 1/2^{n+1}$ and, accordingly, 
\begin{align}
\Big|\sin( \pi \underbrace{2^n(\varphi-\tilde{\varphi}_\text{opt})}_{\in [-1/2,1/2]} ) \Big| \geq \abs{2^{n+1} (\varphi-\tilde{\varphi}_\text{opt})},
\end{align}
where we used $\abs{\sin(x)} \geq 2 \abs{x}/\pi$ for $x \in [-\pi/2,\pi/2]$. Using this in Eq.~\eqref{eq:PhaseEstDer3} then yields
\begin{align}
p(2^n \tilde{\varphi}_\text{opt})&\geq \abs{\frac{2^{n+1}(\varphi-\tilde{\varphi}_\text{opt})}{2^n \pi (\varphi-\tilde{\varphi}_\text{opt})}}^2=\frac{4}{\pi^2} \approx 0.405.
\end{align}
The connection of phase estimation to order and period finding will possibly be discussed in the Exercises.

\subsection{Grover's search algorithm}\label{sec:Grover}
Suppose we wish to search through an unstructured data base of $N=2^n$ elements which are labeled by $0,1,\dots,2^n-1$ (e.g. find the name associated with a given phone number in a phone book). Once you're suggested a solution (a name), it's easy to check its validity (numbers are sorted by names), but it's difficult to find the solution to start with. Classically, this takes $\mathcal{O}(N)$ operations (validity checks) to solve the problem. Using quantum mechanics L. K. Grover showed in 1996 that only $\mathcal{O}(\sqrt{N})$ operations are needed. This is a quadratic speedup. 

Consider a register of $n$ qubits to represent the data base with $N=2^n$ labels, and a single oracle qubit. The oracle verifies whether a given label is an element of the set of solutions (note that there can be more than a single solution). Hence, the speed of the algorithm is determined by the number of oracle calls. For $\ket{x}$ a state vector of the data base and $\ket{y}$ a state vector of the oracle, the action of the oracle is defined as
\begin{align}
\ket{x}\ket{y} \xrightarrow{\text{oracle}} \ket{x}\ket{y\oplus f(x)}
\end{align}
with 
\begin{align}
f(x)=\begin{cases}
0\quad \text{if $x$ is no solution} \\
1\quad \text{if $x$ is a solution.}
\end{cases}
\end{align}
For $\ket{y}=(\ket{0}-\ket{1})/\sqrt{2}$ this implies the mapping
\begin{align}
\ket{x}\ket{y} \xrightarrow{\text{oracle}} (-1)^{f(x)}\ket{x}\ket{y},
\end{align}
such that $\ket{y}$ can be disregarded [c.f. Deutsch's algorithm in~\eqref{cir:Deutsch}], and we can compactly write 
\begin{align}\label{eq:GroverO}
O: \ket{x} \rightarrow (-1)^{f(x)} \ket{x}.
\end{align}
The Grover algorithm then has the representation 
\begin{align}\label{cir:Grover}
\begin{tikzcd}
\lstick[wires=4]{data  base\\ ($n$ qubits)} & & \lstick{$\ket{0}$}& \qw& \gate[4,nwires={3}]{H^{\otimes n}}& \gate[5,nwires={3}]{G}& \gate[5,nwires={3}]{G}& \ \ldots\ \qw &\gate[5,nwires={3}]{G}&\qw  \\
& & \lstick{$\ket{0}$}&\qw& & & & \ \ldots\ \qw & &\qw  \\
& &\lstick{\vdots}  & \text{$\vdots$} & & & & \text{$\vdots$} & &  \\
& & \lstick{$\ket{0}$}&\qw& & & & \ \ldots\ \qw & &\qw  \\
\lstick[wires=1]{oracle} & & \lstick{$\ket{0}$}&\gate{X}&\gate{H} & & & \ \ldots\ \qw & &\qw 
\end{tikzcd}
\end{align}
with the Grover iteration $G$ given by
\begin{align}\label{cir:GroverIteration}
\begin{tikzcd}
\lstick[wires=4]{data  base\\ ($n$ qubits)} &\qw&\gate[5,nwires={3}]{G}&\qw& & &\qw& \gate[5,nwires={3}]{O} & \gate[4,nwires={3}]{H^{\otimes n}} &\gate[4,nwires={3}]{2\ketbra{0}{0}-\unit} &\gate[4,nwires={3}]{H^{\otimes n}} &\qw \\
 \lstick{ }& \qw&&\qw & & & \qw&&&&&\qw \\
 \lstick{ }& \text{$\vdots$}&&  & =& &\text{$\vdots$}&&&&& \\
 \lstick{ }& \qw&&\qw &&& \qw&&&&&\qw \\ 
\lstick[wires=1]{oracle} & \qw&& \qw& && \qw&&\qw& \qw& \qw&\qw
\end{tikzcd}.
\end{align}
Let us now describe how this algorithm works: Before the first Grover iteration, the state of the data base register is
\begin{align}\label{eq:GroverPsi}
\ket{\psi}=H^{\otimes n} \ket{0}=\left[\frac{1}{\sqrt{2}} (\ket{0}+\ket{1})\right]^{\otimes n}= \frac{1}{\sqrt{2^n}}\sum_{x=0}^{2^n-1} \ket{x}.
\end{align}
The register is thus in a balanced superposition of all states (labels of the data set), from which the subsequent Grover iterations $G$ have to distill the solutions. The Grover iteration $G$, decomposed in~\eqref{cir:GroverIteration}, starts with an oracle operation $O$ according to Eq.~\eqref{eq:GroverO}, which marks the solutions with a minus sign. The subsequent operation $H^{\otimes n}\left( 2 \ketbra{0}{0}-\unit\right)H^{\otimes n}$ only acts on the data set qubits, where $2 \ketbra{0}{0}-\unit$ garnishes all computational basis states except $\ket{0}$ with a minus sign. Together we have
\begin{align}\label{eq:GroverH00H}
\begin{split}
H^{\otimes n}\left( 2 \ketbra{0}{0}-\unit\right)H^{\otimes n}&=2H^{\otimes n}\ketbra{0}{0}H^{\otimes n} - H^{\otimes n}H^{\otimes n}\\
&\overset{\eqref{eq:GroverPsi}}{=}2\ketbra{\psi}{\psi}-\unit,
\end{split}
\end{align}
such that the Grover iteration~\eqref{cir:GroverIteration} performs the operation
\begin{align}\label{eq:GroverG}
G=\left( 2\ketbra{\psi}{\psi}-\unit\right) O.
\end{align}
In order to visualize the action of the Grover algorithm, first consider the action of~\eqref{eq:GroverH00H} on a general state $\ket{\phi}$,
\begin{align}
\left( 2\ketbra{\psi}{\psi}-\unit\right) \ket{\phi} = 2 \bracket{\psi}{\phi} \ket{\psi}- \ket{\phi},
\end{align}
which corresponds to a reflection of $\ket{\phi}$ with respect to $\ket{\psi}$:
\begin{center}
\includegraphics[width=0.7\linewidth]{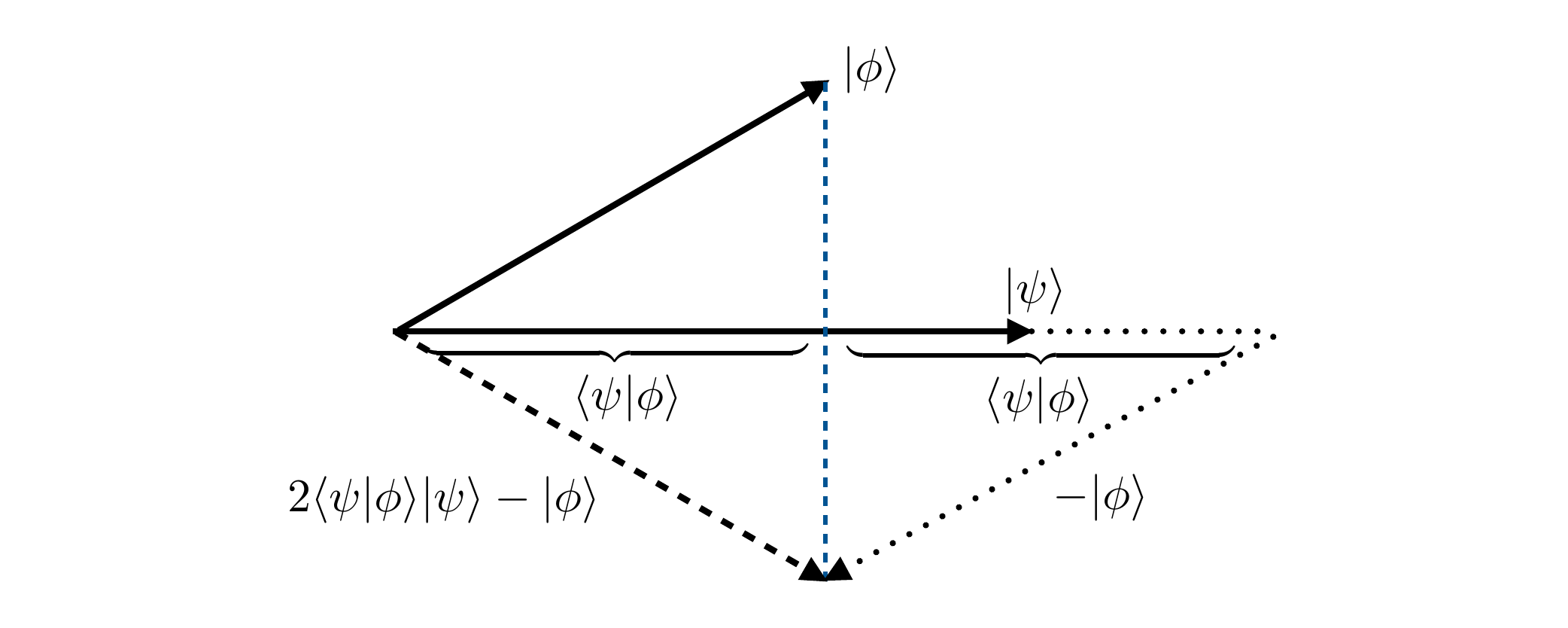}
\end{center}
Likewise also the oracle operation $O$ from Eq.~\eqref{eq:GroverO} can be understood as a reflection. To this end, decompose the Hilbert space into the subspace $S$, spanned by $M$ states $\{\ket{x}\}_{x\in S}$, which represent a solution, and the subspace $\bar{S}$, spanned by $N-M$ states $\{\ket{x}\}_{x\in \bar{S}}$, which are no solution. Then define the states $\ket{\alpha}$ and $\beta$ on these subspaces, 
\begin{align}\label{eq:Groverab}
\begin{split}
\ket{\alpha}&=\frac{1}{\sqrt{N-M}} \sum_{x \in \bar{S}} \ket{x},\\
\ket{\beta}&=\frac{1}{\sqrt{M}} \sum_{x \in S} \ket{x},
\end{split}
\end{align}
such that, by Eq.~\eqref{eq:GroverO},
\begin{align}
O(\alpha \ket{\alpha} + \beta \ket{\beta}) = \alpha \ket{\alpha} - \beta \ket{\beta}.
\end{align}
Geometrically, this is a reflection of $\alpha \ket{\alpha} + \beta \ket{\beta}$ with respect to $\ket{\alpha}$:
\begin{center}
\includegraphics[width=0.6\linewidth]{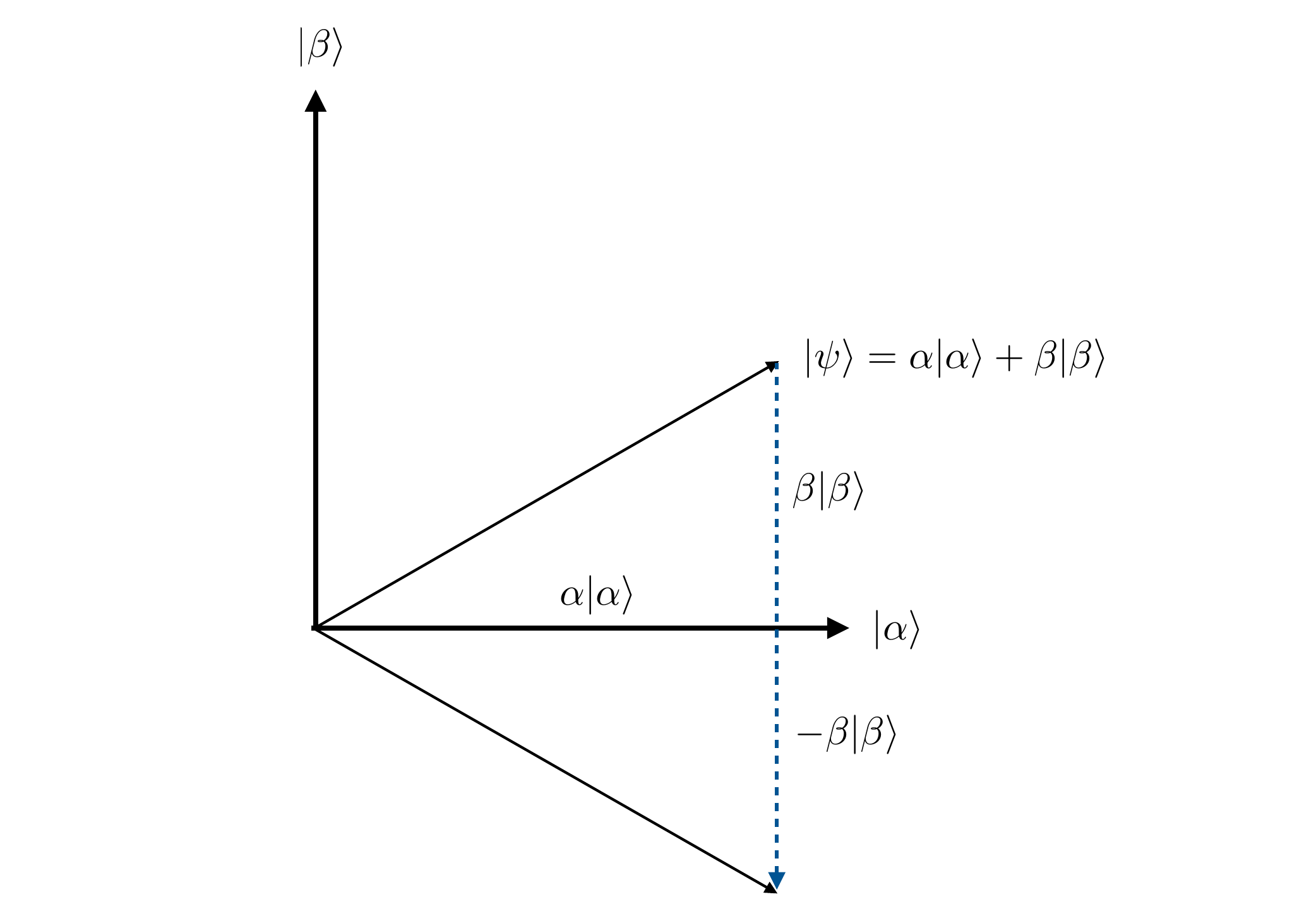}
\end{center}
Altogether, the two reflections $O$ and $2\ketbra{\psi}{\psi}-\unit$ acting on the state 
\begin{align}\label{eq:GroverPsiab}
\ket{\psi}\overset{\eqref{eq:GroverPsi}}{=}\frac{1}{\sqrt{2^n}}\sum_{x=0}^{2^n-1} \ket{x}\overset{\eqref{eq:Groverab}}{=} \sqrt{\frac{N-M}{N}} \ket{\alpha} + \sqrt{\frac{M}{N}} \ket{\beta} 
\end{align}
before the first Grover iteration [see Eq.~\eqref{eq:GroverPsi}] can geometrically be visualized as follows:
\begin{center}
\includegraphics[width=0.6\linewidth]{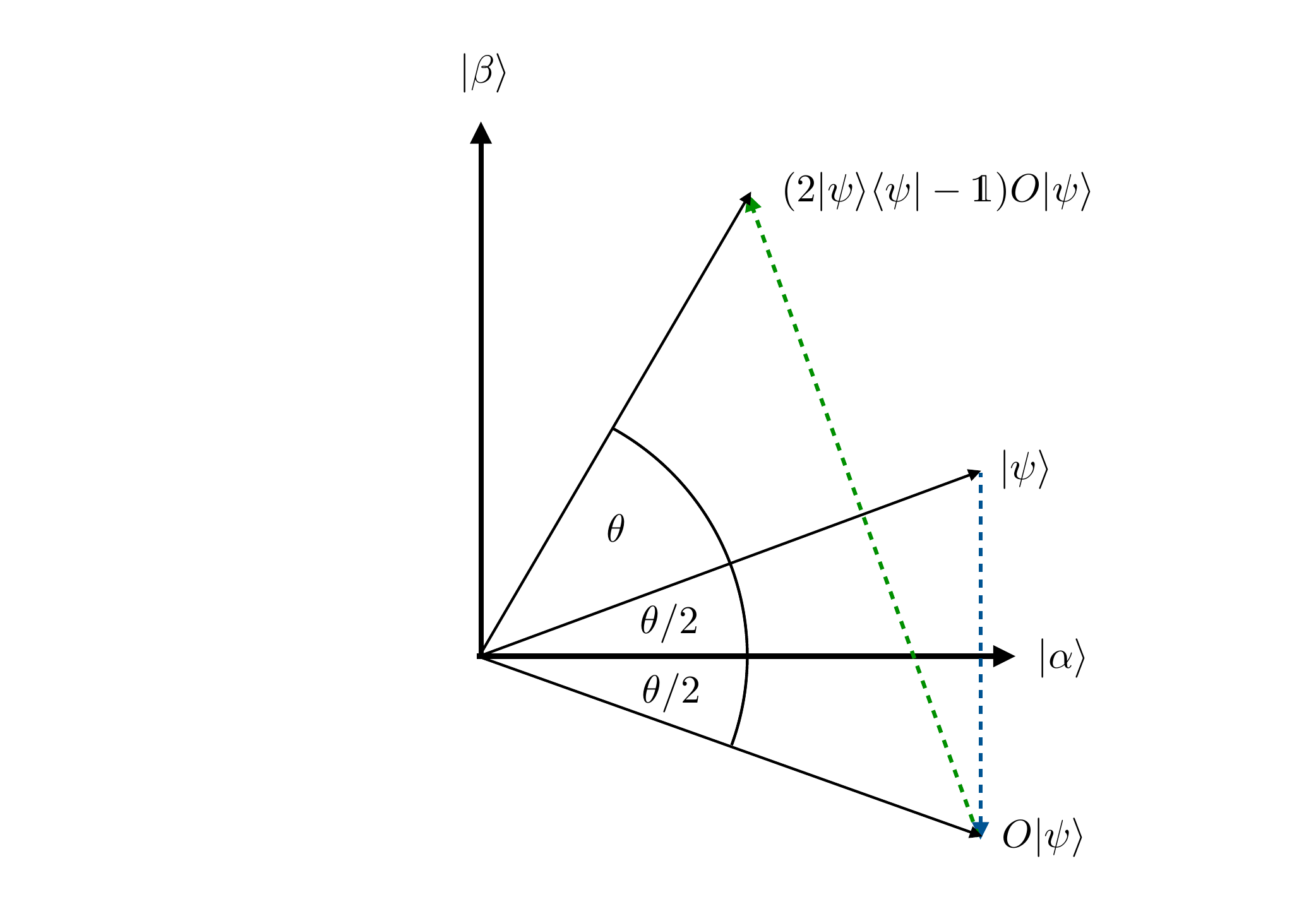}
\end{center}
By Eq.~\eqref{eq:GroverPsiab}, we have the angle
\begin{align}\label{eq:GroverSinCos}
\cos\left( \frac{\theta}{2}\right) =\sqrt{\frac{N-M}{N}}, \quad \sin\left( \frac{\theta}{2}\right) =\sqrt{\frac{M}{N}},
\end{align}
and, consequently, by~\eqref{eq:GroverG},
\begin{align}\label{eq:GroverGPsi}
G \ket{\psi} = \cos\left( \frac{3 \theta}{2}\right) \ket{\alpha} + \sin\left( \frac{3 \theta}{2}\right) \ket{\beta}.
\end{align}
From the graphical representation of $G$, we see that the action of $G$ merely happens in the plane spanned by $\ket{\alpha}$ and $\ket{\beta}$. In particular, by comparing~\eqref{eq:GroverPsiab} and~\eqref{eq:GroverGPsi}, it can be seen as a rotation by the angle 
\begin{align}\label{eq:Grovertheta2arsin}
\theta = 2 \arcsin\sqrt{\frac{M}{N}}.
\end{align}

Now, coming back to the quantum circuit~\eqref{cir:Grover}, we apply $G$ multiple times. After $k$ Grover iterations, the state reads
\begin{align}
G^k\ket{\psi}=\cos\left( \frac{(2k+1) \theta}{2}\right) \ket{\alpha} + \sin\left( \frac{(2k+1) \theta}{2}\right) \ket{\beta}.
\end{align}
For $k=4$, this can be illustrated as follows:
\begin{center}
\includegraphics[width=0.7\linewidth]{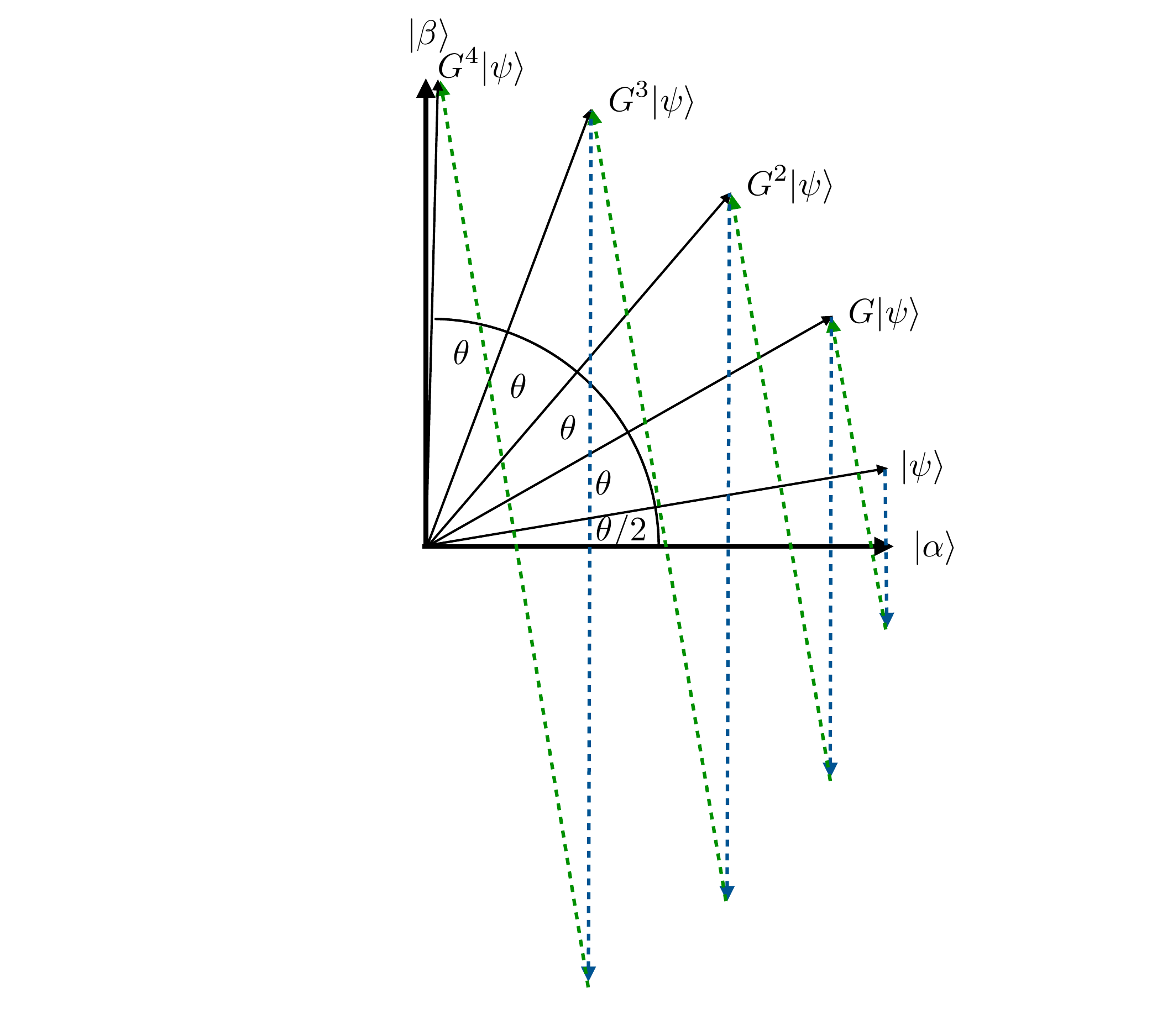}
\end{center}
The optimal solution to the search problem is then achieved for $k$ such that the amplitude of $\ket{\beta}$ (which is the amplitude of coherent superpositions of solutions) is as close as possible to one, i.e.,
\begin{align}\label{eq:GroverKupperBound}
\begin{split}
\frac{2k+1}{2} \theta &= \frac{\pi}{2} \\
\Leftrightarrow k \theta &= \frac{\pi}{2}-\frac{\theta}{2} \\
\Leftrightarrow k  &= \frac{\pi}{2\theta}-\frac{1}{2}.
\end{split}
\end{align}
This implies that the integer $k_0$ closest to $\frac{\pi}{2\theta}-\frac{1}{2}$ achieves the largest probability to find a solution by measurement of the output state, where
\begin{align}\label{eq:Groverk0}
k_0=\min\left\{ k\in \mathbb{N}\ : \ k+\frac{1}{2} \geq \frac{\pi}{ 4 \arcsin\sqrt{M/N}} -\frac{1}{2}\right\}.
\end{align}
\begin{center}
\includegraphics[width=0.6\linewidth]{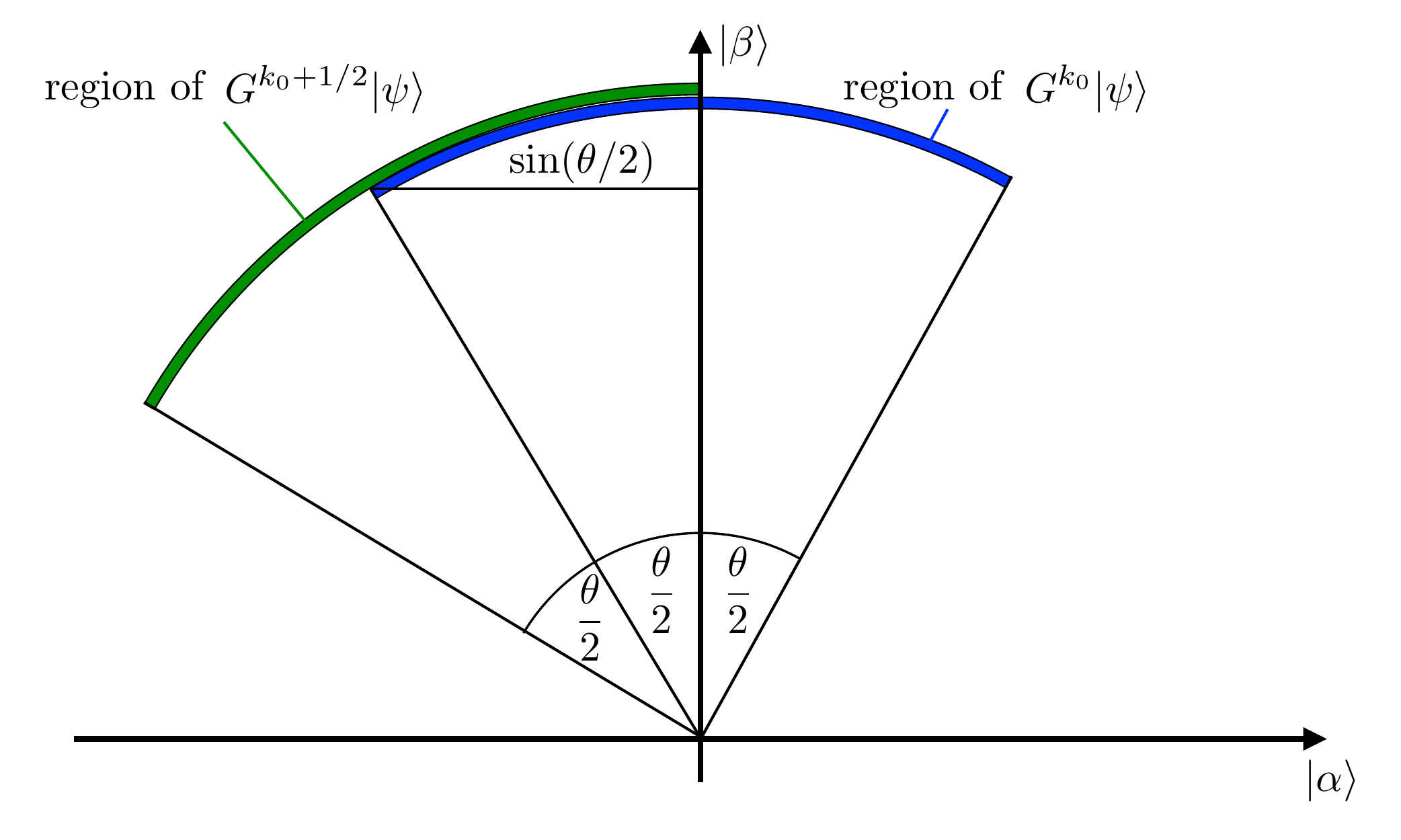}
\end{center}
After $k_0$ iterations, the state $G^{k_0}\ket{\psi}$ is not further away from $\ket{\beta}$ than $\theta/2$. Hence, the probability to detect a wrong solution on output (by measuring $G^{k_0}\ket{\psi}$ in the computational basis) is bounded by
\begin{align}
p_\text{error} \leq \sin^2 \frac{\theta}{2} \overset{\eqref{eq:GroverSinCos}}{=} \frac{M}{N}.
\end{align}
Furthermore, from Eq.~\eqref{eq:GroverKupperBound}, we find 
\begin{align}
\begin{split}
k  &\overset{\eqref{eq:Grovertheta2arsin}}{=}\frac{\pi}{4 \arcsin \sqrt{M/N}}-\frac{1}{2} \\
&\leq \frac{\pi}{4} \sqrt{\frac{N}{M}} ,
\end{split}
\end{align}
where we used $\arcsin x\geq x$ for $x\geq 0$. Hence, the Grover algorithm converges after $\mathcal{O}(\sqrt{N/M})$ oracle calls, which is a quadratic speedup compared to classical algorithms, which need $\mathcal{O}(N/M)$ oracle calls. 

Note that iterations beyond $k_0$ times will degrade the result. Hence, by virtue of Eq.~\eqref{eq:Groverk0}, one needs to know $M$ to properly determine $k_0$. There are modifications of the algorithm for which a priori knowledge of $M$ is not required [e.g. see Sec.~6.3 in \cite{Nielsen-QC-2011}].

\subsubsection*{Grover search generated by a Hamiltonian}
The Grover iteration $G$ represents a rotation in the two-dimensional subspace spanned by $\ket{\alpha}$ and $\ket{\beta}$. This suggests to consider the rotation to be generated by a suitable Hamiltonian in continuous time. Indeed, such rotation is generated by the Hamiltonian
\begin{align}\label{eq:GroverH}
H=\ketbra{\beta}{\beta}+\ketbra{\psi}{\psi},
\end{align}
with $\ket{\psi}$ from Eq.~\eqref{eq:GroverPsiab}, reading
\begin{align}
\ket{\psi}=\alpha\ket{\alpha} + \beta\ket{\beta},
\end{align}
where $\alpha=\sqrt{(N-M)/N}$ and $\beta=\sqrt{M/N}$, satisfying $\alpha^2+\beta^2=1$. In the representation $\ket{\alpha}=(1,0)^\top, \ket{\beta}=(0,1)^\top$, the Hamiltonian~\eqref{eq:GroverH} reads
\begin{align}
\begin{split}
H&=\begin{pmatrix}0&0\\0&1\end{pmatrix}+\begin{pmatrix}\alpha^2&\alpha\beta\\ \alpha\beta&\beta^2\end{pmatrix}\\
&= \begin{pmatrix}1-\beta^2&\alpha\beta\\ \alpha\beta&1+\beta^2\end{pmatrix}\\
&\overset{\eqref{eq:Pauli}}{=} \unit + \beta \left( \alpha X - \beta Z\right)\\
&= \unit + \beta \hat{n} \cdot \vec{\sigma},
\end{split}
\end{align}
where $\hat{n}=(\alpha,0,-\beta)$, and $\abs{\hat{n}} = \alpha^2+\beta^2=1$. This generates the unitary (see Sec.~\ref{sec:EvoOfStates}, and set $\hbar=1$)
\begin{align}
\begin{split}
U&=e^{-\im tH}\\
&=e^{-\im t(\unit + \beta \hat{n} \cdot \vec{\sigma})}\\
&=e^{-\im t} e^{-\im t \beta\hat{n} \cdot \vec{\sigma} }\\
&=e^{-\im t} \left[ \cos(-\beta t) \unit + \im \sin(-\beta t) \hat{n} \cdot \vec{\sigma} \right]\\
&=e^{-\im t} \left[ \cos(\beta t) \unit - \im \sin(\beta t) \hat{n} \cdot \vec{\sigma} \right].
\end{split}
\end{align}
The unitary acts on the initial state $\ket{\psi}$ as 
\begin{align}\label{eq:GroverUcont}
U\ket{\psi}=e^{-\im t} \left[ \cos(\beta t) \ket{\psi} - \im \sin(\beta t) \hat{n} \cdot \vec{\sigma} \ket{\psi}\right],
\end{align}
where
\begin{align}
\begin{split}
\hat{n} \cdot \vec{\sigma} \ket{\psi}&=(\alpha,0,-\beta)\cdot (X,Y,Z) \ket{\psi}\\
&= \begin{pmatrix}-\beta&\alpha\\ \alpha & \beta\end{pmatrix} \begin{pmatrix}\alpha \\ \beta\end{pmatrix}\\
&=\begin{pmatrix}0\\1\end{pmatrix}\\
&=\ket{\beta}.
\end{split}
\end{align}
Using this in~\eqref{eq:GroverUcont} yields
\begin{align}
U \ket{\psi} = e^{-\im t} \left[ \cos(\beta t) \ket{\psi} - \im \sin(\beta t) \ket{\beta}\right],
\end{align}
Hence, with $\beta =\sqrt{M/N}$, we find that after
\begin{align}
t=\frac{\pi}{2\beta} = \frac{\pi}{2} \sqrt{\frac{N}{M}}
\end{align}
the solution $\ket{\beta}$ is found with certainty. 

In conclusion, the Hamiltonian~\eqref{eq:GroverH} generates a unitary transformation which rotates $\ket{\psi}$ to the solution $\ket{\beta}$ after time $t=\pi/2 \sqrt{N/M}$. However, simulating the Hamiltonian via a quantum circuit does not reduce the computational costs compared to Grover's search algorithm, which needs $\mathcal{O}(\sqrt{N})$ oracle calls. Indeed, no quantum algorithm can perform the search problem using fewer than $\Omega(\sqrt{N})$ accesses to the oracle (note that $\Omega(f(x))$ denotes a ``lower bound'', similar to $\mathcal{O}(f(x))$ denoting an ``upper bound''). Hence, Grover's algorithm is optimal (you possibly prove this in the exercises).

\chapter{Quantum information}\label{ch:qinfo}

\section{Quantum noise and quantum operations} 
So far we dealt with closed (i.e., noiseless) quantum systems, which evolve unitarily. However, real systems suffer from unwanted interactions with the outside world (environment) that manifest as noise in quantum information processing (recall the introduction). Such open quantum systems can be described via the formalism of quantum operations, which is a powerful tool to address a wide range of physical scenarios. 

\begin{definition}
A \underline{trace preserving (TP) quantum operation} is a transformation of a quantum state (i.e., density operator) $\rho$ on $\mathcal{H}$ according to the map
\begin{align}\label{eq:E}
\begin{split}
\mathcal{E}: \mathcal{D}(\mathcal{H}) &\rightarrow \mathcal{D}(\mathcal{H})\\
\rho&\mapsto \mathcal{E}(\rho),
\end{split}
\end{align}
where $\mathcal{D}(\mathcal{H})$ denotes the space of density operators on $\mathcal{H}$, i.e. $\mathcal{E}$ maps density operators to density operators. 
\end{definition}
We already came across two examples, the unitary transformation $\mathcal{E}(\rho)=U\rho U^\dagger$ from Eq.~\eqref{eq:Uevolve}, and the measurement $\mathcal{E}_m(\rho)=M_m \rho M_m^\dagger/\tr{M_m^\dagger M_m\rho}$ according to Eq.~\eqref{eq:measurementstate}.

\subsection{Environments}
From postulate~\ref{p:evolution} in Sec.~\ref{sec:EvoOfStates}, we know that closed quantum systems evolve unitarily. From this point of view it is natural to consider the evolution of an open \underline{principal system} to be generated by a unitary evolution $U$ of a larger system composed of the principal system and an environment. By assuming no initial correlations between principal system and environment, $\rho\otimes \rho_\mathrm{E}$, the corresponding map reads
\begin{align}\label{eq:OpenU}
\mathcal{E}(\rho) = \trp{E}{U(\rho \otimes \rho_\mathrm{E})U^\dagger}.
\end{align}

\begin{remarks}\leavevmode
\begin{enumerate}[1)]
\item It is reasonable to consider initially uncorrelated states, since such states can often be prepared in the lab. The case of initial correlations is possibly discussed later. 

\item If $U$ describes no interaction between principal system and environment (i.e., a closed principal system), then $U=U_\mathrm{S} \otimes U_\mathrm{E}$, and~\eqref{eq:OpenU} becomes $\mathcal{E}(\rho)=\trp{E}{U_\mathrm{S}\rho U_\mathrm{S}^\dagger \otimes U_\mathrm{E}\rho_\mathrm{E}U_\mathrm{E}^\dagger}=U_\mathrm{S}\rho U_\mathrm{S}^\dagger$. This is a unitary evolution of the principal system, which is consistent with postulate~\ref{p:evolution}.

\item By the possibility to purify the environment [see Eq.~\eqref{eq:purification}], it is sufficient to consider the environment initially in a pure state, $\rho_\mathrm{E}=\ketbra{e_0}{e_0}$.
\end{enumerate}
\end{remarks}

\begin{example}
The circuit
\begin{align}
\begin{tikzcd}
\lstick{$\rho$}&\ctrl{1} & \qw \rstick{$\mathcal{E}(\rho)$}\\
 \lstick{$\ket{0}$}&\targ{} & \qw
\end{tikzcd}
\end{align}
illustrates a map $\mathcal{E}(\rho)$ with combined initial state $\rho\otimes\ketbra{0}{0}$ and unitary $U=\ketbra{0}{0}\otimes \unit + \ketbra{1}{1} \otimes X = P_0 \otimes \unit + P_1 \otimes X$. Hence, 
\begin{align}
\begin{split}
\mathcal{E}(\rho)&=\trp{2}{ [P_0\otimes \unit + P_1 \otimes X] \rho \otimes \ketbra{0}{0}[ P_0 \otimes \unit + P_1 \otimes X] }\\
&=\trp{2}{ [P_0\rho \otimes \ketbra{0}{0} + P_1 \rho \otimes \ketbra{1}{0}] [ P_0 \otimes \unit + P_1 \otimes X] }\\
&=\trp{2}{ P_0\rho P_0 \otimes \ketbra{0}{0} + P_0 \rho P_1 \otimes \ketbra{0}{1}  +P_1 \rho P_0 \otimes \ketbra{1}{0}+P_1\rho P_1 \otimes \ketbra{1}{1} }\\
&=P_0 \rho P_0 + P_1 \rho P_1.
\end{split}
\end{align}
This map destroys all coherences, since for $\rho=\begin{pmatrix}
\rho_{00}&\rho_{01}\\ \rho_{10}&\rho_{11}\end{pmatrix}$ we get $\mathcal{E}(\rho)=\begin{pmatrix}
\rho_{00}&0\\ 0&\rho_{11}\end{pmatrix}$.
\end{example}

\subsection{Operator-sum representation}
Let us suppose that the environment is w.l.o.g. initially in the state $\rho_\mathrm{E}=\ketbra{e_0}{e_0}$, and that $\{\ket{e_k}\}_k$ is an ON-basis of the environment. Equation~\eqref{eq:OpenU} can then be re-stated in terms of its \underline{operator-sum representation} or \underline{Kraus representation}
\begin{align}\label{eq:osr}
\begin{split}
\mathcal{E}(\rho)&=\trp{E}{ U(\rho \otimes \ketbra{e_0}{e_0}) U^\dagger }\\
&=\sum_k \bra{e_k} U(\rho \otimes \ketbra{e_0}{e_0}) U^\dagger \ket{e_k}\\
&=\sum_k \bra{e_k}U\ket{e_0} \rho \bra{e_0} U^\dagger \ket{e_k}\\
&=E_k \rho E_k^\dagger,
\end{split}
\end{align}
where $E_k=\bra{e_k}U\ket{e_0}$ are operators on the state space of the principal system known as \underline{operation elements} (or \underline{Kraus operators}).

\begin{remarks}\leavevmode
\begin{enumerate}[1)]
\item For $\mathcal{E}$ is a trace-preserving quantum operation, we must have
\begin{align}
\begin{split}
1&=\tr{\mathcal{E}(\rho)}\\
&=\tr{\sum_k E_k \rho E_k^\dagger}\\
&=\tr{\sum_k E_k^\dagger E_k \rho }.
\end{split}
\end{align}
Since this must hold for all $\rho$, we find that the operation elements must satisfy the \underline{completeness} \underline{relation}
\begin{align}
\sum_k E_k^\dagger E_k=\unit. 
\end{align}

\item Non-trace-preserving quantum operations can also be expressed via the operator-sum representation~\eqref{eq:osr}. However, their operation elements satisfy $\sum_k E_k^\dagger E_k < \unit$. See Sec.~\ref{sec:KrausTheorem} below. 

\item The operator-sum representation describes the dynamics of the principal system without having to explicitly consider properties of the environment. All that is needed are the operation elements $E_k$, which act on the principal system alone. 
\end{enumerate}
\end{remarks}

\subsubsection*{Physical interpretation}
There is an intuitive interpretation of the operator-sum representation: Suppose that after the unitary evolution $U$ of principal system and environment a projective measurement of the environment is performed in the basis $\{\ket{e_k}\}_k$. If the outcome $k$ occurs, the principal system is in the state
\begin{align}
\begin{split}
\rho_k &\propto \trp{E}{ (\unit \otimes \ketbra{e_k}{e_k})U( \rho \otimes \ketbra{e_0}{e_0})U^\dagger}\\
&=\bra{e_k}U\ket{e_0} \rho \bra{e_0}U^\dagger \ket{e_k}\\
&= E_k \rho E_k^\dagger,
\end{split}
\end{align}
and, including normalization, we have
\begin{align}
\rho_k =\frac{E_k \rho E_k^\dagger}{\tr{E_k^\dagger E_k \rho}}.
\end{align}
Since outcome $k$ occurs with probability 
\begin{align}
\begin{split}
p(k)&=\tr{ (\unit \otimes \ketbra{e_k}{e_k}) U(\rho \otimes \ketbra{e_0}{e_0})U^\dagger}\\
&=\tr{E_k^\dagger E_k \rho},
\end{split}
\end{align}
the state of the principal system after the measurement becomes
\begin{align}
\mathcal{E}(\rho)=\sum_k p(k) \rho_k = \sum_k E_k \rho E_k^\dagger. 
\end{align}

\subsubsection*{System-environment model for any operator-sum representation}
\begin{theorem}\label{thm:EtoU}
Given a trace-preserving quantum operation in the operator-sum representation $\mathcal{E}(\rho)=\sum_k E_k \rho E_k^\dagger$, we can construct a system-environment model which gives rise to the operation elements $E_k$. 
\end{theorem}
\begin{proof}
Let $\{\ket{e_k}\}_k$ be an ON-basis of the environment, in a one-to-one correspondence with the index $k$ of the operators $E_k$. Define an operator $U$, which acts on $\ket{\psi}\ket{e_0}$ as
\begin{align}\label{eq:opsUpr}
U\ket{\psi}\ket{e_0}= \sum_k \left( E_k \ket{\psi}\right)\ket{e_k}.
\end{align}
Since this operator satisfies 
\begin{align}
\begin{split}
\bra{\psi}\bra{e_0} U^\dagger U \ket{\phi}\ket{e_0}&=\sum_{k,k'} \left( \bra{\psi} E_{k'}^\dagger\right) \bra{e_{k'}} \left( E_k\ket{\phi} \right) \ket{e_k}\\
&=\sum_{k,k'} \bra{\psi} E_{k'}^\dagger E_k \ket{\phi} \bracket{e_{k'}}{e_k}\\
&= \sum_k \bra{\psi} E_k^\dagger E_k \ket{\phi}\\
&=\bracket{\psi}{\phi},
\end{split}
\end{align}
it can be extended to a unitary operator (which preserves the inner product) on the entire state space of the joint system. Using the spectral decomposition $\rho = \sum_j \lambda_j \ketbra{\lambda_j}{\lambda_j}$, we then have
\begin{align}
\begin{split}
\trp{E}{ U( \rho \otimes \ketbra{e_0}{e_0})U^\dagger}&=\trp{E}{ \sum_j \lambda_j U\ket{\lambda_j} \ket{e_0} \bra{\lambda_j} \bra{e_0}U^\dagger}\\
&\overset{\eqref{eq:opsUpr}}{=}\trp{E}{ \sum_{k,k'}\sum_j \lambda_j \left(E_k\ket{\lambda_j}\right) \ket{e_k} \left(\bra{\lambda_j} E_{k'}^\dagger\right)\bra{e_{k'}}U^\dagger}\\
&=\trp{E}{ \sum_{k,k'} E_k \left[\sum_j \lambda_j \ketbra{\lambda_j}{\lambda_j} \right]E_{k'}^\dagger \otimes \ketbra{e_{k'}}{e_k}}\\
&=\trp{E}{ \sum_{k,k'} E_k \rho E_{k'}^\dagger \otimes \ketbra{e_{k'}}{e_k}}\\
&=\sum_{k,k'} E_k \rho E_{k'}^\dagger \bracket{e_k}{e_{k'}}\\
&=\sum_{k} E_k \rho E_k^\dagger.
\end{split}
\end{align}
\end{proof}

\subsection{Kraus representation theorem}\label{sec:KrausTheorem}
It is reasonable to assume that a general quantum map with domain and range $\mathcal{H}$, i.e. $\mathcal{E}: \mathcal{H} \rightarrow \mathcal{H}$, must satisfy the following properties:
\begin{enumerate}
\item[(P1)] $\mathcal{E}$ is \underline{completely positive} (CP), i.e., 
\begin{align}\label{eq:P1}
\begin{split}
\mathcal{E}(A) \geq 0 \quad &\forall\ A\geq 0,\ A\text{ on }\mathcal{H},\\
(\unit \otimes \mathcal{E})(B) \geq 0 \quad &\forall\ B\geq 0,\ B\text{ on }\mathcal{H}_\mathrm{E}\otimes\mathcal{H} .
\end{split}
\end{align}

\item[(P2)] The trace must satisfy
\begin{align}
0 \leq \tr{\mathcal{E}(\rho)} \leq 1 \quad \forall \rho\in\mathcal{D}(\mathcal{H}).
\end{align}

\item[(P3)] $\mathcal{E}$ is linear, i.e., 
\begin{align}
\mathcal{E}\left( \sum_j p_j \rho_j  \right) = \sum_j p_j \mathcal{E}(\rho_j) \quad \forall p_j\geq 0, \ \sum_j p_j=1.
\end{align}
\end{enumerate}

Here, (P1) ensures that positive (density) operators are mapped to positive (density) operators. In (P2), $\tr{\mathcal{E}(\rho)}$ is the probability that the process (i.e. transition) represented by $\mathcal{E}$ occurs. If $\tr{\mathcal{E}(\rho)} <1$, then $\mathcal{E}$ does not provide a complete description of all possible transitions (it is not trace preserving). Hence, (P2) simply states that probabilities never exceed $1$. (P3) is due to quantum mechanics being linear. 

\begin{theorem}[\textbf{Kraus representation theorem}]
The map $\mathcal{E}$ satisfies (P1), (P2), and (P3) if and only if it has an operator-sum representation (or Kraus representation)
\begin{align}
\mathcal{E}(\rho)=\sum_j E_j \rho E_j^\dagger
\end{align}
for some set of operators $\{E_j\}_j$, with $E_j: \mathcal{H}\rightarrow \mathcal{H}$ linear, and $\sum_j E_j^\dagger E_j \leq \unit$.\footnote{For an operator $A$ on $\mathcal{H}$, you can read $A\leq \unit$ as $\unit-A \geq 0$, which means that $\unit-A$ must be positive. Hence, all eigenvalues of $A$ must be smaller or equal to one.}
\end{theorem}

\begin{proof}
``$\Leftarrow$'': Suppose that $\mathcal{E}(\rho) = \sum_j E_j \rho E_j^\dagger$, with $\sum_j E_j^\dagger E_j \leq \unit$, $E_j$ linear. 
\begin{enumerate}
\item[(P3):] 
\begin{align}
\mathcal{E}\left( \sum_k p_k \rho_k\right) = \sum_j E_j \sum_k p_k \rho_k E_j^\dagger\overset{E_j \text{ linear}}{=} \sum_k p_k \sum_j E_j \rho_k E_j^\dagger = \sum_k p_k \mathcal{E}(\rho_k).
\end{align}

\item [(P1):] Let $B\geq 0$ be any positive operator on $\mathcal{H}_\text{E} \otimes \mathcal{H}$, let $\ket{\psi} \in \mathcal{H}_\text{E} \otimes \mathcal{H}$ be any state, and define $\ket{\phi_j}=(\unit \otimes E_j^\dagger)\ket{\psi}$. By writing $B$ as $B=\sum_k b_k B_\text{E}^k \otimes B_\text{S}^k$ [see Eq.~\eqref{eq:OponH1H2}], it follows that
\begin{align}
\begin{split}
\bra{\psi} (\unit \otimes \mathcal{E}) (B)\ket{\psi} &= \bra{\psi} (\unit \otimes \mathcal{E}) \left( \sum_k b_k  B_\text{E}^k \otimes B_\text{S}^k \right) \ket{\psi}\\
&\overset{\mathcal{E} \text{ linear}}{=} \bra{\psi} \sum_k b_k  B_\text{E}^k \otimes \mathcal{E}(B_\text{S}^k)  \ket{\psi}\\
&=\bra{\psi} \sum_k b_k  B_\text{E}^k \otimes \sum_j E_j B_\text{S}^k E_j^\dagger  \ket{\psi}\\
&=\sum_j \bra{\psi} (\unit \otimes E_j)\left( \sum_k b_k  B_\text{E}^k \otimes  B_\text{S}^k\right) (\unit \otimes E_j^\dagger ) \ket{\psi}\\
&=\sum_j \bra{\psi} (\unit \otimes E_j)B (\unit \otimes E_j^\dagger ) \ket{\psi}\\
&=\sum_j \bra{\phi_j} B \ket{\phi_j}\\
&\geq 0
\end{split}
\end{align}
for any $\ket{\psi}$. The last line follows from $B$ being a positive operator. Hence, for any positive operator $B$, the operator $(\unit \otimes \mathcal{E})(B)$ is also positive. The first line in~\eqref{eq:P1} follows by choosing $B=C\otimes A$. 

\item [(P2):] Let's use the spectral decomposition $\rho=\sum_k \lambda_k \ketbra{\lambda_k}{\lambda_k}$, with $\lambda_k \geq 0$, $\sum_k \lambda_k=1$:
\begin{align}
\begin{split}
\tr{\mathcal{E}(\rho)} &= \tr{ \sum_j E_j \sum_k \lambda_k \ketbra{\lambda_k}{\lambda_k} E_j^\dagger} \\
&=\sum_k \lambda_k \sum_j \bra{\lambda_k } E_j^\dagger E_j  \ket{\lambda_k}  \\
&=\sum_k \lambda_k  \bra{\lambda_k } \underbrace{\sum_j E_j^\dagger E_j}_{\leq \unit}  \ket{\lambda_k}  \\
&\leq \sum_k \lambda_k\\
&=1.
\end{split}
\end{align}
The lower bound $0 \leq \tr{\mathcal{E}(\rho)} $ follows from $\mathcal{E}(\rho)$ being a positive operator.
\end{enumerate}

``$\Rightarrow$'': Suppose that $\mathcal{E}$ satisfies (P1), (P2), and (P3). Let us introduce $\mathcal{H}_\text{E}$, which is of the same dimension as $\mathcal{H}$, and let $\ket{j_\text{E}}$ and $\ket{j_\text{S}}$ be orthonormal basis states of $\mathcal{H}_\text{E}$ and $\mathcal{H}$, respectively. Then define the non-normalized maximally entangled state
\begin{align}
\ket{\alpha} = \sum_j \ket{j_\text{E}} \ket{j_\text{S}} \in \mathcal{H}_\text{E} \otimes \mathcal{H},
\end{align}
and the operator
\begin{align}\label{eq:KrausSigma}
\begin{split}
\sigma &= (\unit \otimes \mathcal{E}) \ketbra{\alpha}{\alpha}\\
&\overset{\text{(P3)}}{=} \sum_{j,k} \ketbra{j_\text{E}}{k_\text{E}} \otimes \mathcal{E}(\ketbra{j_\text{S}}{k_\text{S}}),
\end{split}
\end{align}
which is the image of $\ketbra{\alpha}{\alpha}$ under a one-sided channel $\unit \otimes \mathcal{E}$. Further, let 
\begin{align}\label{eq:KrausPsiS}
\ket{\psi_\text{S}} = \sum_j \psi_j \ket{j_\text{S}} \in \mathcal{H}
\end{align}
be any state in $\mathcal{H}$, and define a corresponding state in $\mathcal{H}_\text{E}$,
\begin{align}\label{eq:KrausPsiE}
\ket{\psi_\text{E}} = \sum_j \psi_j^* \ket{j_\text{E}} \in \mathcal{H}_\text{E}.
\end{align}
First note that
\begin{align}\label{eq:KrausSigmaIdE}
\begin{split}
\bra{\psi_\text{E}}\sigma \ket{\psi_\text{E}} &\overset{\eqref{eq:KrausSigma}}{=} \bra{\psi_\text{E}}\sum_{j,k} \ketbra{j_\text{E}}{k_\text{E}} \otimes \mathcal{E}(\ketbra{j_\text{S}}{k_\text{S}})\ket{\psi_\text{E}} \\
&\overset{\eqref{eq:KrausPsiE}}{=}\sum_{j,k} \psi_j \psi_k^* \ \mathcal{E}(\ketbra{j_\text{S}}{k_\text{S}})\\
&\overset{\text{(P3)}}{=} \mathcal{E}\left(\sum_{j,k} \psi_j \psi_k^* \ \ketbra{j_\text{S}}{k_\text{S}} \right)\\
&\overset{\eqref{eq:KrausPsiS}}{=} \mathcal{E}(\ketbra{\psi_\text{S}} {\psi_\text{S}} ).
\end{split}
\end{align}
Using this and (P1) we find that $\sigma$ from Eq.~\eqref{eq:KrausSigma} is positive, hence, we can write it in its spectral decomposition $\sigma=\sum_j \ketbra{s_j}{s_j}$, with not necessarily normalized states $\ket{s_j} \in \mathcal{H}_\text{E} \otimes \mathcal{H}$. Then let us define the map
\begin{align}
\begin{split}
E_j: \mathcal{H} &\rightarrow \mathcal{H}\\
E_j(\ket{\psi_\text{S}}) &= \bracket{\psi_\text{E}}{s_j},
\end{split}
\end{align}
which is linear, since
\begin{align}
E_j(\lambda \ket{\psi_\text{S}}) \overset{\eqref{eq:KrausPsiS},~\eqref{eq:KrausPsiE}}{=} \bra{\psi_\text{E}} \lambda \ket{s_j}= \lambda \bracket{\psi_\text{E}}{s_j} = \lambda E_j(\ket{\psi_\text{S}})
\end{align}
for all $\lambda \in \mathbb{C}$, and 
\begin{align}
\begin{split}
E_j(\ket{\psi_\text{S}} + \ket{\phi_\text{S}}) &= (\bra{\psi_\text{E}} + \bra{\phi_\text{E}}) \ket{s_j}\\
&=\bracket{\psi_\text{E}}{s_j}+\bracket{\phi_\text{E}}{s_j}\\
&=E_j(\ket{\psi_\text{S}}) +E_j(\ket{\phi_\text{S}}).
\end{split}
\end{align}
Hence, we write $E_j(\ket{\psi_\text{S}}) \equiv E_j \ket{\psi_\text{S}}$ and find
\begin{align}\label{eq:KrausEsumEjs}
\begin{split}
\sum_j E_j \ketbra{\psi_\text{S}}{\psi_\text{S}} E_j^\dagger &= \sum_j \bracket{\psi_\text{E}}{s_j}\bracket{s_j}{\psi_\text{E}} \\
&= \bra{\psi_\text{E}} \sigma \ket{\psi_\text{E}}\\
&\overset{\eqref{eq:KrausSigmaIdE}}{=} \mathcal{E}(\ketbra{\psi_\text{S}}{\psi_\text{S}} )
\end{split}
\end{align}
for all states $\ket{\psi_\text{S}} \in \mathcal{H}$. Hence, for $\rho=\sum_k p_k \ketbra{\psi_k}{\psi_k} $ we have
\begin{align}
\begin{split}
\mathcal{E}(\rho) &= \mathcal{E}\left( \sum_k p_k \ketbra{\psi_k}{\psi_k}\right)\\
&\overset{\text{(P3)}}{=}\sum_k p_k\mathcal{E}\left(  \ketbra{\psi_k}{\psi_k}\right)\\
&\overset{\eqref{eq:KrausEsumEjs}}{=} \sum_k p_k \sum_j E_j \ketbra{\psi_k}{\psi_k} E_j^\dagger \\
&\overset{E_j\text{ linear}}{=} \sum_j E_j \sum_k p_k \ketbra{\psi_k}{\psi_k}E_j^\dagger \\
&=\sum_j E_j \rho E_j^\dagger.
\end{split}
\end{align}
Finally we have to show that $\sum_j E_j^\dagger E_j \leq \unit$. To this end consider 
\begin{align}
\tr{\mathcal{E}(\rho)}=\tr{ \sum_j E_j \rho E_j^\dagger}= \tr{ \sum_j E_j^\dagger E_j \rho}.
\end{align}
By (P2), we must have $\tr{ \sum_j E_j^\dagger E_j \rho} \leq 1$ for all $\rho$. This can only be if $\sum_j E_j^\dagger E_j \leq \unit$. For example, if $\rho=\ketbra{\psi}{\psi}$, then for all $\ket{\psi}$
\begin{align}
\begin{split}
\tr{\mathcal{E}(\ketbra{\psi}{\psi})}=\bra{\psi} \sum_j E_j^\dagger E_j \ket{\psi} \leq 1 \Leftrightarrow \sum_j E_j^\dagger E_j \leq \unit.
\end{split}
\end{align}
\end{proof}

\subsection{Freedom in the operator-sum representation}
Is the operator-sum representation a unique description of the corresponding quantum operation $\mathcal{E}$? The answer turns out to be no. 

\begin{theorem}[\textbf{Unitary freedom in the operator-sum representation}]
Let $\{ E_1,\dots, E_m\}$ and $\{F_1,\dots, F_n\}$ be operation elements giving rise to the quantum operation $\mathcal{E}$ and $\mathcal{F}$, respectively. By adding zero operators to the shorter list, we may ensure $m=n$. Then $\mathcal{E}=\mathcal{F}$ if and only if there exists an $m\times m$ unitary $U$, such that
\begin{align}\label{eq:UnitaryFreedomOperatorSum}
E_j = \sum_k U_{j,k} F_k.
\end{align}
\end{theorem}
\begin{proof}
The key of the proof is Eq.~\eqref{eq:Ufreedom}, which states that for normalized states $\ket{\psi_j}$ and $\ket{\phi_k}$ and probability distributions ${p_j}_j$ and ${q_k}_k$ we have $\rho=\sum_j p_j \ketbra{\psi_j}{\psi_j}=\sum_k q_k \ketbra{\phi_k}{\phi_k}$ if and only if 
\begin{align}\label{eq:UfreedomKraus}
\sqrt{p_j} \ket{\psi_j} = \sum_k U_{j,k} \sqrt{q_k} \ket{\phi_k}
\end{align}
for some unitary $U$. \\
``$\Rightarrow$'': Suppose $\sum_j E_j \rho E_j^\dagger=\sum_k F_k \rho F_k^\dagger$ for all $\rho$. We start by defining the not necessarily normalized states 
\begin{align}\label{eq:USumef}
\begin{split}
\ket{e_j} &= \sum_l \ket{l_\text{E}} E_j\ket{l_\text{S}}\\
\ket{f_k} &= \sum_l \ket{l_\text{E}} F_k\ket{l_\text{S}},
\end{split}
\end{align}
where $\ket{l_\text{E}}$ and $\ket{l_\text{S}}$ are orthonormal basis states of $\mathcal{H}_\text{E}$ and $\mathcal{H}$, respectively. Now, note that
\begin{align}
\begin{split}
\sum_j \ketbra{e_j}{e_j}& \overset{\eqref{eq:USumef}}{=} \sum_{l,l'} \ketbra{l_\text{E}}{l'_\text{E}} \otimes \sum_j E_j \ketbra{l_\text{S}}{l'_\text{S}} E_j^\dagger \\
&= \sum_{l,l'} \ketbra{l_\text{E}}{l'_\text{E}} \otimes \sum_k F_k \ketbra{l_\text{S}}{l'_\text{S}} F_k^\dagger \\
& \overset{\eqref{eq:USumef}}{=} \sum_k \ketbra{f_k}{f_k}.
\end{split}
\end{align}
Thus, by Eq.~\eqref{eq:UfreedomKraus}, there exists a unitary $U$, such that
\begin{align}\label{eq:UKrausefU}
\ket{e_j}=\sum_k U_{j,k} \ket{f_k}.
\end{align}
Now, for an arbitrary state 
\begin{align}\label{eq:UKrausPsiS}
\ket{\psi_\text{S}} = \sum_l \psi_l \ket{l_\text{S}} \in \mathcal{H},
\end{align}
with corresponding state
\begin{align}\label{eq:UKrausPsiE}
\ket{\psi_\text{E}} = \sum_l \psi_l^* \ket{l_\text{E}} \in \mathcal{H}_\text{E},
\end{align}
we have
\begin{align}
\begin{split}
E_j \ket{\psi_\text{S}} &= \sum_l \psi_l E_j \ket{l_\text{S}} \\
&\overset{\eqref{eq:USumef},~\eqref{eq:UKrausPsiE}}{=} \bracket{\psi_\text{E}}{e_j}\\
&\overset{\eqref{eq:UKrausefU}}{=} \sum_k U_{j,k}\bracket{\psi_\text{E}}{f_k}\\
&\overset{\eqref{eq:USumef},~\eqref{eq:UKrausPsiS}}{=} \sum_k U_{j,k} F_k \ket{\psi_\text{S}},
\end{split}
\end{align}
Hence, we found that $E_j=\sum_k U_{j,k} F_k$. \\
``$\Leftarrow$'': Suppose that $E_j=\sum_k U_{j,k} F_k$. We then find
\begin{align}
\begin{split}
\sum_j E_j \rho E_j^\dagger &= \sum_j\left( \sum_k U_{j,k} F_k\right) \rho \left(\sum_{k'} U^*_{j,k'} F^\dagger_{k'} \right)\\
&=\sum_{k,k'} \underbrace{\left(\sum_j U^\dagger_{k',j} U_{j,k} \right)}_{[U^\dagger U]_{k',k}=\unit_{k',k}=\delta_{k',k}} F_k \rho F^\dagger_{k'}\\
&=\sum_k F_k \rho F_k^\dagger.
\end{split}
\end{align}
\end{proof}

\begin{example}
For the operation elements $\{E_1=\frac{1}{\sqrt{2}}(\ketbra{0}{0} + \ketbra{1}{1}), E_2=\frac{1}{\sqrt{2}}(\ketbra{0}{0} - \ketbra{1}{1})\}$ and $\{F_1=\ketbra{0}{0}, F_2=\ketbra{1}{1}\}$ on $\mathcal{H}=\mathbb{C}^2$, where
\begin{align}\label{eq:UfrEFex}
\begin{split}
E_1&=\frac{1}{\sqrt{2}}\left(F_1+F_2\right),\\
E_2&=\frac{1}{\sqrt{2}}\left(F_1-F_2\right),
\end{split}
\end{align}
we have
\begin{align}
\begin{split}
\mathcal{E}(\rho)&=E_1 \rho E_1^\dagger + E_2\rho E_2^\dagger\\
&= \frac{1}{2}(F_1+F_2)\rho(F_1^\dagger +F_2^\dagger )+\frac{1}{2}(F_1-F_2)\rho(F_1^\dagger -F_2^\dagger )\\
&= F_1 \rho F_1^\dagger + F_2 \rho F_2^\dagger \\
&= \mathcal{F}(\rho).
\end{split}
\end{align}
In Eq.~\eqref{eq:UfrEFex} we can identify the unitary $U=\frac{1}{\sqrt{2}} \begin{pmatrix}1&1\\1&-1\end{pmatrix}$ for which $E_j=\sum_k U_{j,k} F_k$ holds. 
\end{example}

\begin{remarks} With the unitary freedom~\eqref{eq:UnitaryFreedomOperatorSum} of the operator-sum representation one can show that all quantum operations $\mathcal{E}$ on a Hilbert space of dimension $\mathrm{dim}(\mathcal{H})=d$ can be generated by an operator-sum representation containing at most $d^2$ elements, i.e. $\mathcal{E}(\rho)=\sum_{j=1}^M E_k \rho E_k^\dagger$, where $1\leq M \leq d^2$. 
\begin{proof}
Exercise. 
\end{proof}
\end{remarks}

\subsection{Examples of single qubit quantum noise}
We now use the operator-sum representation to describe typical quantum noise channels of qubits. This is for example relevant for the detection and correction of errors in quantum computation known as \underline{quantum error correction}, or for the description of open quantum systems.

First recall from Eq.~\eqref{eq:QubitRho} that the state of a single qubit can be written as
\begin{align}\label{eq:rhoQbitOp}
\begin{split}
\rho&=\frac{1}{2}(\unit + \vec{v} \cdot \vec{\sigma})\\
&=\frac{1}{2} \begin{pmatrix}1+v_z & v_x- \im v_y \\ v_x + \im v_y & 1-v_z\end{pmatrix},
\end{split}
\end{align}
with $\vec{v}\in\mathbb{R}^3$ the Bloch vector representation of $\rho$ on the unit sphere.

\subsubsection*{Bit flip channel}
The bit flip channel describes a flip of the qubit (mediated by $X$) with probability $p$, while nothing happens with probability $1-p$. The corresponding operation elements are 
\begin{align}\label{eq:BitflipOpel}
E_0=\sqrt{1-p}\  \unit = \sqrt{1-p}\begin{pmatrix}1&0\\0&1\end{pmatrix}, \quad E_1= \sqrt{p}\ X=\sqrt{p}\begin{pmatrix}0&1\\1&0\end{pmatrix}.
\end{align}
Hence, the qubit transforms as
\begin{align}
\begin{split}
\mathcal{E}_\text{bit flip}(\rho)&= E_0 \rho E_0^\dagger + E_1 \rho E_1^\dagger \\
&=(1-p) \unit \rho \unit + p X \rho X\\
&\overset{\eqref{eq:rhoQbitOp}}{=} (1-p) \frac{1}{2}(\unit + \vec{v}\cdot \vec{\sigma}) + p \frac{1}{2}(X^2 + v_x X^3 + v_y XYX + v_z XZX)\\
&=(1-p) \frac{1}{2}(\unit + \vec{v}\cdot \vec{\sigma}) +p\frac{1}{2}(\unit + v_x X - v_y Y - v_z Z)\\
&=\frac{1}{2}(\unit +v_x X + (1-2p) v_y Y + (1-2p) v_z Z).
\end{split}
\end{align}
By comparison with~\eqref{eq:rhoQbitOp}, this implies the rescaling
\begin{align}
\begin{split}
v_x & \rightarrow v_x\\
v_y & \rightarrow (1-2p) v_y\\
v_z & \rightarrow (1-2p) v_z.
\end{split}
\end{align}

\subsubsection*{Phase flip channel}
The phase flip channel flips the phase of the qubit (mediated by $Z$) with probability $p$, while nothing happens with probability $1-p$. This is described by
\begin{align}\label{eq:ElementsPhaseFlip}
E_0=\sqrt{1-p}\  \unit = \sqrt{1-p}\begin{pmatrix}1&0\\0&1\end{pmatrix}, \quad E_1= \sqrt{p}\ Z=\sqrt{p}\begin{pmatrix}1&0\\0&-1\end{pmatrix}.
\end{align}
The outcome of the channel is
\begin{align}
\mathcal{E}_\text{phase flip}(\rho)=\frac{1}{2}( \unit + (1-2p) v_x X + (1-2p) v_y Y + v_z Z),
\end{align}
implying the rescaling
\begin{align}\label{eq:rescalingPhaseFlip}
\begin{split}
v_x & \rightarrow (1-2p) v_x\\
v_y & \rightarrow (1-2p) v_y\\
v_z & \rightarrow v_z.
\end{split}
\end{align}

\subsubsection*{Bit-phase flip channel}
Analogously, the bit-phase flip channel is mediated by $Y=\im XZ$, with operation elements
\begin{align}
E_0=\sqrt{1-p}\  \unit = \sqrt{1-p}\begin{pmatrix}1&0\\0&1\end{pmatrix}, \quad E_1= \sqrt{p}\ Y=\sqrt{p}\begin{pmatrix}0&-\im\\ \im&0\end{pmatrix},
\end{align}
resulting in the rescaling
\begin{align}
\begin{split}
v_x & \rightarrow (1-2p) v_x\\
v_y & \rightarrow v_y\\
v_z & \rightarrow (1-2p) v_z.
\end{split}
\end{align}

\subsubsection*{Depolarizing channel}
The depolarizing channel
\begin{align}\label{eq:DepolChannel}
\mathcal{E}(\rho) = (1-p) \rho + p \frac{\unit}{2}
\end{align}
replaces the qubit by the completely mixed state $\unit/2$ with probability $p$, while nothing happens with probability $1-p$. The operator elements are obtained using
\begin{align}
\frac{\unit}{2}= \frac{\rho + X\rho X + Y \rho Y + Z \rho Z}{4}
\end{align}
in Eq.~\eqref{eq:DepolChannel}, such that
\begin{align}
E_0=\sqrt{1-\frac{3}{4}p}\ \unit, \quad E_1=\frac{\sqrt{p}}{2}X, \quad E_2=\frac{\sqrt{p}}{2}Y, \quad E_3=\frac{\sqrt{p}}{2}Z.
\end{align}
Moreover, using~\eqref{eq:rhoQbitOp} in~\eqref{eq:DepolChannel} yields
\begin{align}
\begin{split}
\mathcal{E}(\rho)&=(1-p) \frac{1}{2}(\unit + \vec{v}\cdot\vec{\sigma}) + p \frac{1}{2} \unit \\
&=\frac{1}{2}(\unit + (1-p) \vec{v}\cdot\vec{\sigma}) ,
\end{split}
\end{align}
which shows that the depolarizing channel induces a shrinking of the Bloch vector $\vec{v}$, with rescaling
\begin{align}
\begin{split}
v_x & \rightarrow (1-p) v_x\\
v_y & \rightarrow (1-p) v_y\\
v_z & \rightarrow (1-p) v_z.
\end{split}
\end{align}
Note that the depolarizing channel results from a consecutive action of the bit, phase, and bit-phase flip channel, since $\mathcal{E}_\text{bit flip}( \mathcal{E}_\text{phase flip}(\mathcal{E}_\text{bit-phase flip}(\rho)))$, and all other permutations of these three channels, result in the rescaling $v_x \rightarrow (1-2p)^2 v_x$, $v_y \rightarrow (1-2p)^2 v_y$, $v_z\rightarrow (1-2p)^2 v_z$.

\subsubsection*{Amplitude damping}
Amplitude damping is an important description of energy dissipation. For example, consider a two-level atom which can loose energy through spontaneously emitting a photon into the environment. This process is described by the operation elements
\begin{align}
E_0=\begin{pmatrix}1&0\\0&\sqrt{1-\gamma}\end{pmatrix}, \quad E_1=\begin{pmatrix}0&\sqrt{\gamma}\\0&0\end{pmatrix},
\end{align}
where $\gamma$ is the probability of a de-excitation. The operator $E_0$ leaves $\ketbra{0}{0}$ unchanged, but reduces the amplitude of $\ketbra{1}{1}$ and the coherences $\ketbra{0}{1}$ and $\ketbra{1}{0}$, 
\begin{align}\label{eq:E0E1ADamp}
\begin{split}
E_0 \ketbra{0}{0} E_0^\dagger &= \ketbra{0}{0} \hspace{2.2cm} E_0 \ketbra{0}{1}E_0^\dagger = \sqrt{1-\gamma} \ketbra{0}{1}\\
E_0 \ketbra{1}{1} E_0^\dagger &= (1- \gamma) \ketbra{1}{1} \hspace{1cm} E_0 \ketbra{1}{0}E_0^\dagger = \sqrt{1-\gamma} \ketbra{1}{0},
\end{split}
\end{align}
and $E_1$ changes $\ketbra{1}{1}$ into $\ketbra{0}{0}$,
\begin{align}\label{eq:E0E1ADamp2}
\begin{split}
E_1 \ketbra{0}{0} E_1^\dagger &= 0 \hspace{1.9cm} E_1 \ketbra{0}{1}E_1^\dagger = 0\\
E_1\ketbra{1}{1} E_1^\dagger &= \gamma \ketbra{0}{0} \hspace{1cm} E_1 \ketbra{1}{0}E_1^\dagger = 0.
\end{split}
\end{align}
Together, this describes the dissipation of an excitation with the following rescaling of the Bloch vector $\vec{v}$:
\begin{align}
\begin{split}
v_x & \rightarrow \sqrt{1-\gamma} \ v_x\\
v_y & \rightarrow \sqrt{1-\gamma}\ v_y\\
v_z & \rightarrow \gamma + (1-\gamma) v_z.
\end{split}
\end{align}
\begin{remarks}
In quantum optics/mechanics one often has the time dependent de-excitation probability $\gamma=1-e^{-t/T_1}$, with $T_1$ called \underline{coherence time} (or \underline{relaxation time}).
\end{remarks}

\subsubsection*{Phase damping}
Phase damping describes the loss of quantum information without loss of energy. For example, if a spin state is influenced by a weakly fluctuating magnetic environment. The operation elements are
\begin{align}
E_0=\begin{pmatrix}1&0\\0&\sqrt{1-\lambda}\end{pmatrix}, \quad E_1=\begin{pmatrix}0&0\\0&\sqrt{\lambda}\end{pmatrix},
\end{align}
where $\lambda$ is the probability for the damping process to occur. Equivalently to~\eqref{eq:E0E1ADamp}, $E_0$ leaves $\ketbra{0}{0}$ unchanged and reduces the amplitude of $\ketbra{1}{1}$ and the coherences $\ketbra{0}{1}$ and $\ketbra{1}{0}$,
\begin{align}
\begin{split}
E_0 \ketbra{0}{0} E_0^\dagger &= \ketbra{0}{0} \hspace{2.2cm} E_0 \ketbra{0}{1}E_0^\dagger = \sqrt{1-\gamma} \ketbra{0}{1}\\
E_0 \ketbra{1}{1} E_0^\dagger &= (1- \lambda) \ketbra{1}{1} \hspace{1cm} E_0 \ketbra{1}{0}E_0^\dagger = \sqrt{1-\gamma} \ketbra{1}{0},
\end{split}
\end{align}
while $E_1$ reduces the amplitude of $\ketbra{1}{1}$ but (different to~\eqref{eq:E0E1ADamp2}) does not change it into a $\ketbra{0}{0}$ state, 
\begin{align}
\begin{split}
E_1 \ketbra{0}{0} E_1^\dagger &= 0\hspace{1.9cm} E_1 \ketbra{0}{1}E_1^\dagger = 0\\
E_1\ketbra{1}{1} E_1^\dagger &= \lambda \ketbra{1}{1} \hspace{1cm} E_1 \ketbra{1}{0}E_1^\dagger = 0.
\end{split}
\end{align}
By applying the unitary freedom~\eqref{eq:UnitaryFreedomOperatorSum} of the operation elements, one finds that phase damping is equivalent to the phase flip channel with operation elements~\eqref{eq:ElementsPhaseFlip},
\begin{align}
F_0=\sqrt{1-p}\begin{pmatrix}1&0\\0&1\end{pmatrix}, \quad F_1=\sqrt{p}\begin{pmatrix}1&0\\0&-1\end{pmatrix},
\end{align}
where $p=\frac{1}{2}(1-\sqrt{1-\lambda})$. Hence, the rescaling~\eqref{eq:rescalingPhaseFlip} becomes 
\begin{align}
\begin{split}
v_x & \rightarrow \sqrt{1-\lambda} v_x\\
v_y & \rightarrow \sqrt{1-\lambda} v_y\\
v_z & \rightarrow v_z.
\end{split}
\end{align}
\begin{remarks}
In quantum optics/mechanics one often has the time dependent damping probability $\gamma=1-e^{-t/T_2}$, where $T_2$ is called \underline{spin-spin relaxation time}.
\end{remarks}

\section{Distance measures} 
So far, we found that quantum noise can be described via quantum maps. That is, if an initial state $\rho$ (which carries some quantum information, e.g. a music track -- recall the introduction), is subject to quantum noise according to the map $\mathcal{E}$, it changes to $\rho'=\mathcal{E}(\rho)$. But how much does $\rho'$ differ from $\rho$? How can we find out whether $\rho$ and $\rho'$ are similar, and what does \emph{similar} mean in this context? In order to answer these questions, we now introduce distance measures.

\subsection{Classical distance measures}\label{sec:ClDistMeas}
In classical information theory you can think of information being equal to how much communication is needed in order to convey it. For example, suppose Alice and Bob share an alphabet $\{a_1, \dots, a_n\}$ composed of statistically independent letters $a_j$, which appear (in the message) with probability $p_j$, and the set of probabilities $\{p_1,\dots,p_n\}$ are known. If Alice sends Bob a message (which is usually considered to consist of either a single or infinitely many letters) using this alphabet, then information can be seen as how much communication is needed in order for Alice to send Bob her message. Information is minimal if Alice's message is build from a single letter (e.g. $a_3a_3 \dots a_3)$, since then $p_3=1$ and $p_j=0$ for $j\neq 3$, such that Bob already knows with certainty that the message will be $a_3a_3\dots a_3$. On the other hand, information is maximal if all letters appear with equal probability, $p_j=1/n$ (e.g. $a_2a_9\dots a_7$), since then Bob is maximally uncertain what the message will be. From this example, we see that information can be modeled via a random variable $X$ with outcomes $a_1,\dots,a_n$, where outcome $a_j$ appears with probability $p_j$. Hence, we are dealing with a probability distribution $\{p_1,\dots, p_n\}$. For more details on classical information, see our discussion in Sec.~\ref{sec:Shannon} below.

In order to compare classical information, we have to compare two probability distributions $\{p_1,\dots, p_n\}$ and $\{q_1,\dots, q_n\}$, i.e. two random variables $X$ and $Y$. 

\begin{definition}
The \underline{trace distance} (or \underline{$L_1$ distance}, or \underline{Kolmogorov distance}) between two probability distributions $P=\{p_x\}_x$ and $Q=\{q_x\}_x$ over the same index set is defined as
\begin{align}\label{eq:tracedistanceCl}
D(p_x,q_x)\equiv D(P,Q)=\frac{1}{2} \sum_x \abs{p_x-q_x}.
\end{align}
\end{definition}

\begin{remarks} The trace distance is a metric on space of probability distributions since it satisfies the following properties for being a metric:
\begin{enumerate}[i)]
\item Non-negative: $D(P,Q)\geq 0$,
\item Identity: $D(P,Q)=0 \Leftrightarrow P=Q$,
\item Symmetry: $D(P,Q)=D(Q,P)$,
\item Triangular inequality: $D(P,R) \leq D(P,Q)+D(Q,R)$
\end{enumerate}

\end{remarks}

\begin{definition}
The \underline{fidelity} between two probability distributions $P=\{p_x\}_x$ and $Q=\{q_x\}_x$ over the same index set is defined as
\begin{align}\label{eq:FidelityCl}
F(p_x,q_x)\equiv F(P,Q)= \sum_x \sqrt{p_xq_x}.
\end{align}
\end{definition}

\begin{remarks}\leavevmode
\begin{enumerate}[1)]
\item The fidelity is the inner product between two vectors $\vec{p}=(\sqrt{p_1},\sqrt{p_2},\dots)$ and $\vec{q}=(\sqrt{q_1},\sqrt{q_2},\dots)$, which lie on the surface of the unit sphere (since $\abs{\vec{p}}=\abs{\vec{q}}=1$). 

\item The fidelity is not a metric. It can be seen as a similarity measure, and satisfies $0\leq F(P,Q)\leq 1$, with $F(P,Q)=1$ if and only if $P=Q$. 
\end{enumerate}
\end{remarks}

\subsection{Trace distance}
We now generalize the trace distance~\eqref{eq:tracedistanceCl} between classical probability distributions to the trace distance between quantum states:
\begin{definition}
The \underline{trace distance} between two quantum states $\rho$ and $\sigma$ is defined as
\begin{align}\label{eq:Tracedistance}
D(\rho,\sigma)=\frac{1}{2}\tr{\abs{\rho-\sigma}},
\end{align}
where $\abs{A}=\sqrt{A^\dagger A}$, with $\sqrt{\cdot}$ the positive square root of a positive semi-definite matrix. 
\end{definition}

\begin{remarks}\leavevmode
\begin{enumerate}[1)]
\item The trace distance is a metric on the space of density operators since it is satisfies the following properties:
\begin{enumerate}[i)]
\item Non-negative: $D(\sigma,\rho)\geq 0$,
\item Identity: $D(\sigma,\rho)=0 \Leftrightarrow \rho=\sigma$,
\item Symmetry: $D(\sigma,\rho)=D(\rho,\sigma)$,
\item Triangular inequality: $D(\sigma,\tau) \leq D(\sigma,\rho)+D(\rho,\tau)$
\end{enumerate}

\item The trance distance is invariant under unitary transformations. That is, for $U$ unitary, we have
\begin{align}
D(U\rho U^\dagger,U\sigma U^\dagger)=D(\rho,\sigma).
\end{align}

\item If $\rho$ and $\sigma$ commute, then their trance distance $D(\sigma,\rho)$ is equal to the classical trance distance $D(S,R)$ between their eigenvalues $R=\{r_x\}_x$ and $S=\{s_x\}_x$. 

\item If $\rho=\frac{1}{2}(\unit + \vec{r}\cdot \vec{\sigma})$ and $\sigma=\frac{1}{2}(\unit + \vec{s}\cdot \vec{\sigma})$ describes a qubit with Bloch vector $\vec{r}$ and $\vec{s}$, respectively, then 
\begin{align}
D(\rho,\sigma)=\frac{\abs{\vec{r}- \vec{s}}}{2}.
\end{align}
\end{enumerate}
\end{remarks}
\begin{proof}
1)-4): Exercise. 
\end{proof}

Classical and quantum trace distance are related by the following theorem:
\begin{theorem}\label{thm:Dmax}
Let $\{E_m\}_m$ be a POVM, and let $p_m=\tr{E_m \rho}$ and $q_m=\tr{E_m \sigma}$ be the corresponding output probabilities by measuring $\rho$ and $\sigma$, respectively. Then 
\begin{align}\label{eq:Dmax}
D(\rho,\sigma)=\max_{\{E_m\}_m} D(p_m,q_m),
\end{align}
where the maximization is over all POVMs $\{E_m\}_m$.
\end{theorem}
\begin{proof}
Let us express $D(p_m,q_m)$ with the help of Eq.~\eqref{eq:tracedistanceCl},
\begin{align}\label{eq:ProofDistcl}
\begin{split}
D(p_m,q_m)&=\frac{1}{2}\sum_m \abs{p_m-q_m}\\
&=\frac{1}{2}\sum_m \abs{\tr{E_m(\rho-\sigma)}}.
\end{split}
\end{align}
Since $\rho-\sigma$ is Hermitian, i.e. $(\rho-\sigma)^\dagger=\rho^\dagger-\sigma^\dagger=\rho-\sigma$, it has a spectral decomposition with real eigenvalues $\lambda_j$,
\begin{align}
\begin{split}
\rho-\sigma &=\sum_j \lambda_j \ketbra{\lambda_j}{\lambda_j}\\
&=\underbrace{\sum_{\lambda_j>0}\lambda_j \ketbra{\lambda_j}{\lambda_j} }_{Q}+\underbrace{\sum_{\lambda_j\leq 0}\lambda_j \ketbra{\lambda_j}{\lambda_j} }_{-S} \\
&=Q-S.
\end{split}
\end{align}
Here we defined $Q=\sum_{\lambda_j>0}\lambda_j \ketbra{\lambda_j}{\lambda_j}$ and $S=-\sum_{\lambda_j\leq 0}\lambda_j \ketbra{\lambda_j}{\lambda_j}$. Note that $Q$ and $S$ are positive operators with orthogonal support. Accordingly, $\abs{\rho-\sigma}=\abs{Q-S}=Q+S$, and
\begin{align}\label{eq:DclDquineq0}
\begin{split}
\abs{\tr{E_m(\rho-\sigma)}}&=\abs{ \tr{E_m(Q-S)}}\\
& \leq \abs{\tr{E_m(Q+S)}}\\
&=\tr{E_m(Q+S)} \\
&=\tr{E_m\abs{\rho-\sigma}}.
\end{split}
\end{align}
Using this in Eq.~\eqref{eq:ProofDistcl} yields
\begin{align}\label{eq:DclDquineq}
\begin{split}
D(p_m,q_m)&\leq \frac{1}{2}\sum_m \tr{E_m\abs{ \rho-\sigma}}\\
&=\frac{1}{2} \tr{\sum_m E_m\abs{ \rho-\sigma}}\\
&=\frac{1}{2} \tr{\abs{ \rho-\sigma}}\\
&=D(\rho,\sigma),
\end{split}
\end{align}
where we used the completeness relation $\sum_m E_m=\unit$. 

Finally, we have to show that one can always find a POVM, such that the inequality in~\eqref{eq:DclDquineq} saturates. To this end, we choose a measurement whose POVM elements $E_m$ are projectors onto the support of $Q$ and $S$. In this case, the inequality in Eq.~\eqref{eq:DclDquineq0} saturates, and from Eq.~\eqref{eq:DclDquineq} we get $D(p_m,q_m)=D(\rho,\sigma)$, which finishes the proof. 
\end{proof}

The above theorem provides us with a clear interpretation of the trace distance: The trace distance $D(\rho,\sigma)$ is an upper bound of the classical trace distance $D(p_m,q_m)$ between the probability distributions $\{p_m\}_m$ and $\{q_m\}_m$ obtained by measuring $\sigma$ and $\rho$ according to any POVM, i.e. $D(p_m,q_m) \leq D(\rho,\sigma)$. Accordingly, if the trace distances between $\rho$ and $\sigma$ is small, for any possible measurement, measuring $\rho$ will yield a similar outcome as measuring $\sigma$. On the other hand, if the trace distances between $\rho$ and $\sigma$ is large, there are measurements such that the outcomes for $\rho$ and $\sigma$ differ significantly.

\begin{theorem}[\textbf{Trace preserving quantum operations are contractive}]\label{th:tpqoac}
Let $\rho$ and $\sigma$ be density operators on $\mathcal{H}$. If $\mathcal{E}$ is a trace preserving quantum operation, then 
\begin{align}\label{eq:Dtpqo}
D\left(\mathcal{E}(\rho), \mathcal{E}(\sigma)\right) \leq D(\rho,\sigma). 
\end{align}
\end{theorem}
\begin{proof}
Exercise.
\end{proof}

This is a very important theorem, which tells us that trace preserving quantum operations can only lead to a loss of information. In order to illustrate this, consider the following example:
\begin{example}
Suppose that Alice and Bob can communicate via a noisy quantum channel described by a trace preserving quantum operation $\mathcal{E}$. Alice can prepare two states $\rho$ and $\sigma$, which satisfy $D(\rho,\sigma)=1$. That is, $\rho$ and $\sigma$ can be fully distinguished (they have orthogonal support), such that Alice and Bob can associate different messages with these states (for example $\rho=$``yes'' and $\sigma=$``no''). Suppose now that Alice prepares $\rho$ and sends her message through the noisy quantum channel to Bob, who receives the state $\mathcal{E}(\rho)$. However, if $D(\mathcal{E}(\rho),\mathcal{E}(\sigma)) < 1$, the states $\mathcal{E}(\rho)$ and $\mathcal{E}(\sigma)$ (partially) overlap in Hilbert space, and Bob cannot be certain whether he received $\rho=$``yes'' or $\sigma=$``no''. Accordingly, the noisy quantum channel led to a loss of information. In the worst case, $D(\mathcal{E}(\rho),\mathcal{E}(\sigma))=0$. Hence, $\mathcal{E}(\rho)=\mathcal{E}(\sigma)$, such that Bob has no information about Alice's original message. 
\end{example}

\begin{theorem}[\textbf{Strong convexity of the trace distance}]\label{thm:Dconvex}
Let $\{p_j\}_j$ and $\{q_j\}_j$ be probability distributions, and $\{\rho_j\}_j$ and $\{\sigma_j\}_j$ be density operators over the same index set. Then 
\begin{align}
D\left( \sum_j p_j \rho_j, \sum_j q_j \sigma_j \right) \leq D(p_j, q_j) +\sum_j p_j D(\rho_j,\sigma_j). 
\end{align}
\end{theorem}
\begin{proof}
Exercise.
\end{proof}

\begin{remarks}\leavevmode
\begin{enumerate}[1)]
\item From the above theorem on the strong convexity of the trace distance, we directly find that the trace distance is convex in its first entry,
\begin{align}
D\left( \sum_j p_j \rho_j,\sigma\right) \leq \sum_j p_j D(\rho_j,\sigma).
\end{align}

\item By the symmetry property, the trace distance is also convex in the second entry. 
\end{enumerate}
\end{remarks}

\subsection{Fidelty}
Another important measure to compare two quantum states is the quantum fidelity, which is a generalization of the classical fidelity from Eq.~\eqref{eq:FidelityCl}.

\begin{definition}
The \underline{fidelity} between two quantum states $\rho$ and $\sigma$ is defined as
\begin{align}
F(\rho,\sigma)=\tr{ \sqrt{ \sqrt{\rho}\sigma \sqrt{\rho}}},
\end{align}
with $\sqrt{\cdot}$ the positive square root of a positive semi-definite matrix. 
\end{definition}

\begin{remarks}\leavevmode
\begin{enumerate}[1)]
\item The fidelity can be interpreted as a similarity measure. It is symmetric, $F(\rho,\sigma)=F(\sigma,\rho)$, and satisfies $0\leq F(\rho,\sigma)\leq 1$, with $F(\rho,\sigma)=0$ if and only if $\rho$ and $\sigma$ have orthogonal support, and $F(\rho,\sigma)=1$ if and only if $\rho=\sigma$. 

\item Similar to the trace distance, the fidelity is invariant under unitary transformations $U$,
\begin{align}
F(U\rho U^\dagger,U\sigma U^\dagger)=F(\rho,\sigma).
\end{align}

\item Similar to the trace distance, if $\rho$ and $\sigma$ commute, then their fidelity $F(\sigma,\rho)$ is equal to the classical fidelity $F(S,R)$ between their eigenvalues $R=\{r_x\}_x$ and $S=\{s_x\}_x$. 

\item The fidelity between a pure state $\ket{\psi}$ and an arbitrary state $\rho$ is
\begin{align}\label{eq:Frhomixed}
F(\ket{\psi},\rho)=\sqrt{\bra{\psi}\rho\ket{\psi}},
\end{align}
and the fidelity between two pure states $\ket{\psi}$ and $\ket{\phi}$ is
\begin{align}
F(\ket{\psi},\ket{\phi})=\abs{\bracket{\psi}{\phi}}.
\end{align}

\end{enumerate}
\end{remarks}
\begin{proof}
2)-4): Exercise. 
\end{proof}

\begin{theorem}[\textbf{Uhlmann's theorem}]\label{thm:Uhlmann}
Let $\rho$ and $\sigma$ be states on $\mathcal{H}$. Then 
\begin{align}\label{eq:Uhlmann}
F(\rho,\sigma)=\max_{\ket{\psi},\ket{\phi}} \abs{\bracket{\psi}{\phi}},
\end{align}
where the maximization is over all purifications $\ket{\psi} \in \mathcal{H} \otimes \mathcal{H} $ of $\rho$ and $\ket{\phi}\in \mathcal{H} \otimes \mathcal{H}$ of $\sigma$.
\end{theorem}
\begin{proof}
In order to prove Uhlmann's theorem, we first have to prove three useful results: 
\begin{enumerate}[1)]
\item For any operator $A$ and unitary $U$, we have
\begin{align}\label{eq:UmodA}
\abs{\tr{AU}} \leq \tr{\abs{A}},
\end{align}
with equality if $U=V^\dagger$, where $A=V\abs{A}$ is the polar decomposition of $A$ [see Eq.~\eqref{eq:polardecomposition}]. 

\item Consider two quantum systems $Q$ and $R$ with the same Hilbert space $\mathcal{H}$, and ON-bases $\{\ket{i_Q}\}_i$ and $\{\ket{i_R}\}_i$. For operators $A$ on $Q$ and $B$ on $R$ we have with $\ket{m}=\sum_i \ket{i_Q}\ket{i_R}$
\begin{align}\label{eq:TrABtop}
\tr{A B^\top}=\bra{m} (A\otimes B)\ket{m},
\end{align}
where the matrix multiplication on the left hand side is with respect to the matrix representations of $A$ and $B$ in the basis $\{\ket{i_Q}\}$ and $\{\ket{i_R}\}$, respectively. 

\item For $Q$, $R$, and $\ket{m}$ as in 2), any purification $\ket{\psi}$ (living on the combined system of $Q$ and $R$) of $\rho$ (living on $Q$) can be written as
\begin{align}\label{eq:purificationsqrtrho}
\ket{\psi}=(\sqrt{\rho}U_Q \otimes U_R)\ket{m},
\end{align}
with $U_Q$ and $U_R$ unitary. 

\end{enumerate}

Let us start with proving 1): We use the polar decomposition $A=V\abs{A}$, and get
\begin{align}
\begin{split}
\abs{\tr{AU}} &= \abs{\tr{V \abs{A} U}}\\
&=\abs{\tr{V \sqrt{\abs{A}} \sqrt{\abs{A}}  U}}.
\end{split}
\end{align}
Next we use the Cauchy-Schwarz inequality $\abs{(A,B)}^2\leq (A,A)(B,B)$ for the Hilbert-Schmidt norm $(A,B)=\tr{A^\dagger B}$, yielding
\begin{align}
\begin{split}
\abs{\tr{AU}} &\leq \sqrt{  \tr{ \sqrt{\abs{A}} V^\dagger V \sqrt{\abs{A}} }  \tr{ U^\dagger \sqrt{\abs{A}}  \sqrt{\abs{A}} U } } \\
&= \sqrt{  \tr{ \sqrt{\abs{A}}  \sqrt{\abs{A}} }  \tr{  \sqrt{\abs{A}}  \sqrt{\abs{A}}  } } \\
&=\tr{\abs{A}}.
\end{split}
\end{align}
Clearly, the inequality saturates for $U=V^\dagger$, since then $|\tr{AV^\dagger}| =|\tr{V\abs{A}V^\dagger}| =|\tr{\abs{A}}|=\tr{\abs{A}}$. \\

Next we prove 2): Using the matrix multiplication $[A B^\top]_{j,k}=\sum_l A_{j,l}B^\top_{l,k}$ and the matrix elements $A_{j,l}=\bra{j_Q}A\ket{l_Q}$ and $B^\top_{l,k}=\bra{l_R}B^\top\ket{k_R}$ yields
\begin{align}
\begin{split}
\tr{A B^\top}&=\sum_j [A B^\top]_{jj}\\
&=\sum_{j,l} \bra{j_Q}A \ket{l_Q}\bra{l_R}B^\top\ket{j_R}\\
&=\sum_{j,l} \bra{j_Q}A\ket{l_Q}\bra{j_R}B\ket{l_R}\\
&=\left( \sum_j \bra{j_Q}\bra{j_R}\right) \left( A \otimes B \right) \left( \sum_l \ket{l_Q}\ket{l_R}\right)\\
&=\bra{m}(A \otimes B)\ket{m}.
\end{split}
\end{align}

In order to prove result 3), we first write
\begin{align}
\begin{split}
\ket{\psi}&=(\sqrt{\rho}U_Q \otimes U_R)\ket{m}\\
&=(\sqrt{\rho}U_Q \otimes U_R)\sum_i \ket{i_Q}\ket{i_R}\\
&=\sum_i \sqrt{\rho}U_Q \ket{i_Q} \otimes U_R\ket{i_R} \\
&=\sum_i \sqrt{\rho} \ket{\tilde{i}_Q} \ket{\tilde{i}_R},
\end{split} 
\end{align}
with $\ket{\tilde{i}_Q}=U_Q\ket{i_Q}$, and similar for $R$. We then find
\begin{align}
\begin{split}
\trp{R}{\ket{\psi}\bra{\psi}}&=\trp{R}{ \sum_{i,j} \sqrt{\rho}   \ketbra{\tilde{i}_Q}{\tilde{j}_Q} \sqrt{\rho} \otimes \ketbra{\tilde{i}_R}{\tilde{j}_R}}\\
&=\sum_i \sqrt{\rho}   \ketbra{\tilde{i}_Q}{\tilde{i}_Q} \sqrt{\rho}\\
&=\sqrt{\rho}   \sum_i  \ketbra{\tilde{i}_Q}{\tilde{i}_Q} \sqrt{\rho}\\
&=\sqrt{\rho}\sqrt{\rho}\\
&=\rho.
\end{split} 
\end{align}
That is, for any $U_Q$ and $U_R$, Eq.~\eqref{eq:purificationsqrtrho} is a purification of $\rho$. Moreover, by the Schmidt decomposition~\eqref{eq:Schmidt} (which is unique up to unitary transformations of the ON-bases), we find that any purification of $\rho$ on the combined system of $Q$ and $R$ can be written as in Eq.~\eqref{eq:purificationsqrtrho}.

We are now set to proof Uhlmann's theorem: Let us use Eq.~\eqref{eq:purificationsqrtrho} and suppose $\ket{\psi}=(\sqrt{\rho}U_Q\otimes U_R)\ket{m}$ is a purification of $\rho$ and $\ket{\phi}=(\sqrt{\sigma}V_Q\otimes V_R)\ket{m}$ is a purification of $\sigma$, with $U_Q,U_R,V_Q,V_R$ unitary. Then we get the inner product
\begin{align}
\begin{split}
\abs{\bracket{\psi}{\phi}}&= \abs{\bra{m} (U_Q^\dagger \sqrt{\rho}\otimes U_R^\dagger)  (\sqrt{\sigma}V_Q\otimes V_R)   \ket{m}}\\
&=\abs{\bra{m} (U_Q^\dagger \sqrt{\rho} \sqrt{\sigma}V_Q \otimes U_R^\dagger V_R)    \ket{m}}
\end{split}
\end{align}
Using Eq.~\eqref{eq:TrABtop} from above yields
\begin{align}
\begin{split}
\abs{\bracket{\psi}{\phi}}&=\abs{ \tr{ U_Q^\dagger \sqrt{\rho}\sqrt{\sigma}V_Q V_R^\top U_R^* }}\\
&=\abs{ \tr{ \sqrt{\rho}\sqrt{\sigma}V_Q V_R^\top U_R^* U_Q^\dagger }}.
\end{split}
\end{align}
By defining $U=V_Q V_R^\top U_R^* U_Q^\dagger $ and using Eq.~\eqref{eq:UmodA} from above, we find
\begin{align}
\begin{split}
\abs{\bracket{\psi}{\phi}}&=\abs{ \tr{ \sqrt{\rho}\sqrt{\sigma}U }}\\
&\leq \tr{\abs{\sqrt{\rho}\sqrt{\sigma}}}\\
&= \tr{\abs{\sqrt{\sigma}\sqrt{\rho}}}\\
&=\tr{\sqrt{\sqrt{\rho}   \sqrt{\sigma} \sqrt{\sigma} \sqrt{\rho}  }}\\
&=\tr{ \sqrt{ \sqrt{\rho} \sigma \sqrt{\rho}}  }.
\end{split}
\end{align}
To see that the equality can be reached, consider the polar decomposition $\sqrt{\rho}\sqrt{\sigma}=V|\sqrt{\rho}\sqrt{\sigma}|$. Hence, by choosing $V_R^\top =U_R^* =U_Q^\dagger=\unit$ and $V_Q=V^\dagger$ the inequality saturates. 
\end{proof}

\begin{remarks}
Uhlmann's theorem is very useful since many properties of the fidelity (which are listed above) directly follow from Eq.~\eqref{eq:Uhlmann}:
\begin{enumerate}[1)]
\item The fidelity is normalized, $0\leq F(\sigma,\rho)\leq 1$.
\item $F(\sigma,\rho)= 1$ if and only if $\sigma=\rho$.
\item $F(\sigma,\rho)= 0$ if and only if $\sigma$ and $\rho$ have orthogonal support.
\item The fidelity is symmetric $F(\sigma,\rho)=F(\rho,\sigma)$.
\item The fidelity is invariant under unitary transformations,
\begin{align}
F(U\rho U^\dagger,U\sigma U^\dagger)=\max_{\ket{\psi},\ket{\phi}} \abs{\bra{\psi}(U^\dagger \otimes \unit)(U \otimes \unit)\ket{\phi}} = \max_{\ket{\psi},\ket{\phi}} \abs{\bracket{\psi}{\phi}}=F(\rho,\sigma).
\end{align}
\end{enumerate}
\end{remarks}

Similar to the trace distance, the quantum fidelity is related to the classical fidelity via the following theorem: 
\begin{theorem}\label{thm:Fmin}
Let $\{E_m\}_m$ be a POVM, and let $p_m=\tr{E_m \rho}$ and $q_m=\tr{E_m \sigma}$ be the corresponding output probabilities by measuring $\rho$ and $\sigma$, respectively. Then 
\begin{align}\label{eq:Fmin}
F(\rho,\sigma)=\min_{\{E_m\}_m} F(p_m,q_m),
\end{align}
where the minimization is over all POVMs $\{E_m\}_m$.
\end{theorem}
\begin{proof}
First note that by the polar decomposition we have $\sqrt{\sigma}\sqrt{\rho}=U^\dagger |\sqrt{\sigma}\sqrt{\rho}|=U^\dagger \sqrt{\sqrt{\rho} \sigma \sqrt{\rho}}$, and, hence, $\sqrt{\sqrt{\rho} \sigma \sqrt{\rho}}=U \sqrt{\sigma}\sqrt{\rho}$. The fidelity can then be written as
\begin{align}
\begin{split}
F(\rho,\sigma)&=\tr{U \sqrt{\sigma}\sqrt{\rho}} \\
&=\sum_m \tr{ U \sqrt{\sigma} \sqrt{E_m} \sqrt{E_m} \sqrt{\rho}},
\end{split}
\end{align}
where we inserted the completeness relation $\sum_m E_m=\unit$. Next we use the Cauchy-Schwarz inequality $\tr{A^\dagger B} \leq \sqrt{ \tr{A^\dagger A} \tr{B^\dagger B}}$ for the Hilbert-Schmidt norm, with $A^\dagger = U \sqrt{\sigma} \sqrt{E_m}$ and $B=\sqrt{E_m} \sqrt{\rho}$,
\begin{align}\label{eq:Fidminproof}
\begin{split}
F(\rho,\sigma)&\leq \sum_m \sqrt{ \tr{  U \sqrt{\sigma} \sqrt{E_m} \sqrt{E_m} \sqrt{\sigma} U^\dagger}  \tr{\sqrt{\rho} \sqrt{E_m} \sqrt{E_m} \sqrt{\rho}} } \\
&=\sum_m \sqrt{\tr{E_m \sigma} \tr{E_m \rho}}\\
&= \sum_m \sqrt{q_m p_m}\\
&=F(p_m,q_m).
\end{split}
\end{align}
That is, we found $F(\rho,\sigma)\leq \min_{\{E_m\}_m} F(p_m,q_m)$. In order to see that there is a POVM for which the inequality saturates, let us consider the use of the Cauchy-Schwarz inequality in~\eqref{eq:Fidminproof}. Here the inequality saturates if $c_m A^\dagger = B^\dagger$, for any set of coefficients $c_m \in \mathbb{C}$. That is, we require
\begin{align}
c_m U \sqrt{\sigma} \sqrt{E_m} = \sqrt{\rho} \sqrt{E_m}.
\end{align}
Now assume that $\rho$ is invertible (which is the case if it is positive definite). Then
\begin{align}
\begin{split}
&\ c_m U \sqrt{\sigma} \sqrt{\rho} \ \rho^{-1/2} \sqrt{E_m} = \sqrt{\rho} \sqrt{E_m}\\
\Leftrightarrow &\  c_m \sqrt{\sqrt{\rho} \sigma \sqrt{\rho}} \ \rho^{-1/2} \sqrt{E_m} = \sqrt{\rho} \sqrt{E_m}\\
\Leftrightarrow &\  c_m \underbrace{\rho^{-1/2} \sqrt{\sqrt{\rho} \sigma \sqrt{\rho}} \ \rho^{-1/2}}_{D} \sqrt{E_m} =  \sqrt{E_m}\\
\Leftrightarrow &\  (c_m D-\unit) \sqrt{E_m} =0.
\end{split}
\end{align}
Now let $D=\sum_j d_j \ketbra{d_j}{d_j}$ be the spectral decomposition of $D$, then choose $E_m = \ketbra{d_m}{d_m}$ and $c_m=1/d_m$, such that
\begin{align}
\begin{split}
&\ \left( \frac{1}{d_m} \sum_j d_j \ketbra{d_j}{d_j} - \unit \right) \ketbra{d_m}{d_m}=0\\
&\Leftrightarrow\ \frac{1}{d_m} d_m \ketbra{d_m}{d_m}-\ketbra{d_m}{d_m}=0.
\end{split}
\end{align}
Since this is true, we found a POVM for which the inequality in~\eqref{eq:Fidminproof} saturates if $\rho$ is invertible. The case of non-invertible $\rho$ (i.e. if $\rho$ has eigenvalues equal to zero) follows from continuity. 
\end{proof}

\begin{theorem}[\textbf{Monotonicity of the fidelity}]
Let $\rho$ and $\sigma$ be density operators on $\mathcal{H}$. If $\mathcal{E}$ is a trace preserving quantum operation, then 
\begin{align}
F\left(\mathcal{E}(\rho), \mathcal{E}(\sigma)\right) \geq F(\rho,\sigma). 
\end{align}
\end{theorem}
\begin{proof}
Exercise.
\com{
We prove the theorem using theorem~\ref{thm:EtoU} and Uhlmann's theorem~\ref{thm:Uhlmann}. Let $\ket{\psi}\in\mathcal{H}_\text{S}\otimes \mathcal{H}_\text{R}$ be a purification of $\rho$, and let $\ket{\phi}\in\mathcal{H}_\text{S}\otimes \mathcal{H}_\text{R}$ be a purification of $\sigma$ (with $\mathcal{H}_\text{S}= \mathcal{H}_\text{R}$), such that $F(\rho,\sigma)=|\bracket{\psi}{\phi}|$. Further introduce an environment initially in the pure state $\ket{0} \in \mathcal{H}_\text{E}$. According to theorem~\eqref{thm:EtoU}, for $\mathcal{E}$ a trace preserving quantum operation, we can find a unitary $U$ (acting on the principal system $S$ and the environment $E$), such that $U\ket{\psi}\ket{0}$ is a purification of $\mathcal{E}(\rho)$, and $U\ket{\phi}\ket{0}$ is a purification of $\mathcal{E}(\sigma)$. Using Uhlmann's theorem~\eqref{thm:Uhlmann} then yields
\begin{align}
\begin{split}
F\left(\mathcal{E}(\rho), \mathcal{E}(\sigma)\right) &\geq \abs{\bra{\psi}\bra{0} U^\dagger U\ket{\phi}\ket{0}}\\
&=\abs{\bracket{\psi}{\phi}}\\
&=F(\rho,\sigma).
\end{split}
\end{align}
}
\end{proof}

\begin{theorem}[\textbf{Strong concavity of the fidelity}]\label{thm:Fconcav}
Let $\{p_j\}_j$ and $\{q_j\}_j$ be probability distributions, and $\{\rho_j\}_j$ and $\{\sigma_j\}_j$ be density operators over the same index set. Then 
\begin{align}
F\left( \sum_j p_j \rho_j, \sum_j q_j \sigma_j \right) \geq \sum_j \sqrt{p_jq_j} F(\rho_j,\sigma_j). 
\end{align}
\end{theorem}
\begin{proof}
Exercise.
\end{proof}

\begin{remarks}\leavevmode
\begin{enumerate}[1)]
\item Note that with the help of theorem~\ref{thm:Fconcav} we find that the fidelity is concave in its first entry, 
\begin{align}
F\left( \sum_j p_j \rho_j,\sigma\right) \geq \sum_j p_j F(\rho_j,\sigma).
\end{align}

\item By the symmetry property, the fidelity is also concave in the second entry. 
\end{enumerate}
\end{remarks}

By comparing the properties of the fidelity (theorems~\ref{thm:Fmin}-\ref{thm:Fconcav}) with the properties of the trace distance (theorems~\ref{thm:Dmax}-\ref{thm:Dconvex}), we see that the fidelity behaves ``opposite'' to the trace distance. This suggests that both measures can be related to each other.

\subsection{Relation between trace distance and fidelity}
 Both, the fidelity and the trace distance can be used to quantify the closeness of two quantum states. In principle it does not matter which quantity to use, and often one can prove similar properties for both quantifiers. This is because of their close relation. 
 \begin{theorem}[\textbf{Fuchs-van de Graaf inequalities}]
 The fidelity $F(\rho,\sigma)$ and the trace distance $D(\rho,\sigma)$ between two quantum states $\rho$ and $\sigma$ satisfy
 \begin{align}\label{eq:Fuchs}
 1-F(\rho,\sigma) \leq D(\rho,\sigma) \leq \sqrt{1-F^2(\rho,\sigma)}.
 \end{align}
 \end{theorem}
\begin{proof}
We start with proving the second inequality $D(\rho,\sigma) \leq \sqrt{1-F^2(\rho,\sigma)}$: Let $\ket{\psi}$ be a purification of $\rho$, and let $\ket{\phi}$ be a purification of $\sigma$, such that $F(\rho,\sigma)=|\bracket{\psi}{\phi}|=F(\ket{\psi},\ket{\phi})$. Using that the partial trace is a trace preserving quantum operation, Eq.~\eqref{eq:Dtpqo} yields $D(\rho,\sigma) \leq D(\ket{\psi},\ket{\phi})$. As you will prove in the exercise [see remark 2) below], we further have $D(\ket{\psi},\ket{\phi})=\sqrt{1-F^2(\ket{\psi},\ket{\phi})}$, such that
\begin{align}
\begin{split}
D(\rho,\sigma)& \leq \sqrt{1-F^2(\ket{\psi},\ket{\phi})}\\
&=\sqrt{1- F^2(\rho,\sigma)}.
\end{split}
\end{align}

Next, we prove the first inequality $1-F(\rho,\sigma) \leq D(\rho,\sigma)$: Consider the POVM $\{E_m\}$, which minimizes~\eqref{eq:Fmin}, such that $F(\rho,\sigma)=\sum_m \sqrt{p_mq_m}$, with $p_m=\tr{E_m \rho}$ and $q_m=\tr{E_m \sigma}$. On the one hand, we then have
\begin{align}\label{eq:DFp1}
\begin{split}
\sum_m(\sqrt{p_m}-\sqrt{q_m})^2 &= \sum_m p_m + \sum_m q_m - 2 \sum_m \sqrt{p_mq_m}\\
&=2-2F(\rho,\sigma).
\end{split}
\end{align}
On the other hand, using $|\sqrt{p_m}-\sqrt{q_m}| \leq |\sqrt{p_m} + \sqrt{q_m}|$, we have
\begin{align}\label{eq:DFp2}
\begin{split}
\sum_m(\sqrt{p_m}-\sqrt{q_m})^2 & \leq \sum_m \abs{\sqrt{p_m} - \sqrt{q_m}} \abs{\sqrt{p_m} + \sqrt{q_m}} \\
&= \sum_m \abs{p_m-q_m}\\
&=2 D(p_m,q_m)\\
&\overset{\eqref{eq:Dmax}}{\leq} 2 D(\rho,\sigma).
\end{split}
\end{align}
By comparing Eqs~\eqref{eq:DFp1} and~\eqref{eq:DFp2}, we obtain $1-F(\rho,\sigma)\leq D(\rho,\sigma)$.
\end{proof}

\begin{remarks}\leavevmode
\begin{enumerate}[1)]
\item If one compares a pure state $\ket{\psi}$ with a mixed state $\rho$, the lower bound in~\eqref{eq:Fuchs} can be made tighter,
\begin{align}
1-F(\ket{\psi},\rho)^2 \leq D(\ket{\psi},\rho).
\end{align}

\item When comparing pure states, the fidelity and the trace distance are related by the equality
\begin{align}
D(\ket{\psi},\ket{\phi}) = \sqrt{1-F^2(\ket{\psi},\ket{\phi})}.
\end{align}
\end{enumerate}
\end{remarks}
\begin{proof}
1) and 2): Exercise. 
\end{proof}

\subsection{Quantifying the effect of quantum operations} 
Our initial motivation to study distance measures was to quantify how much the quantum map $\rho'=\mathcal{E}(\rho)$ (e.g. describing a quantum channel or the storage in an imperfect quantum memory) affects the initial state $\rho$ (e.g. described the information of a music track). Using the quantum fidelity and trace distance, we can now quantify to which extent the quantum map $\mathcal{E}$ affects our initial state $\rho$ by calculating $F(\rho,\rho')$ and $D(\rho,\rho')$, respectively. 
\begin{example}
Consider the Depolarizing channel with map $\mathcal{E}(\rho)=(1-p)\rho + p \unit/2$ from Eq.~\eqref{eq:DepolChannel} for a qubit initially in the pure state $\rho=\ketbra{\psi}{\psi}$. For the fideltiy $F(\ket{\psi},\mathcal{E}(\ketbra{\psi}{\psi}))$ we use Eq.~\eqref{eq:Frhomixed}, and obtain
\begin{align}\label{eq:ExDpchF}
\begin{split}
F(\ket{\psi},\mathcal{E}(\ketbra{\psi}{\psi}))&=\sqrt{ \bra{\psi} \left[(1-p) \ketbra{\psi}{\psi} + p \frac{\unit}{2} \right] \ket{\psi}}\\
&=\sqrt{1-\frac{p}{2}}.
\end{split}
\end{align}
On the other hand, using the trace distance~\eqref{eq:Tracedistance}, we have
\begin{align}
\begin{split}
D(\ket{\psi},\mathcal{E}(\ketbra{\psi}{\psi}))&=\frac{1}{2} \tr{\abs{ \ketbra{\psi}{\psi}- (1-p) \ketbra{\psi}{\psi} - p \frac{\unit}{2}}}\\
&=p\frac{1}{2} \tr{\abs{ \ketbra{\psi}{\psi} - \frac{\unit}{2}}}.
\end{split}
\end{align}
Since $\ketbra{\psi}{\psi}$ and $\unit/2$ commute, their trace distance is equal to the classical trace distance between their eigenvalues $\{1,0\}$ and $\{1/2,1/2\}$, i.e. $D(\ketbra{\psi}{\psi},\unit/2)=1/2(|1-1/2|+|0-1/2|)=1/2$, leading to
\begin{align}\label{eq:ExDpchD}
D(\ket{\psi},\mathcal{E}(\ketbra{\psi}{\psi}))&=\frac{p}{2}.
\end{align}
As intuitively expected, we see that the higher the probability $p$ to depolarize the initial state, the smaller the fidelity~\eqref{eq:ExDpchF}, and the larger the distance~\eqref{eq:ExDpchD} between the initial state $\rho$ and the final state $\rho'=\mathcal{E}(\rho)$. 
\end{example}

\begin{example}
As another interesting example, consider the bit flip channel with operation elements $E_0=\sqrt{1-p}\unit$ and $E_1=\sqrt{p} X$ from Eq.~\eqref{eq:BitflipOpel}. Similarly, for a qubit initially in the pure state $\rho=\ketbra{\psi}{\psi}$, the fidelity under the action of $\mathcal{E}$ becomes
\begin{align}
\begin{split}
F(\ket{\psi},\mathcal{E}(\ketbra{\psi}{\psi}))&=\sqrt{ \bra{\psi}\left[ (1-p)\ketbra{\psi}{\psi} + p X\ketbra{\psi}{\psi}X \right]\ket{\psi} }\\
&=\sqrt{1-p+p\abs{\bra{\psi}X\ket{\psi}}^2}.
\end{split}
\end{align}
Now, by the factor $\abs{\bra{\psi}X\ket{\psi}}^2$, the fidelity depends on the initial state. For example, if $\ket{\psi}=(\ket{0}+\ket{1})/\sqrt{2}$, we have $\abs{\bra{\psi}X\ket{\psi}}^2=1$, such that $F(\ket{\psi},\mathcal{E}(\ketbra{\psi}{\psi}))=1$. That is, the bit flip channel does not affect our initial state, since it flips $\ket{0}\rightarrow\ket{1}$ and $\ket{1}\rightarrow\ket{0}$, and, hence, $\ket{\psi}\rightarrow\ket{\psi}$. On the other hand, for the initial sate $\ket{\psi}=\ket{0}$, we get $\abs{\bra{\psi}X\ket{\psi}}^2=0$, and, accordingly, $F(\ket{\psi},\mathcal{E}(\ketbra{\psi}{\psi}))=\sqrt{1-p}$. In this case, the bit flip channel strongly affect our initial state by flipping $\ket{\psi}=\ket{0}\rightarrow\ket{1}$ with probability $p$. Note that if we don't know the initial state in advance, we can quantify the worst-case behavior by minimizing over all possible initial states, yielding $F_\text{min}(\mathcal{E})=\min_{\ket{\psi}} F(\ket{\psi},\mathcal{E}(\ketbra{\psi}{\psi})=\sqrt{1-p}$.
\end{example}

Distance measures are not only useful to quantify how much a quantum state changes if it is transmitted through a channel or stored in a memory, they also allow us to quantify how well a certain quantum operation on a state $\rho$ has been performed (e.g. a gate operation in the circuit of a quantum computational task). To this end, suppose we want to perform an operation ideally described by the unitary $U$, however, due to imperfections the performed operation is according to the map $\mathcal{E}$. We can then quantify the success of the gate operation by the \underline{gate fidelity}
\begin{align}\label{eq:GateFidelity}
F(U,\mathcal{E})=\min_{\ket{\psi}} F(U \ket{\psi}, \mathcal{E}(\ketbra{\psi}{\psi})),
\end{align}
which measures the similarity between the desired operation $U$, and the noisy operation $\mathcal{E}$.

\begin{example}
As an example, suppose we want to perform the quantum NOT gate (i.e., the Pauli-$X$ gate), but, instead, the noisy operation acts according to $\mathcal{E}(\rho)=(1-p) X\rho X + p Z \rho Z$. The gate fidelity~\eqref{eq:GateFidelity} then reads
\begin{align}
\begin{split}
F(X,\mathcal{E})&=\min_{\ket{\psi}} \sqrt{ \bra{\psi}X \left[ (1-p) X\ketbra{\psi}{\psi} X + p Z \ketbra{\psi}{\psi} Z   \right]    X\ket{\psi}}\\
&=\min_{\ket{\psi}} \sqrt{1-p + p \abs{ \bra{\psi}ZX \ket{\psi}}^2}\\
&=\sqrt{1-p}.
\end{split}
\end{align}
\end{example}

\section{Entropy and information} 
In the previous sections, we discussed and quantified how quantum maps affect quantum states, for example, if Alice wants to communicate with Bob by sending him information encoded in a quantum state $\rho$ through a quantum channel. However, yet we didn't discuss how much information a state $\rho$ contains, and, how much information will be lost when communicating through a noisy quantum channel. In order to treat such scenarios from an information theoretic point of view, we now start with quantifying the amount of information using entropy measures.

\subsection{Shannon entropy} \label{sec:Shannon}
Before we start discussing the information content of quantum states, let us first consider classical information as already introduced in our discussion about classical distance measures in Sec.~\ref{sec:ClDistMeas}. That is, we consider a message (think of Alice communicating with Bob), given by a string of letters from an alphabet $\{a_1, \dots, a_n\}$ with statistically independent letters $a_x$, which occur with a priori probabilities $p_x$ (note that these probabilities are a priori known). Classical information can, thus, be modeled by a random variable $X$ whose outcomes $a_1,\dots,a_n$ occur with probability $p_1,\dots,p_n$. 

Now suppose that Alice communicates the outcome of the random variable $X$ with Bob. If the a-priori probabilities are $p_1=1$ and $p_x=0$ for $x\neq 1$, then the message contains no information since Bob can be certain about the outcome before he learns the message. On the other hand, if $p_x=1/n$ for all $x$, then Bob is maximally uncertain about the outcome of the random variable, and the message has maximal information. In general, this information content is quantified by the Shannon entropy: 
\begin{definition}
The \underline{Shannon entropy} associated with a random variable $X$, whose outcomes $a_1,\dots,a_n$ occur with probability $p_1,\dots,p_n$, is
\begin{align}\label{eq:ShannonH}
H(X)\equiv H(p_1,\dots,p_n)=-\sum_{x=1}^n p_x \log p_x,
\end{align}
with (here and below) the logarithm taken to base $2$.
\end{definition}
The Shannon entropy can thus be seen as quantifying the uncertainty before we learn the outcome of $X$, or, equivalently, as quantifying the amount of information gained after learning the outcome. The reason why the Shannon entropy is defined in this particular form is due to $H(X)$ also quantifying the minimal physical resources required to store the information, i.e. it quantifies how many bits per source symbol are required to store the message. This is known as Shannon's noiseless coding theorem. Details as well as a proof of the theorem can be found in \cite{Nielsen-QC-2011}.

\begin{example}
Suppose Alice throws a balanced dice and communicates her result with Bob, who knows about Alice throwing a dice, but he doesn't know the outcome. In this case, Alice's message is chosen from the alphabet $\{1,2,3,4,5,6\}$ with a-priori probabilities $p_x=1/6$ for all $x=1,\dots,6$. Since the dice is balanced, the Shannon entropy gets maximal, $H(X)=-\log(1/6)=\log(6)\approx 2.585$. That is, it quantifies that Bob is maximally uncertain about the outcome, and that the information gain after learning the outcome is maximal. On the other hand, for a loaded dice with a-priori probabilities $p_6=1/3$ and $p_x=2/15$ for $x\neq 6$, we get a reduced Shannon entropy of $H(X)\approx 2.466$, and for $p_6=1$ and $p_x=0$ for $x\neq 6$, the Shannon entropy vanishes, $H(X)=0$, which tells us that there is no information gain by learning the outcome of the dice (note that $\lim_{p_x\rightarrow 0} p_x\log p_x=0$). 
\end{example}

\begin{example}
In order to see that the Shannon entropy also quantifies the minimal physical resources required to store the information, consider the example of an alphabet with four symbols, $\{1,2,3,4\}$, occurring with a-priori probabilities $p_1=1/2$, $p_2=1/4$, $p_3=p_4=1/8$. Naively, we would store the symbols using two bits, e.g. via $1=01$, $2=10$, $3=11$, $4=00$. That is, a message with $k$ symbols would require $2k$ bits. However, we can make use of the fact that the symbols occur with different probabilities $p_1>p_2>p_3=p_4$, and choose an encoding, such that symbols occurring with large probability are encoded in short bit strings. One of such a ``compressed'' encoding is $1=0$, $2=10$, $3=110$, and $4=111$. The average length of a symbol in this encoding is then $1\times 1/2+2\times 1/4 + 3\times 1/8+3\times 1/8=7/4$. That is, on average a message with $k$ symbols requires only $k 7/4 < 2k$ bits. Indeed, it turns out that this is the most compact way to compress the message. With this, we now see why the Shannon entropy is defined as in Eq.~\eqref{eq:ShannonH}: It quantifies the average symbol length, i.e. the minimal physical resources to store the information, $H(X)=-1/2\log(1/2) - 1/4\log(1/4)-1/8\log(1/8)-1/8\log(1/8)=7/4$.
\end{example}

\subsection{Further classical entropies} 
Before we can proceed with similar resource-theoretical considerations for quantum states, we must learn about some useful properties of the Shannon entropy, and study relations between two random variables $X$ and $Y$. 

\subsubsection*{Relative entropy}
First, let us ask how to measure the closeness between two probability distributions $\{p_x\}_x$ and $\{q_x\}_x$ (similar as the classical distance measures introduced in Sec.~\ref{sec:ClDistMeas}) via an entropic measure related to the Shannon entropy. 

\begin{definition}
Given two probability distributions $P=\{p_x\}_x$ and $Q=\{q_x\}_x$ over the same index set, the \underline{relative entropy} of $P$ to $Q$ is defined as
\begin{align}\label{eq:RelativeEntropy}
H(p_x\parallel q_x)\equiv H(P\parallel Q)=\sum_x p_x \log\frac{p_x}{q_x}=-H(X)-\sum_x p_x \log q_x.
\end{align}
\end{definition}

The following theorem makes it clear that the relative entropy is a distance measure: 

\begin{theorem}[\textbf{Non-negativity of the relative entropy}]\label{th:nonnegativerelativentropy}
The relative entropy is non-negative, $H(p_x\parallel q_x)\geq 0$, with equality if and only if $p_x=q_x$ for all $x$. 
\end{theorem}

\begin{proof}
Consider $z=2^{\log z}$, such that $\ln(z)=\ln(2^{\log z})= \log z\ln 2$. Using $\ln(z)\leq z-1$, with equality if and only if $z=1$, we then have $\log z \leq (z-1)/\ln 2 \Leftrightarrow -\log z \geq (1-z)/\ln 2$. Hence, for $z=q_x/p_x$, the relative entropy~\eqref{eq:RelativeEntropy} becomes
\begin{align}
\begin{split}
H(p_x\parallel q_x)&=-\sum_x p_x \log\frac{q_x}{p_x}\\
&\geq \frac{1}{\ln 2} \sum_x p_x \left(1-\frac{q_x}{p_x}\right)\\
&=\frac{1}{\ln 2} \sum_x \left( p_x-q_x\right)\\
&=\frac{1}{\ln 2} \left( 1-1\right)\\
&=0,
\end{split}
\end{align}
with equality if an only if $z=q_x/p_x=1$, i.e. $p_x=q_x$ for all $x$. 
\end{proof}

\begin{remarks}
Note that the relative entropy is a distance measure between two probability distributions. However, since it is not symmetric and doesn't satisfy the triangular inequality, it is not a metric on the space of probability distributions [see below Eq.~\eqref{eq:tracedistanceCl}].
\end{remarks}

In the above example for the Shannon entropy, we stated that $H(X)$ is maximal for a uniformly distributed random variable $X$. With the help of theorem~\eqref{th:nonnegativerelativentropy}, we can now prove this statement: 

\begin{theorem}[\textbf{Maximum of the Shannon entropy}]\label{th:maxShanentropy}
For $X$ a random variable with $n$ outcomes, the Shannon entropy satisfies $H(X)\leq \log n$, with equality if and only if $X$ is uniformly distributed over those $n$ outcomes. 
\end{theorem}

\begin{proof}
Let us consider the relative entropy~\eqref{eq:RelativeEntropy} between the probability distribution $P=\{p_1,\dots, p_n\}$ of a random variable $X$, and the uniform distribution $Q=\{1/n,\dots, 1/n\}$,
\begin{align}
\begin{split}
H(p_x\parallel q_x)&=-H(X)-\sum_{x=1}^n p_x \log\left(\frac{1}{n}\right)\\
&=-H(X)+\log n. 
\end{split}
\end{align}
Using theorem~\ref{th:nonnegativerelativentropy}, we then find $H(X) \leq \log n$, with equality if and only if the random variable $X$ corresponds to the uniform distribution of all $n$ outcomes. 
\end{proof}

\subsubsection*{Joint entropy}
We now proceed with asking how the information content of two random variables $X$ and $Y$ are related to each other. To this end, let $p(x,y)$ be the \underline{joint probability} that the outcome of $X$ is $x$, and the outcome of $Y$ is $y$. At this point, it is useful to recall that if the two random variable are independent from each other, we have $p(x,y)=p(x)p(y)$, with $p(x)$ and $p(y)$ the unconditioned probability that $X$ and $Y$ yields the outcome $x$ and $y$, respectively. However, if the outcome of $X$ affects the outcome of $Y$ (or vice versa), then $p(x,y)=p(x)p(y|x)=p(y)p(x|y)$, where $p(x|y)$ is the \ul{conditional probability}, i.e. the probability that the outcome of $X$ is $x$ under the condition that the outcome of $Y$ is $y$. For the example that $x$ and $y$ can take the values $x\in\{A,B\}$ and $y\in\{0,1\}$, respectively, the tree diagram is as follows:

\begin{center}
\includegraphics[width=0.7\linewidth]{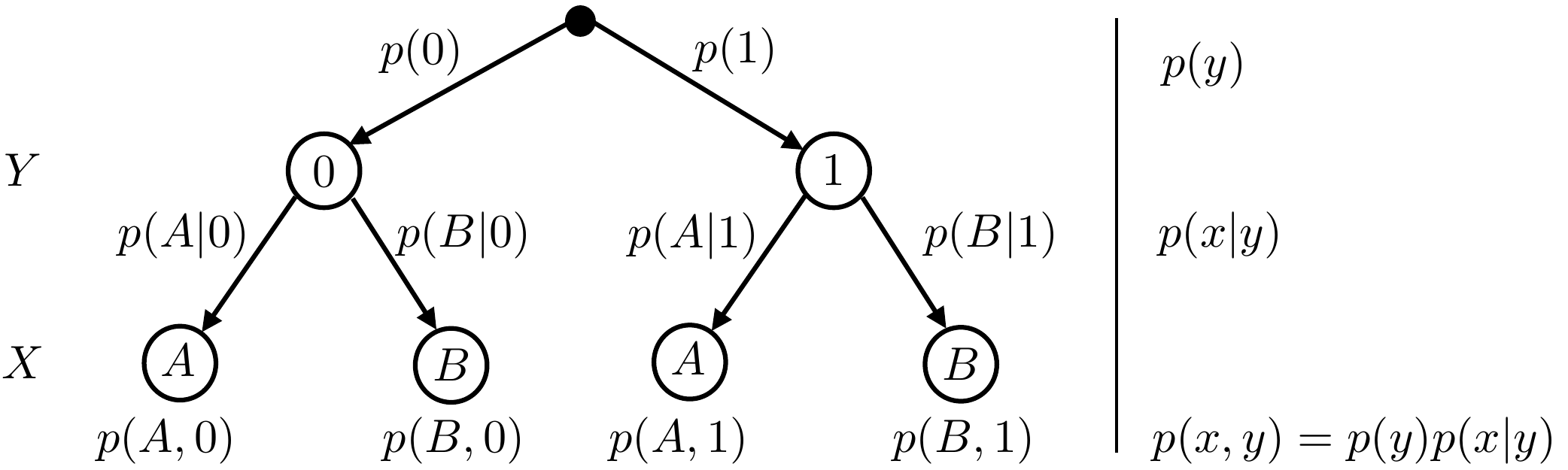}
\end{center}

With this in mind, it is intuitive to start quantifying the relation of the random variables via their joint entropy:
\begin{definition}
The \underline{joint entropy} of two random variables $X$ and $Y$ is defined as
\begin{align}\label{eq:jointentropy}
H(X,Y)=-\sum_{x,y} p(x,y) \log p(x,y). 
\end{align}
\end{definition}
By comparing with Eq.~\eqref{eq:ShannonH}, we see that the joint entropy $H(X,Y)$ is simply the Shannon entropy of the joint probabilities $p(x,y)$ of $X$ and $Y$. Hence, it quantifies our uncertainty about the pair $(X,Y)$, or, equivalently, the amount of information gained by learning the outcome of the pair $(X,Y)$. It is interesting to note that if the two random variables $X$ and $Y$ are independent, i.e. $p(x,y)=p(x)p(y)$, then the joint entropy yields $H(X,Y)=H(X)+H(Y)$. Hence, given the pair $(X,Y)$, if we know the outcome of $Y$, i.e. we gain the information $H(Y)$, we know nothing about the outcome of $X$. 

\subsubsection*{Conditional entropy}
Now suppose we are given the pair of random variables $(X,Y)$, and we already learned the outcome of $Y$, that is, we acquired $H(Y)$ bits of information about the pair $(X,Y)$. Our remaining uncertainty about the outcome of $X$ is then quantified by the conditional entropy:
\begin{definition}
For a pair of random variables $(X,Y)$, the entropy of $X$ conditioned on knowing the outcome of $Y$ defines the \underline{conditional entropy} 
\begin{align}\label{eq:conditionalentropy}
H(X|Y)=H(X,Y)-H(Y).
\end{align}
\end{definition}
Note that the expression of the conditional entropy is similar to the conditional probability $p(x|y)=p(x,y)/p(y)$, i.e. the probability that $X$ yields $x$ under the condition that $Y$ yields $y$. The only difference is that due to the logarithm in the entropy, we do not divide but subtract $H(Y)$ in Eq.~\eqref{eq:conditionalentropy}. Again, it is illustrative to consider the case of independent random variables for which $p(x,y)=p(x)p(y)$, and, as we saw above, $H(X,Y)=H(X)+H(Y)$. Using this in Eq.~\eqref{eq:conditionalentropy}, we see that, in this case, the conditional entropy becomes $H(X|Y)=H(X)$, since the outcome of $Y$ reveals no information about $X$. However, if knowing the outcome of $Y$ reveals some information about the outcome of $X$, then $H(X|Y)<H(X)$, i.e. $H(X,Y)<H(X)+H(Y)$. In other words, $X$ and $Y$ possess some mutual information. 

\subsubsection*{Mutual information}
\begin{definition}
The \underline{mutual information} of a pair of random variables $(X,Y)$ is
\begin{align}
H(X:Y)=H(X)+H(Y)-H(X,Y).
\end{align}
\end{definition}
Hence, we see that the mutual information between $X$ and $Y$ measures how much information they have in common. In the case of independent random variables $X$ and $Y$, we have $H(X,Y)=H(X)+H(Y)$, and, thus, $H(X:Y)=0$. On the other hand, if the outcomes of $X$ and $Y$ are correlated, then $H(X,Y)<H(X)+H(Y)$, and, accordingly, $H(X:Y)>0$. 

\begin{example}
While we already discussed the case of independent random variables to get an intuitive understanding of the defined entropic measures, let us now do an explicit example of two correlated random variables. Suppose my socks are covered by my trousers, and I tell you that I always wear one dark and one bright sock, but I randomly choose on which foot I wear the dark, and on which the bright sock. Now let $X$ be the random variable for the color of the left sock, which takes the outcome $x=0$ with probability $p(x=0)=1/2$ (corresponding to a dark sock), and the outcome $x=1$ with probability $p(x=1)=1/2$ (corresponding to a bright sock). Similarly, the random variable for the color of the right sock $Y$ yields $y=0$ with probability $p(y=0)=1/2$, and $y=1$ with probability $p(y=1)=1/2$. If I lift the left (right) trouser leg and you learn the color of the left (right) sock, you gain one bit of information, as measured by the Shannon entropy~\eqref{eq:ShannonH} $H(X)=H(Y)=1$. 

However, since you know that the color of the socks is anticorrelated, by lifting only one trouser leg, you can even be certain about the color of the other sock. This can now be quantified via the conditional probabilities $p(x=0|y=1)=p(x=1|y=0)=1$ and $p(x=0|y=0)=p(x=1|y=1)=0$, from which we find the joint probabilities $p(x=0,y=1)=p(x=1|y=0)=1/2$ and $p(x=0,y=0)=p(x=1,y=1)=0$. With this, the joint entropy~\eqref{eq:jointentropy} becomes $H(X,Y)=-1/2\log 1/2 - -1/2\log 1/2=\log 2=1$, and the conditional entropy $H(X|Y)=0$. The latter tells us that with knowing the outcome of $Y$, we can be certain about the outcome of $X$ (since the color of the socks is perfectly anticorrelated). Finally, for the mutual information, we get $H(X:Y)=1$, telling us that $X$ and $Y$ have one bit of information in common, namely all their information, since $H(X)=H(Y)=1$. 
\end{example}

Finally, let us state some useful properties of the above entropic measures:
\begin{theorem}\label{th:entropiesXYproperties}
For two random variables $X$ and $Y$ the following holds:
\begin{enumerate}[1)]
\item $H(X,Y)=H(Y,X)$, and $H(X:Y)=H(Y:X)$.
\item $H(Y|X)\geq 0$ and thus $H(X:Y) \leq H(Y)$, with equality if and only if $Y$ is a function of $X$, i.e., $Y=f(X)$.
\item $H(X) \leq H(X,Y)$, with equality if and only if $Y$ is a function of $X$.
\item $H(X,Y)\leq H(X)+H(Y)$, with equality if and only if $X$ and $Y$ are independent.
\item $H(Y|X)\leq H(Y)$ and thus $H(X:Y)\geq 0$, with equality if and only if $X$ and $Y$ are independent.
\end{enumerate}
\end{theorem}
\begin{proof}
Exercise. 
\end{proof}

\subsection{Von Neumann entropy} 
We are now set to continue with quantifying the information content of quantum states. Why would we want to do this? For example, suppose that Alice communicates with Bob through a quantum channel using an alphabet $\{\rho_1,\dots,\rho_n\}$ whose letters are quantum states $\rho_x$ (instead of classical quantities), which appear with a-prior probabilities $p_x$. Accordingly, Bob receives the state $\rho=\sum_x p_x \rho_x$. If the letters $\rho_x$ received by Bob are mutually orthogonal, we can treat the problem classically. However, the letters $\rho_x$ might not necessarily be orthogonal such that Bob cannot discriminate them with certainty. Hence, we now would like to know the actual information content of $\rho$ (as an explicit example, you might think of finding the maximum capacity for the communication via single photons through an optical fiber, where the photons' wavepackets can partially overlap). As another example, suppose that Alice and Bob share an entangled state. Can we  quantify their mutual information similar as for classical correlations (remember the example with the correlated socks)? We now set the stage to study such scenarios by defining entropy measures for quantum states.

First recall that the density operator $\rho$ on a $n$-dimensional Hilbert space $\mathcal{H}$ is a generalization of a random variable $X$ with $n$ outcomes, i.e. a generalization of a probability distribution $P=\{p_1,\dots,p_n\}$. Hence, we want to generalize the Shannon entropy~\eqref{eq:ShannonH} such that it accounts for quantum states $\rho$, and reduces to the Shannon entropy when considering $n$ orthogonal (pure) states. This is achieved by the von Neumann entropy:
\begin{definition}
The \underline{von Neumann entropy} of a state $\rho$ on $\mathcal{H}$ is defined as
\begin{align}\label{eq:vonNeumannentropy}
S(\rho)=-\tr{\rho \log \rho}.
\end{align}
\end{definition}

The connection of the von Neumann entropy~\eqref{eq:vonNeumannentropy} to the Shannon entropy~\eqref{eq:ShannonH} becomes more clear by inserting the eigendecomposition $\rho=\sum_x \lambda_x \ketbra{\lambda_x}{\lambda_x}$, with eigenvalues $\lambda_x$, into~\eqref{eq:vonNeumannentropy}, resulting in 
\begin{align}\label{eq:vonNeumannentropyEigen}
S(\rho)=-\sum_x \lambda_x \log \lambda_x.
\end{align}
Hence, by comparing with~\eqref{eq:ShannonH}, we see that the von Neumann entropy is the Shannon entropy of the eigenvalues of $\rho$, i.e. the Shannon entropy with respect to the orthogonal (pure) eigenstates $\ket{\lambda_x}$. The expression~\eqref{eq:vonNeumannentropyEigen} of the von Neumann entropy in terms of the eigenvalues of $\rho$ is very useful for calculations, and allows us to directly infer some useful properties of $S(\rho)$. 
 
\begin{theorem}[\textbf{Basic properties of the von Neumann entropy}]\label{th:basicpropvonNeumann}
\ 
\begin{enumerate}[1)]
\item $S(\rho)\geq 0$ with equality if and only if $\rho$ is pure, i.e. $\rho=\ketbra{\psi}{\psi}$. 
\item $S(\rho) \leq \log n$ with equality if and only if $\rho$ is maximally mixed, i.e. $\rho=\unit/n$.
\item $S(\rho)$ is invariant under unitary transformations $U$, i.e. $S(\rho)=S(U\rho U^\dagger)$.  \item If a composite system $AB$ is pure, then $S(A)=S(B)$. 
\item If $p_j$ are probabilities, and $\rho_j$ states having orthogonal support, then 
\begin{align}\label{eq:Sorthogonal}
S\left( \sum_j p_j \rho_j\right)=H(p_j)+ \sum_j p_j S(\rho_j).
\end{align}
\item $S(\rho \otimes \sigma)=S(\rho)+S(\sigma)$.
\end{enumerate}
\end{theorem}
\begin{proof}
Exercise.
\end{proof}

\begin{example}
As an example, suppose that Alice and Bob communicate via an alphabet $\{\rho_1=\ketbra{\phi_1}{\phi_1}, \rho_2=\ketbra{\phi_2}{\phi_2}\}$ composed of two pure qubit states $\ket{\phi_x}$, which appear with equal a-priori probability $p_1=p_2=1/2$. First suppose that the states are orthogonal, $\ket{\phi_1}=\ket{0}$ and $\ket{\phi_2}=\ket{1}$. Accordingly, we have $\rho=(\ketbra{0}{0}+\ketbra{1}{1})/2$, which has eigenvalues $\lambda_1=\lambda_2=1/2$, resulting in the von Neumann entropy $S(\rho)=1$. 

Now suppose non-orthogonal states $\ket{\phi_1}=\ket{0}$ and $\ket{\phi_2}=(\ket{0}+\ket{1})/\sqrt{2}$, such that Bob cannot discriminate $\rho_1$ and $\rho_2$ with certainty. A short calculation shows that $\rho=(\rho_1+\rho_2)/2$ has eigenvalues $\lambda_{1/2}=(2\pm \sqrt{2})/4$, leading to a reduced von Neumann entropy of $S(\rho)\approx 0.60$. Hence, we see that due to the non-orthogonality of the states $\rho_x$, the information content (per letter) of Alice's message reduces. 

Finally, in the extreme case of parallel states, i.e. $\ket{\phi_1}=\ket{\phi_2}=\ket{0}$, we find that $\rho=\ketbra{0}{0}$ is pure, and, thus, $S(\rho)=0$. Hence, Alice's message has no information content. 
\end{example}

\subsection{Further quantum entropies} 

\subsubsection*{Relative entropy}
Similar to the relative entropy~\eqref{eq:RelativeEntropy} of two random variables, we can define the relative entropy of two quantum states:
\begin{definition}
For two density operators $\rho$ and $\sigma$ on $\mathcal{H}$, the \underline{relative entropy} is 
\begin{align}\label{eq:quantumrelativeentropy}
S(\rho \parallel \sigma)= \tr{\rho \log \rho}-\tr{\rho \log \sigma}.
\end{align}
\end{definition}
In analogy to theorem~\ref{th:nonnegativerelativentropy}, the relative entropy of quantum states is also non-negative. This is known as Klein's inequality:
\begin{theorem}[\textbf{Klein's inequality}]
The relative entropy of two density operators $\rho$ and $\sigma$ on $\mathcal{H}$ is non-negative, i.e.,
\begin{align}\label{eq:Kleinineq}
S(\rho \parallel \sigma)\geq 0,
\end{align}
with equality if and only if $\rho=\sigma$. 
\end{theorem}
\begin{proof}
Let $\rho=\sum_x r_x \ketbra{r_x}{r_x}$ and $\sigma=\sum_y s_y \ketbra{s_y}{s_y}$ be the eigendecompositions of $\rho$ and $\sigma$. Plugging this into~\eqref{eq:quantumrelativeentropy} yields
\begin{align}\label{eq:Kleinproof1}
\begin{split}
S(\rho \parallel \sigma)&=\tr{\sum_x r_x \ketbra{r_x}{r_x} \sum_y \log r_y \ketbra{r_y}{r_y}}- \tr{\sum_x r_x \ketbra{r_x}{r_x} \sum_y \log s_y \ketbra{s_y}{s_y} } \\
&=\sum_x r_x \log r_x - \sum_{x,y}  r_x \log s_y \abs{\bracket{s_y}{r_x}}^2\\
&=\sum_x r_x \left( \log r_x - \sum_y  \log s_y \abs{\bracket{s_y}{r_x}}^2 \right).
\end{split}
\end{align}
Next, consider the right term in the parentheses. Since the logarithm is a strictly concave function, i.e. $\log((1-t)x_1+tx_2) \geq (1-t) \log x_1 + t \log x_2$ for $t\in[0,1]$, and $\abs{\bracket{s_k}{r_j}}^2\geq 0$ and $\sum_y \abs{\bracket{s_y}{r_x}}^2=1$, we have 
\begin{align}\label{eq:Kleinproof2}
\log\left(  \sum_y  s_y \abs{\bracket{s_y}{r_x}}^2\right) \geq  \sum_y \abs{\bracket{s_y}{r_x}}^2 \log s_y,
\end{align}
with equality if and only if for each $x$ there exists a value of $y$ for which $\abs{\bracket{s_y}{r_x}}^2=1$, i.e. if and only if $P_{x,y}=\abs{\bracket{s_y}{r_x}}^2$ is a permutation matrix. Using the shorthand $q_x=\sum_y  s_y \abs{\bracket{s_y}{r_x}}^2$, and plugging~\eqref{eq:Kleinproof2} into Eq.~\eqref{eq:Kleinproof1} yields
\begin{align}
\begin{split}
S(\rho \parallel \sigma)& \geq \sum_x r_x \left( \log r_x - \log q_x \right)\\
&=\sum_x r_x \log \frac{r_x}{q_x}.
\end{split}
\end{align}
The right hand side is the classical relative entropy from Eq.~\eqref{eq:RelativeEntropy}, which is non-negative according to theorem~\ref{th:nonnegativerelativentropy}. Hence, we find
\begin{align}
S(\rho \parallel \sigma)& \geq 0,
\end{align}
with equality if and only if $P_{x,y}=\abs{\bracket{s_y}{r_x}}^2$ is a permutation matrix and $r_x=q_x$ for all $x$ (the latter condition is due to theorem~\ref{th:nonnegativerelativentropy}). We can simplify the equality condition by relabeling the eigenstates of $\sigma$, such that $P_{x,y}$ is the identity. Together with the condition $r_x=q_x$, this is equivalent to $\rho=\sigma$, such that $S(\rho \parallel \sigma) \geq 0$ with equality if and only if $\rho=\sigma$.
\end{proof}

\subsubsection*{Joint entropy, conditional entropy, and mutual information}
In analogy to the classical joint entropy $H(X,Y)$, conditional entropy $H(X|Y)$, and mutual information $H(X:Y)$ of the pair of random variables $(X,Y)$, composed of the random variables $X$ and $Y$ with Shannon entropys $H(X)$ and $H(Y)$, we can now introduce similar definitions for a  quantum state $\rho_\mathrm{AB}$ of a composite system $AB$, and the reduced states $\rho_\mathrm{A}=\trp{B}{\rho_\mathrm{AB}}$ and $\rho_\mathrm{B}=\trp{A}{\rho_\mathrm{AB}}$ of subsystems $A$ and $B$, which have von Neumann entropy $S(A)=S(\rho_\mathrm{A})$ and $S(B)=S(\rho_\mathrm{B})$:
\begin{definition}\label{def:jointentropyQ}
For $\rho_\mathrm{AB}$ a density matrix of a composite quantum system with components $A$ and $B$, the \underline{joint entropy} is defined as
\begin{align}
S(A,B)=-\tr{\rho_\mathrm{AB} \log \rho_\mathrm{AB}}.
\end{align}
\end{definition}
\begin{definition}
For a composite quantum system $AB$, the \underline{conditional entropy} is defined as
\begin{align}\label{eq:conditionalentropyQ}
S(A|B)=S(A,B)-S(B).
\end{align}
\end{definition}
\begin{definition}\label{def:mutualinfoQ}
The \underline{mutual information} of a composite quantum system $AB$ is defined as
\begin{align}\label{eq:mutualinfoQ}
S(A:B)=S(A)+S(B)-S(A,B).
\end{align}
\end{definition}

\begin{remarks}\leavevmode
\begin{enumerate}[1)]
\item The joint entropy satisfies $S(A,B)=S(\rho_\mathrm{AB})$.

\item The quantum counterparts of the entropies do not necessarily satisfy the same relations as the classical entropies. For example, while $H(X|Y)\geq 0$ (cf. theorem~\ref{th:entropiesXYproperties}), it can happen that $S(A|B)<0$ (see the example below).

\item If $\ket{\psi}$ is a pure state of a composite quantum system $AB$, then $S(A|B)<0$ if and only if $\ket{\psi}$ is entangled.
\begin{proof}
Exercise. 
\end{proof}

\end{enumerate}
\end{remarks}

\begin{example}
Recall our above example of \textit{classical correlations}, where we considered perfectly anticorrelated  random variables $X$ and $Y$, which correspond to the color of two socks and take the outcome $0$ (for a dark sock) and $1$ (for a bright sock). There we found 
\begin{align*}
H(X)=H(Y)=1,\quad H(X,Y)=1,\quad H(X|Y)=0,\quad H(X:Y)=1.
\end{align*}
If we were to describe this classical correlation by a composite quantum state, it would correspond to the mixed state $\rho_\mathrm{AB}=(\ketbra{10}{10}+\ketbra{01}{01})/2$, with reduced states $\rho_\mathrm{A}=\rho_\mathrm{B}=\unit/2$. Using the above definitions of the quantum entropies, it is easy to see that they coincide with the classical entropies, 
\begin{align*}
S(A)=S(B)=1,\quad S(A,B)=1,\quad S(A|B)=0,\quad S(A:B)=1.
\end{align*}

Now let us see that \textit{quantum correlations} can differ from that. To this end, consider a composite system $AB$ in the entangled two-qubit state $\ket{\psi^+}=(\ket{01}+\ket{10})/\sqrt{2}$, such that $\rho_\mathrm{AB}=\ketbra{\psi^+}{\psi^+}$ is pure, and $\rho_\mathrm{A}=\rho_\mathrm{B}=\unit/2$ is maximally mixed. Hence, by theorem~\ref{th:basicpropvonNeumann} and definitions~\ref{def:jointentropyQ} -~\ref{def:mutualinfoQ}, we have 
\begin{align*}
S(A)=S(B)=1,\quad S(A,B)=0,\quad S(A|B)=-1,\quad S(A:B)=2.
\end{align*}
Hence, compared to classical correlations, we see that entanglement can lead to an enhanced mutual information.
\end{example}

\subsection{Further properties of the von Neumann entropy} 
We continue with stating further important properties of the von Neumann entropy $S(\rho)$. 
\begin{theorem}[\textbf{Projective measurements increase entropy}]\label{th:entroproject}
Let $\rho$ be a density operator on $\mathcal{H}$, let $\{P_x\}_x$ be a complete set of orthogonal projectors, and let $\rho'=\sum_x P_x \rho P_x$ be the state after the projective measurement when we don't learn the outcome. Then 
\begin{align}
S(\rho') \geq S(\rho),
\end{align}
with equality if and only if $\rho=\rho'$.
\end{theorem}
\begin{proof}
Consider Klein's inequality~\eqref{eq:Kleinineq} for $\rho$ and $\rho'$, and the definition of the relative entropy from Eq.~\eqref{eq:quantumrelativeentropy},
\begin{align}
0\leq S(\rho \parallel \rho') = -S(\rho) - \tr{\rho \log \rho'},
\end{align}
with equality if and only if $\rho=\rho'$. Next, we show that $-\tr{\rho \log \rho'}=-\tr{\rho' \log \rho'}=S(\rho')$ from which our claim follows. To this end, let us use the completeness relation $\sum_x P_x =\unit$, and $P_x=P_x^2$, such that
\begin{align}\label{eq:ProjecEntrop1}
\begin{split}
-\tr{\rho \log \rho'}&=-\tr{ \sum_x P_x^2 \rho \log \rho'}\\
&=-\tr{ \sum_x P_x \rho \log \rho' P_x}.
\end{split}
\end{align}
Next, note that $\rho' P_x = \sum_y P_y \rho P_y P_x=P_x \rho P_x= \sum_y P_x P_y \rho P_y= P_x \rho'$. Hence, $\rho'$ and $P_x$ commute, and so does $\log \rho'$ and $P_x$. Using this in Eq.~\eqref{eq:ProjecEntrop1} yields
\begin{align}
\begin{split}
-\tr{\rho \log \rho'}&=-\tr{ \sum_x P_x \rho P_x \log \rho' }\\
&=-\tr{ \rho' \log \rho' }\\
&=S(\rho'),
\end{split}
\end{align}
which completes the proof.
\end{proof}

\begin{theorem}[\textbf{Subadditivity}]
For a composite quantum system $AB$,
\begin{align}\label{eq:subadditivity}
S(A,B) \leq S(A) + S(B),
\end{align}
with equality if and only if $\rho_\mathrm{AB}=\rho_\mathrm{A}\otimes \rho_\mathrm{B}$.
\end{theorem}
\begin{proof}
Let us consider Klein's inequality~\eqref{eq:Kleinineq} and  Eq.~\eqref{eq:quantumrelativeentropy}, $S(\rho) \leq -\tr{\rho \log \sigma}$. By setting $\rho=\rho_\mathrm{AB}$ and $\sigma=\rho_\mathrm{A} \otimes \rho_\mathrm{B}$, we get
\begin{align}
S(A,B) &\leq -\tr{\rho_\mathrm{AB} \log\left( \rho_\mathrm{A} \otimes \rho_\mathrm{B}\right)}.
\end{align}
Next, we use 
\begin{align}
\begin{split}
\log\left( \rho_\mathrm{A} \otimes \rho_\mathrm{B}\right) &= \log\left(\left( \rho_\mathrm{A} \otimes \unit\right) \left( \unit \otimes \rho_\mathrm{B}\right)\right)\\
&=\log\left( \rho_\mathrm{A} \otimes \unit\right) +\log \left( \unit \otimes \rho_\mathrm{B}\right)\\
&=\log  \rho_\mathrm{A} \otimes \unit + \unit \otimes  \log \rho_\mathrm{B},
\end{split}
\end{align}
such that
\begin{align}
\begin{split}
S(A,B) &\leq -\tr{\rho_\mathrm{AB} \left( \log\rho_\mathrm{A} \otimes \unit + \unit \otimes \log \rho_\mathrm{B}\right)}\\
&=-\tr{\rho_\mathrm{AB} \left( \log\rho_\mathrm{A} \otimes \unit \right)}-\tr{\rho_\mathrm{AB} \left( \unit \otimes \log\rho_\mathrm{B} \right)}\\
&=-\trp{A}{\rho_\mathrm{A} \log\rho_\mathrm{A} }-\trp{B}{\rho_\mathrm{B}  \log\rho_\mathrm{B}}\\
&=S(A)+S(B).
\end{split}
\end{align}
Since Klein's inequality saturates if and only if $\rho=\sigma$, this proves the subadditivity from Eq.~\eqref{eq:subadditivity}. 
\end{proof}

\begin{theorem}[\textbf{Concavity of the entropy; Entropy of a mixture}]
Let $\rho_x$ be density operators and $p_x$ probabilities satisfying $\sum_x p_x=1$. Then 
\begin{align}\label{eq:concavityS}
\sum_x p_x S(\rho_x) \leq S\left(\sum_x p_x \rho_x \right) \leq H(p_x) + \sum_x p_x S(\rho_x).
\end{align}
The first inequality saturates if and only if all states $\rho_x$ are identical, and the second inequality saturates if and only if the states $\rho_x$ have support on orthogonal subspaces. 
\end{theorem}

\begin{proof}
We start with proving the first inequality, which is the concavity of the entropy: Consider a composite system $AB$, and the state
\begin{align}\label{eq:concavvNrhoab}
\rho_\mathrm{AB}=\sum_x p_x \rho_x \otimes \ketbra{x}{x},
\end{align}
with $\rho_x$ belonging to subsystem $A$, and $\{\ket{x}\}_x$ an ON-basis of subsystem $B$. Then we have
\begin{align}
S(A)&=S\left( \sum_x p_x \rho_x  \right)\\
S(B)&=S\left( \sum_x p_x \ketbra{x}{x} \right)=H(p_x),
\end{align}
and
\begin{align}\label{eq:EntropyCalc}
\begin{split}
S(A,B)&=S\left(\sum_x p_x \rho_x \otimes \ketbra{x}{x}\right)\\
&\overset{\eqref{eq:Sorthogonal}}{=} H(p_x) + \sum_x p_x S(\rho_x \otimes \ketbra{x}{x})\\
&=H(p_x) + \sum_x p_x \left[S(\rho_x) + S(\ketbra{x}{x}) \right]\\
&=H(p_x) + \sum_x p_x S(\rho_x).
\end{split}
\end{align}
The first and second equation follows from taking the partial trace in Eq.~\eqref{eq:concavvNrhoab} over subsystem $B$ and $A$, respectively, and the last equation follows from the orthogonality of $\rho_x \otimes \ketbra{x}{x}$ for different $x$ and items 1),5), and 6) in theorem~\ref{th:basicpropvonNeumann}. Plugging these equations into the subadditivity from Eq.~\eqref{eq:subadditivity} then leads to
\begin{align}\label{eq:proofconcavS}
 \sum_x p_x S(\rho_x) \leq S\left( \sum_x p_x \rho_x  \right),
\end{align}
which coincides with the first inequality in~\eqref{eq:concavityS}. Since the subadditivity~\eqref{eq:subadditivity} saturates if and only if $\rho_\mathrm{AB}=\rho_\mathrm{A}\otimes \rho_\mathrm{B}$, by Eq.~\eqref{eq:concavvNrhoab}, we see that~\eqref{eq:proofconcavS} saturates if and only if all $\rho_x$ are identical. This finishes the proof of the first inequality in~\eqref{eq:concavityS}.

Next we prove the second inequality in~\eqref{eq:concavityS}. First, we consider pure states $\rho_x=\ketbra{\psi_x}{\psi_x}$. Let us introduce an additional system $B$ with ON-basis $\{\ket{x}\}_x$, and consider the joint pure state
\begin{align}
\ket{AB}= \sum_x \sqrt{p_x} \ket{\psi_x} \ket{x}
\end{align}
on the composite system $AB$. Since $\ket{AB}$ is pure, we have 
\begin{align}
S(B)=S(A)=S\left(\sum_x p_x \ketbra{\psi_x}{\psi_x} \right)=S(\rho),
\end{align}
where we introduced the shorthand $\rho=\sum_x p_x \rho_x=\sum_x p_x \ketbra{\psi_x}{\psi_x}$. Next, suppose that we perform a projective measurement on system $B$ with projectors $P_x=\ketbra{x}{x}$. Then the state of subsystem $B$ after the measurement is $\rho'=\sum_x p_x \ketbra{x}{x}$. Accordingly, using theorem~\ref{th:entroproject}, we find $S(\rho)=S(B) \leq S(B') = H(p_x)$, with equality if and only if all states $\rho_x$ have support on orthogonal subspaces, since then $\rho=\rho'$. Including $S(\rho_x)=0$ for pure states $\rho_x=\ketbra{\psi_x}{\psi_x}$, this can also be written as
\begin{align}
S(\rho) \leq H(p_x) + \sum_x p_x S(\rho_x),
\end{align}
which proves the second inequality in~\eqref{eq:concavityS} for pure states. 

Finally, let us consider the case of mixed states $\rho_x$. Consider the eigendecomposition $\rho_x=\sum_y \lambda_y^x \ketbra{\lambda_y^x}{\lambda_y^x}$, such that $\rho=\sum_{x,y} p_x \lambda_y^x \ketbra{\lambda_y^x}{\lambda_y^x}$. By applying our result for pure states, we find
\begin{align}
\begin{split}
S(\rho) &\leq H(p_x\lambda_y^x)\\
&=-\sum_{x,y} p_x \lambda_y^x \log\left( p_x \lambda_y^x\right)\\
&=-\sum_x\left( \sum_y \lambda_y^x \right) p_x \log p_x - \sum_x  p_x \sum_y \lambda_y^x \log \lambda_y^x.
\end{split}
\end{align}
Using $\sum_y \lambda_y^x=1$ for all $x$, we then arrive at
\begin{align}
S(\rho) &\leq H(p_x) + \sum_x p_x S(\rho_x),
\end{align}
which coincides with the second inequality in~\eqref{eq:concavityS}. From the condition for saturation in the case of pure states, and the fact that $\ket{\lambda_y^x}$ are orthogonal for fixed $x$ and different $y$, we find that the inequality saturates if and only if $\rho_x$ have support on orthogonal subspaces. 
\end{proof}

\begin{theorem}[\textbf{Strong subadditivity}]
For a composite quantum system $ABC$,
\begin{align}\label{eq:strongsubadd}
S(A,B,C)+S(B) \leq S(A,B)+S(B,C)
\end{align}
\end{theorem}
\begin{proof}
See page 521 in \cite{Nielsen-QC-2011}.
\end{proof}

\begin{theorem}[\textbf{Conditioning reduces entropy}]
For a composite quantum system $ABC$,
\begin{align}
S(A|B,C) \leq S(A|B).
\end{align}
\end{theorem}
\begin{proof}
Using the definition of the conditional entropy~\eqref{eq:conditionalentropyQ}, we have $S(A|B,C)=S(A,B,C)-S(B,C)$ and $S(A|B)=S(A,B)-S(B)$, such that 
\begin{align}
\begin{split}
&S(A|B,C) \leq S(A|B)\\
\Leftrightarrow& S(A,B,C)-S(B,C) \leq S(A,B)-S(B) \\
\Leftrightarrow&S(A,B,C) + S(B) \leq S(A,B) +S(B,C),
\end{split}
\end{align}
which is the strong subadditivity from Eq.~\eqref{eq:strongsubadd}. 
\end{proof}

\begin{theorem}[\textbf{Discarding quantum systems never increases mutual information}]\label{th:discardmutual}
For a composite quantum system $ABC$,
\begin{align}\label{eq:discarding}
S(A:B) \leq S(A:B,C).
\end{align}
\end{theorem}
\begin{proof}
With the definition of the mutual information from Eq.~\eqref{eq:mutualinfoQ}, we have
\begin{align}
\begin{split}
&S(A:B) \leq S(A:B,C)\\
\Leftrightarrow & S(A)+S(B)-S(A,B) \leq S(A)+S(B,C) - S(A,B,C) \\
\Leftrightarrow & S(A,B,C)+S(B) \leq S(A,B)+S(B,C),
\end{split}
\end{align}
which is the strong subadditivity from Eq.~\eqref{eq:strongsubadd}. 
\end{proof}

\begin{theorem}[\textbf{Quantum operations never increase mutual information}]\label{th:mutualinfoquantumop}
Let $AB$ be a composite quantum system, $\mathcal{E}$ a trace-preserving quantum operation on subsystem $B$, $S(A:B)$ the mutual information between $A$ and $B$ before the operation, and $S(A',B')$ the mutual information after $\mathcal{E}$ was applied to subsystem $B$. Then
\begin{align}
S(A':B')\leq S(A:B).
\end{align}
\end{theorem}
\begin{proof}
Let us introduce a third subsystem $C$, initially in a pure state and uncorrelated with $B$, such that $\mathcal{E}$ acts as a unitary transformation on the composite system $BC$ [recall Eq.~\eqref{eq:osr}]. Since $C$ is initially in a pure state and uncorrelated with $B$, we have $S(A:B)=S(A:B,C)$, and since $\mathcal{E}$ only acts unitarily on $BC$, we have $S(A:B,C)=S(A':B',C')$. Moreover, using~\eqref{eq:discarding}, we get $S(A':B',C')\geq S(A',B')$. Hence, altogether we have $S(A:B) \geq S(A':B')$. 
\end{proof}

\subsection{The Holevo bound} 
Let us come back to the communication scenario between Alice and Bob. Suppose Alice uses the alphabet $\{\rho_1,\dots,\rho_n\}$ whose letters are not necessarily orthogonal quantum states $\rho_x$. Suppose that these states appear with a-priori probabilities $\{p_1,\dots,p_n\}$, which defines the   random variable $X$. Sometimes, we also say that Alice draws states from the ensemble $\mathcal{E}=\{p_x,\rho_x\}_{x=1}^n$, and sends them to Bob. Hence, Bob receives the state $\rho=\sum_x p_x \rho_x$. His task is to determine the value of $X$ as best as he can. To do so, he performs a measurement according to the POVM $\{E_1,\dots,E_m\}$. The outcome of his measurement is the set of probabilities $\{p(1),\dots,p(m)\}$, where $p(y)=\tr{\rho E_y}$, which we associate with the random variable $Y$. That is, Bob tries to identify the value of $X$ based on his measurement result $Y$. 

The amount of information gained by Bob can be quantified by the mutual information $H(X:Y)$ of the random variable $X$ (which Bob would like to learn) and Bob's measurement outcome $Y$. Bob can infer $X$ with certainty if and only if $H(X:Y)=H(X)$ (for a proof, see theorem~11.5 on page 509 in  \cite{Nielsen-QC-2011}), while in general $H(X:Y)\leq H(X)$ (see theorem~\ref{th:entropiesXYproperties}). Indeed, the closer $H(X:Y)$ to $H(X)$, the more information has Bob about $X$. Accordingly, Bob tries to choose the POVM such that $H(X:Y)$ becomes maximal. This defines Bob's accessible information about $X$: 

\begin{definition}
The \underline{accessible information} of an ensemble $\mathcal{E}=\{p_x,\rho_x\}_{x=1}^n$ associated with the random variable $X$ is defined as
\begin{align}
H_\mathrm{acc}(\mathcal{E})=\max_{\{E_y\}_y} H(X:Y),
\end{align}
with the maximization running over all POVMs $\{E_1,\dots,E_m\}$ with outcome $Y$. 
\end{definition}

Note that for states $\rho_x$ with support on orthogonal subspaces one can always find a POVM for which $H_\mathrm{acc}(\mathcal{E})=H(X:Y)=H(X)$. This doesn't hold for non-orthogonal states, since non-orthogonal states cannot reliably be distinguished, such that by measuring the states $\rho_x$ we can never be certain about $X$. In other words, $H_\mathrm{acc}(\mathcal{E}) < H(X)$ (we will see this in a bit). But how large can the accessible information be? There is no general method for calculating $H_\mathrm{acc}(\mathcal{E})$, however, the following theorem provides an useful upper bound for the accessible information:

\begin{theorem}[\textbf{The Holevo bound}]
For an ensemble $\mathcal{E}=\{p_x,\rho_x\}_{x=1}^n$ with associated state $\rho=\sum_x p_x \rho_x$, the accessible information is bounded by
\begin{align}\label{eq:Holevo}
H_\mathrm{acc}(\mathcal{E}) \leq S(\rho)- \sum_x p_x S(\rho_x).
\end{align}
\end{theorem}
\begin{proof}
Let us consider a tripartite system $PQM$, where subsystem $P$ is the ``preparation system'' with orthonormal basis states $\ket{x}$, subsystem $Q$ is the ``quantum system'' which Alice sends to Bob, and subsystem $M$ is the ``measuring device'' whose orthonormal basis states $\ket{y}$ correspond to the measurement outcomes $y=1,\dots, m$. Before Bob's measurement, we consider the common state
\begin{align}\label{eq:rhoPQMHolevo}
\rho_{PQM}=\sum_x p_x \ketbra{x}{x} \otimes \rho_x \otimes \ketbra{0}{0},
\end{align}
i.e. the preparation system labels the states $\rho_x$ of the quantum system, and the measuring device is initially in $\ketbra{0}{0}$. Bob's measurement of the POVM $\{E_y\}_y$ is modeled by a quantum operation $\mathcal{E}$ acting merely on $QM$ (but not on $P$) according to 
\begin{align}
\mathcal{E}\left( \rho_x \otimes \ketbra{0}{0} \right)= \sum_y \sqrt{E_y} \rho_x \sqrt{E_y} \otimes \ketbra{y}{y}.
\end{align}
Hence, the state after the measurement (indicated by primes) is
\begin{align}\label{eq:rhoPQMPrimeHolevo}
\rho_{P'Q'M'}=\sum_{x,y} p_x \ketbra{x}{x} \otimes \sqrt{E_y} \rho_x \sqrt{E_y}  \otimes \ketbra{y}{y}.
\end{align}
Note that $\mathcal{E}$ is a trace-preserving quantum operation since 
\begin{align}
\begin{split}
\tr{\sum_y \sqrt{E_y} \rho_x \sqrt{E_y} \otimes \ketbra{y}{y}} &= \tr{\sum_y \sqrt{E_y} \rho_x \sqrt{E_y} } \\
 &=\tr{\sum_y E_y \rho_x}\\
 &=\tr{\rho_x}\\
 &=1.
\end{split}
\end{align}

Now, observe that before the measurement we have $S(P:Q)=S(P:Q,M)$, since $Q$ and $M$ are uncorrelated. The measurement operation $\mathcal{E}$ on $QM$ is trace-preserving, and, by theorem~\ref{th:mutualinfoquantumop}, cannot increase mutual information, such that $S(P:Q,M)\geq S(P':Q',M')$. Finally, we discard subsystem $Q$, leading to $S(P':Q',M')\geq S(P':M')$, since discarding quantum systems cannot increase mutual information [see theorem~\ref{th:discardmutual}]. Putting all together gives
\begin{align}\label{eq:ProofHolevo}
S(P':M') \leq S(P:Q).
\end{align}
With a little bit of rewriting, we see that this is the Holevo bound. First, consider the right hand side: $S(P:Q)=S(P)+S(Q)-S(P,Q)$. By tracing over the respective subsystems in Eq.~\eqref{eq:rhoPQMHolevo}, we have $\rho_{PQ}=\sum_x p_x \ketbra{x}{x} \otimes \rho_x$, $\rho_P=\sum_x p_x \ketbra{x}{x}$, and $\rho_Q=\sum_x p_x \rho_x$. From this we find $S(P)=H(X)$, $S(Q)=S(\rho)$, and, by theorem~\ref{th:basicpropvonNeumann} [see items 5) and 6) and Eq.~\eqref{eq:EntropyCalc}], $S(P,Q)=H(X) + \sum_x p_x S(\ketbra{x}{x}\otimes\rho_x)=H(X) + \sum_x p_x S(\rho_x)$. Altogether, this gives
\begin{align}
S(P:Q)=S(\rho)-\sum_x p_x S(\rho_x),
\end{align}
which is the right hand side of the Holevo bound~\eqref{eq:Holevo}. 

Next, consider the left hand side in~\eqref{eq:ProofHolevo}, $S(P':M')=S(P')+S(M')-S(P',M')$: Tracing over subsystem $Q$ in Eq.~\eqref{eq:rhoPQMPrimeHolevo} yields $\rho_{P'M'}=\sum_{x,y} p_x\tr{E_y \rho_x} \ketbra{x}{x} \otimes \ketbra{y}{y}$. Now, observe that $p_x\tr{E_y \rho_x}=p_x p(y|x)=p(x,y)$, such that $\rho_{P'M'}=\sum_{x,y} p(x,y) \ketbra{x}{x} \otimes \ketbra{y}{y}$. Moreover, we have $\rho_{P'}=\sum_x p_x \ketbra{x}{x}$, and $\rho_{M'}=\sum_y p_y \ketbra{y}{y}$. Accordingly, $S(P')=H(X)$, $S(M')=H(Y)$, and $S(P',M')=H(X,Y)$. Altogether, we then get
\begin{align}\label{eq:SPMHol}
S(P':M')=H(X)+H(Y)-H(X,Y)=H(X:Y). 
\end{align}
Finally note that Eq.~\eqref{eq:ProofHolevo} holds independently on the exact POVM. Hence, by choosing the POVM $\{E_y\}_y$, such that $H(X:Y)$ is maximal, Eq.~\eqref{eq:SPMHol} reads 
\begin{align}
S(P':M')=\max_{\{E_y\}_y}H(X:Y)=H_\mathrm{acc}(\mathcal{E}),
\end{align}
which coincides with the left hand side of the Holevo bound~\eqref{eq:Holevo}, and finishes the proof. 
\end{proof}

\begin{remarks}\leavevmode
\begin{enumerate}[1)]
\item The right hand side of~\eqref{eq:Holevo} is often denoted $\chi(\mathcal{E})=S(\rho)- \sum_x p_x S(\rho_x)$, and called \underline{Holevo $\chi$ quantity} or \underline{Holevo information}.

\item By inserting the second inequality from Eq.~\eqref{eq:concavityS} into~\eqref{eq:Holevo}, we find $H_\mathrm{acc}(\mathcal{E}) \leq H(X)$. Since the inequality in~\eqref{eq:concavityS} saturates if and only if the states $\rho_x$ are mutually orthogonal, we see that $H_\mathrm{acc}(\mathcal{E}) < H(X)$ for non-orthogonal states. That is, if Alice prepares non-orthogonal states, Bob can never find the value of $X$ with certainty. 

\item After a moment of thought, we see that the Holevo bound implies that $n$ qubits cannot be used to transmit more than $n$ bits of classical information. 

\end{enumerate}
\end{remarks}

\begin{example}
Suppose that Alice communicates with Bob using the alphabet $\{\ketbra{0}{0},\ketbra{1}{1}\}$, where both states appear with a-priori probability $1/2$. However, their communication channel is noisy, acting as an amplitude damping channel $\mathcal{E}_\mathrm{AD}$, with $\mathcal{E}_\mathrm{AD}(\ketbra{0}{0})=\ketbra{0}{0}$, and $\mathcal{E}_\mathrm{AD}(\ketbra{1}{1})=(1-\gamma)\ketbra{1}{1} + \gamma \ketbra{0}{0}$, where $0\leq \gamma \leq 1$ is the strength of the noise [see Eqs.~\eqref{eq:E0E1ADamp} and~\eqref{eq:E0E1ADamp2}]. Hence, Bob receives 
\begin{align}
\begin{split}
\rho&=\sum_x p_x \mathcal{E}(\rho_x)=\frac{1+\gamma}{2} \ketbra{0}{0} + \frac{1-\gamma}{2} \ketbra{1}{1}
\end{split}
\end{align}
instead of $\rho=(\ketbra{0}{0}+\ketbra{1}{1})/2$ (from which he would be able to reconstruct $X$ with certainty by measuring the POVM with measurement operators $\{\ketbra{0}{0}, \ketbra{1}{1}\}$). How much information got lost due to the noisy channel? This can now be bounded using the Holevo bound: For our example, a short calculation of the Holevo information $\chi(\mathcal{E})$ reveals that 
\begin{align}
H_\mathrm{acc}(\mathcal{E}) \leq \chi(\mathcal{E}) = - \frac{1+\gamma}{2} \log\left(\frac{1+\gamma}{2}\right)  - \frac{1-\gamma}{2} \log\left(\frac{1-\gamma}{2}\right),
\end{align}
where $\chi(\mathcal{E})$ is a monotonously decreasing function in $\gamma$, with $\chi(\mathcal{E})=1$ for $\gamma=0$, and $\chi(\mathcal{E})=0$ for $\gamma=1$. That is, if there is no noise ($\gamma=0$), Bob can find a measurement with which he can learn $X$ with certainty, while for increasing noise $\gamma$, Bob's accessible information about $X$ decreases, with no information accessible in the case of maximal noise, $\gamma=1$. Hence, due to the noisy channel, we loose on average $\chi(\mathcal{E})|_{\gamma=0}-\chi(\mathcal{E})=1-\chi(\mathcal{E})$ bis of information per letter. 
\end{example}

\section{A brief introduction to entanglement theory} 
In Sec.~\ref{sec:bipureent} we briefly introduced entanglement for bipartite pure states. In particular, we defined a state $\ket{\psi}\in\mathcal{H}_\mathrm{A} \otimes \mathcal{H}_\mathrm{B}$ to be entangled if it cannot be written as a product state $\ket{\psi_\mathrm{A}} \otimes \ket{\psi_\mathrm{B}}$ for any choices of $\ket{\psi_\mathrm{A}} \in \mathcal{H}_\mathrm{A} $ and $\ket{\psi_\mathrm{A}} \in \mathcal{H}_\mathrm{B}$. That is, we defined entangled states as ``not being a product state''. 

More precisely, we call a bipartite state entangled if it contains correlations between subsystem $A$ and $B$, which cannot be described by classical correlations. For pure bipartite states, this is equivalent to being not a product state. The reason is that pure bipartite states cannot describe classical correlations between $A$ and $B$, such that all correlations of a pure bipartite state (i.e. everything beyond a product state) must be purely quantum. This is due to classical correlations being described by mixed states. For example, remember our example of two socks with correlated color, which we described by the mixed state $\rho_\mathrm{AB}=(\ketbra{10}{10}+\ketbra{01}{01})/2$. This state describes only classical correlations, and, most importantly, it is \textit{not} a product state. Hence, we see that beyond pure bipartite states, entanglement cannot be described as ``not being a product state''. Accordingly, we have to come up with a more precise criterion to decide whether a bipartite state contains correlations, which cannot be captured by classical correlations. In order to do so, we first have to develop a clear picture of classical correlations, and transformations, which cannot generate entanglement (i.e. any other correlations than classical correlations).

\subsection{LOCC, entanglement, and SEP} 
In order to see how a state without any quantum correlations looks like, let us start from a bipartite state $\rho_\mathrm{AB} \in \mathcal{D}(\mathcal{H}_\mathrm{A} \otimes \mathcal{H}_\mathrm{B})$ without any correlations between subsystem $A$ and $B$ (neither quantum nor classical), 
\begin{align}\label{eq:uncorrprodst}
\rho=\rho_\mathrm{A} \otimes \rho_\mathrm{B},
\end{align}
and ask how this state changes, if Alice (in possession of subsystem $A$) and Bob (in possession of subsystem $B$) perform operations on their system, which cannot generate more than classical correlations. First, we note that any operation $\mathcal{E}_\mathrm{loc}$ which acts locally on each subsystem cannot generate correlations. These local operations are described by the operation elements $\{E_j\otimes F_k\}_{j,k}$, mapping product states to product states,
\begin{align}
\mathcal{E}_\mathrm{loc}(\rho_\mathrm{A} \otimes \rho_\mathrm{B})=\left(\sum_j E_j \rho_\mathrm{A}E_j^\dagger \right) \otimes \left( \sum_k F_k \rho_\mathrm{B} F_k^\dagger \right).
\end{align}
However, there is more than that: Suppose that Alice and Bob can not only perform local operations, but also communicate via a classical channel. Obviously, this cannot allow Alice and Bob to generate more than classical correlations. That is, the resulting state under transformations based on \ul{local operations and classical communication} (LOCC) must be unentangled. 

Lets see how LOCC transformations affect the product state~\eqref{eq:uncorrprodst}. Suppose Alice performs a local operation $\{E_j\otimes \unit\}_j$ on her system (think of a measurement), and sends the outcome $j$ to Bob, who then performs a local operation $\{\unit \otimes F_{jk}\}_k$ conditioned on Alice's message, resulting in the map
\begin{align}\label{eq:LOCC1round}
\begin{split}
\mathcal{E}_\mathrm{LOCC}^{(1)}(\rho_\mathrm{A} \otimes \rho_\mathrm{B})&=\sum_{j,k} E_j \rho_\mathrm{A} E_j^\dagger \otimes F_{jk} \rho_\mathrm{B} F_{jk}^\dagger.
\end{split}
\end{align}
By the conditioned operation of Bob, this procedure introduces classical correlations between $A$ and $B$, such that~\eqref{eq:LOCC1round} is no longer a product state. Alice and Bob can now continue with this procedure, without generating any other than classical correlations. That is, Bob can send his outcome $k$ via classical communication to Alice, who performs a conditioned local operation $\{E_{jkl} \otimes \unit\}_l$ on her system, and sends her outcome to Bob, who continues in the same fashion. After arbitrary much rounds (possibly infinitely many), Alice and Bob end up with the state
\begin{align}\label{eq:LOCCproduct}
\mathcal{E}_\mathrm{LOCC}(\rho_\mathrm{A} \otimes \rho_\mathrm{B})&=\sum_{j,k,l,\cdots} \cdots E_{jkl}E_j \rho_\mathrm{A} E_j^\dagger E_{jkl}^\dagger \cdots \otimes  \cdots F_{jk} \rho_\mathrm{B} F_{jk}^\dagger \cdots.
\end{align}
This procedure of LOCC is fundamental for entanglement theory since it defines the set of operations under which entanglement cannot increase.

\begin{definition}\label{def:LOCC}
The set of \underline{local operations and classical communication} (LOCC)  on a bipartite state $\rho_\mathrm{AB} \in \mathcal{D}(\mathcal{H}_\mathrm{A}\otimes \mathcal{H}_\mathrm{B})$ is defined as the set of maps
\begin{align}
\mathcal{E}_\mathrm{LOCC}(\rho_\mathrm{AB})=\sum_{ j,k,l,\cdots} \cdots (E_{jkl}\otimes \unit)(\unit \otimes F_{jk})(E_j \otimes \unit) \rho_\mathrm{AB} (E_j^\dagger \otimes \unit)(\unit \otimes F_{jk}^\dagger)(E_{jkl}^\dagger\otimes \unit)\cdots,
\end{align}
where $\{E_{ jk \cdots l} \otimes \unit \}_l$ are local operations on subsystem $A$, and $\{ \unit \otimes F_{ j \cdots k} \}_k$ local operations on subsystem $B$. 
\end{definition}

Let us come back to the state~\eqref{eq:LOCCproduct} of Alice and Bob after transforming a product state $\rho_\mathrm{A} \otimes \rho_\mathrm{B}$ via LOCC operations. After a moment of thought, we see that it can be written as
\begin{align}\label{eq:sepLOCC}
\mathcal{E}_\mathrm{LOCC}(\rho_\mathrm{A} \otimes \rho_\mathrm{B})&=\sum_j p_j \ 
\rho_\mathrm{A}^j \otimes \rho^j_\mathrm{B},
\end{align}
where $p_j \geq 0$ are probabilities satisfying $\sum_j p_j =1$, i.e., $\{p_j\}_j$ is a classical probability distribution. For example, for the state~\eqref{eq:LOCC1round}, we have $p_j=\tr{E_j^\dagger E_j \rho_\mathrm{A} }$, $\rho^j_\mathrm{A} =E_j \rho_\mathrm{A} E_j^\dagger/p_j$, and $\rho^j_\mathrm{B}=\sum_k F_{kj} \rho_\mathrm{B} F_{kj}^\dagger$. That is, we found that via LOCC operations (which cannot generate entanglement), only states of the form~\eqref{eq:sepLOCC} can be generated, where the probabilities $p_j$ account for possible classical correlations. These states are called separable:

\begin{definition}
A state $\rho_\mathrm{ABC\cdots} \in \mathcal{D}(\mathcal{H}_\mathrm{A}\otimes \mathcal{H}_\mathrm{B} \otimes \mathcal{H}_\mathrm{C} \cdots)$ is said to be \underline{separable}, if it can be written in the form
\begin{align}\label{eq:separable}
\rho_\mathrm{ABC\cdots}=\sum_j p_j\ \rho_\mathrm{A}^j \otimes \rho_\mathrm{B}^j\otimes \rho_\mathrm{C}^j \otimes \cdots,
\end{align}
with $p_j \geq 0$, and $\sum_j p_j=1$. 
\end{definition}

Although we introduced LOCC operations merely for bipartite states, it can similarly be defined for multipartite states. Accordingly, for the sake of completeness, we here defined separable states for multipartite states. This definition allows us now to properly define entangled states:
\begin{definition}
All non-separable states $\rho_\mathrm{ABC\cdots} \in \mathcal{D}(\mathcal{H}_\mathrm{A}\otimes \mathcal{H}_\mathrm{B} \otimes \mathcal{H}_\mathrm{C} \cdots)$ are said to be \underline{entangled}.
\end{definition}

\begin{remarks}\leavevmode
\begin{enumerate}[1)]
\item Note that, given a separable, non-entangled state, by construction, LOCC operations cannot generate entangled states.

\item Can LOCC operations generate all separable states? Yes: Suppose we start from a product state $\rho_\mathrm{A}\otimes \rho_\mathrm{B} \otimes \rho_\mathrm{C} \otimes \cdots$, Alice samples from the probability distribution $\{p_j\}_j$, informs all other parties via classical communication about the outcome $j$, who then create the state $\rho_X^{j}$ via local operations. This results in the separable state~\eqref{eq:separable}.

\end{enumerate}
\end{remarks}

From definition~\ref{def:LOCC} of LOCC operations on bipartite states, we see that dealing with LOCC can quickly become complicated. A common trick is to consider the set of separable operations (SEP) instead, which is a proper superset of LOCC operations, i.e. $\mathrm{LOCC} \subset \mathrm{SEP}$, and then reducing to LOCC operations. Separable operations are defined as follows:
\begin{definition}
The set of \underline{separable operations} (SEP) on a bipartite state $\rho_\mathrm{AB}\in \mathcal{E}(\mathcal{H}_\mathrm{A}\otimes \mathcal{H}_\mathrm{B})$ is defined as the set of maps with operation elements $\{E_j \otimes F_j\}_j$, acting as
\begin{align}
\mathcal{E}_\mathrm{SEP}(\rho_\mathrm{AB})=\sum_j E_j \otimes F_j \rho_\mathrm{AB} E_j^\dagger \otimes F_j^\dagger.
\end{align}
\end{definition}

\begin{remarks}\leavevmode
\begin{enumerate}[1)]

\item Note that SEP operations map separable states to separable states.

\item Just like LOCC operations, SEP operations can generate all separable states by acting on product states. 

\item Since LOCC is a \textit{proper} subset of SEP, i.e. $\mathrm{LOCC} \subset \mathrm{SEP}$, there are operations in SEP, which don't belong to LOCC. Hence, while we require entanglement to be non-increasing under LOCC operations, the same is not necessarily true for SEP operations.

\end{enumerate}
\end{remarks}

\subsection{Transforming bipartite pure states via LOCC} 
So far, we argued that LOCC cannot generate entanglement. Accordingly, if a state $\rho$ can be transformed into $\sigma$ via LOCC, then $\rho$ must be ``at least as entangled'' as $\sigma$. As a consequence, if the LOCC operation is invertible, which is the case if and only if it corresponds to local unitary operations, then $\rho$ and $\sigma$ must have the same amount of entanglement. Those states are called local unitarily equivalent: 
\begin{definition}
Two states $\rho,\sigma \in \mathcal{D}(\mathcal{H}_\mathrm{A}\otimes \mathcal{H}_\mathrm{B} \otimes \mathcal{H}_\mathrm{C} \cdots)$ are said to be \ul{local unitarily equivalent} if there exists a local unitary transformation $U_\mathrm{A}\otimes U_\mathrm{B} \otimes U_\mathrm{C} \cdots$, with $U_X$ unitary, such that
\begin{align}
\rho=(U_\mathrm{A}\otimes U_\mathrm{B} \otimes U_\mathrm{C} \cdots)\sigma (U_\mathrm{A}^\dagger \otimes U_\mathrm{B}^\dagger \otimes U_\mathrm{C}^\dagger \cdots).
\end{align}
\end{definition}
These considerations are a good starting point to quantify the amount of entanglement of states. However, before we can continue in this direction, let us ask whether and, if so, how we can specify that $\rho$ can be transformed into $\sigma$ via LOCC. 

For simplicity, let us consider bipartite pure states $\ket{\psi},\ket{\phi}\in \mathcal{H}_\mathrm{A} \otimes \mathcal{H}_\mathrm{B}$. First, recall the Schmidt decomposition theorem~\eqref{eq:Schmidt}, which states that any bipartite pure state can be written in Schmidt form $\ket{\psi}=\sum_{j=1}^d \sqrt{\lambda_{\psi,j}} \ket{v_{\psi,j}}\otimes \ket{w_{\psi,j}}$, with Schmidt coefficients $\sqrt{\lambda_{\psi,j}} \geq 0$, and $\{\ket{v_{\psi,j}}\}_j$ and $\{\ket{w_{\psi,j}}\}_j$ an ON-basis of $\mathcal{H}_\mathrm{A}$ and $\mathcal{H}_\mathrm{B}$, respectively. Now suppose we want to compare the entanglement of $\ket{\psi}$ and $\ket{\phi}$: First, we can find local unitary operations (which, as we know, don't affect the state's entanglement), such that $\ket{\psi} \rightarrow \sum_{j=1}^d \sqrt{\lambda_{\psi,j}} \ket{v_j}\otimes \ket{w_j}$ and $\ket{\phi} \rightarrow \sum_{j=1}^d \sqrt{\lambda_{\phi,j}} \ket{v_j}\otimes \ket{w_j}$ only differ by their Schmidt coefficients. That is, deciding whether $\ket{\psi}$ can be transformed into $\ket{\phi}$ via LOCC can only depend on their Schmidt coefficients. In order to compare Schmidt coefficients, let us introduce the $d$-dimensional vector $\lambda_\psi=(\lambda_{\psi,1}, \dots, \lambda_{\psi,d})$ containing all squared Schmidt coefficients $\lambda_{\psi,j}$ of $\ket{\psi}$, and let us denote by $\lambda_{\psi,j}^\downarrow$ the squared Schmidt coefficients in decreasing order, i.e. $\lambda_{\psi,1}^\downarrow \geq \dots \geq\lambda_{\psi,d}^\downarrow$, and similar for $\ket{\phi}$. As we will see now, comparing the set of Schmidt coefficients via majorization will answer our question.

\begin{definition}
Given two vectors $x,y\in\mathbb{R}^d$, we say that $x$ is \underline{mayorized} by $y$, written $x \prec y$, if $\sum_{j=1}^k x_j^\downarrow \leq \sum_{j=1}^k y_j^\downarrow$ for all $k=1,\dots,d$, with equality instead of inequality when $k=d$. 
\end{definition}

\begin{theorem}\label{th:purestatetransLOCC}
Any bipartite pure state $\ket{\psi}\in\mathcal{H}_\mathrm{A}\otimes \mathcal{H}_\mathrm{B}$ can be transformed to another pure state $\ket{\phi}$ with unit probability via LOCC if and only if $\lambda_\psi \prec \lambda_\phi$. 
\end{theorem}

\begin{proof}
See page 576 in \cite{Nielsen-QC-2011}.
\end{proof}

\begin{remarks}\leavevmode
\begin{enumerate}[1)]

\item From this theorem, we can directly infer that there are \underline{maximally entangled states}, which must be local unitarily equivalent to 
\begin{align}\label{eq:maxentstates}
\ket{\psi_\mathrm{max}}= \frac{1}{\sqrt{d}} \sum_{j=1}^d \ket{j,j},
\end{align}
with $d=\min(d_\mathrm{A},d_\mathrm{B})$, and $d_\mathrm{A(B)}$ the dimension of $\mathcal{H}_\mathrm{A(B)}$. Hence, any pure state $\ket{\phi}$ can be generated with certainty from $\ket{\psi_\mathrm{max}}$ via LOCC. 

\item For a bipartite system of two qubits $\mathcal{H}=\mathbb{C}^2\otimes \mathbb{C}^2$, for which $d=2$, Eq.~\eqref{eq:maxentstates} shows that the Bell states~\eqref{eq:Bellstates} (which are local unitarily equivalent) are maximally entangled. 

\item As a direct consequence of theorem~\ref{th:purestatetransLOCC}, there are \underline{incomparable} states, i.e., pairs of states such that neither can be converted with certainty  into the other via LOCC. Hence, we cannot tell which state is more entangled than the other.

\end{enumerate}
\end{remarks}

\begin{example}
As an example for incomparable states, let $\mathcal{H}=\mathbb{C}^3\otimes \mathbb{C}^3$, such that $d=3$, and consider the states
\begin{align}
\begin{split}
\ket{\psi}&=\sqrt{0.4} \ket{00}+ \sqrt{0.4} \ket{11}+ \sqrt{0.2} \ket{22},\\
\ket{\phi}&=\sqrt{0.48} \ket{00}+ \sqrt{0.26} \ket{11}+ \sqrt{0.26} \ket{22},
\end{split}
\end{align}
for which $\lambda_\psi^\downarrow=(0.4,0.4,0.2)$, and $\lambda_\phi^\downarrow=(0.48,0.26,0.26)$. Let us now check whether $\lambda_\psi$ is majorized by $\lambda_\phi$ or vice versa:
\begin{align}
\begin{split}
\sum_{j=1}^1 \lambda_{\psi,j}^\downarrow=0.4 \quad &< \quad 0.48= \sum_{j=1}^1 \lambda_{\phi,j}^\downarrow\\
\sum_{j=1}^2 \lambda_{\psi,j}^\downarrow=0.8\quad &> \quad 0.74=\sum_{j=1}^2 \lambda_{\phi,j}^\downarrow\\
\sum_{j=1}^3 \lambda_{\psi,j}^\downarrow=1\quad &=  \quad 1=\sum_{j=1}^3 \lambda_{\phi,j}^\downarrow.
\end{split}
\end{align}
We see that neither $\lambda_\psi \prec \lambda_\phi$ nor $\lambda_\psi \succ \lambda_\phi$, hence, $\ket{\psi}$ and $\ket{\phi}$ are incomparable states. 
\end{example}

\begin{example}
The majorization criterion from theorem~\ref{th:purestatetransLOCC} tells whether $\ket{\psi}$ can be transformed to $\ket{\phi}$ via LOCC, but it doesn't tell how. To see how it works, consider the following example: Suppose we want to transform the maximally entangled two-qubit state $\ket{\psi}=(\ket{00}+\ket{11})/\sqrt{2}$ into $\ket{\phi}=\sqrt{\lambda_1}\ket{00}+\sqrt{\lambda_2}\ket{11}$. This must be doable, since $\lambda_\psi \prec \lambda_\phi$. In this case, the transformation works as follows: Suppose that Alice first measures her qubit according to the map $\{E_j\}_j$ with operation elements
\begin{align}
\begin{split}
E_1&= \sqrt{\lambda_1} \ketbra{0}{0}+ \sqrt{\lambda_2} \ketbra{1}{1},\\
E_2&= \sqrt{\lambda_1} \ketbra{0}{1}+ \sqrt{\lambda_2} \ketbra{1}{0},
\end{split}
\end{align}
satisfying $\sum_j E_j^\dagger E_j = \unit$. Suppose Alice sends the outcome $j\in\{1,2\}$ of her measurement via classical communication to Bob, who then transforms his state, conditioned on $j$, according to the unitary operation $F_{j1}$, where
\begin{align}
F_{11}= \ketbra{0}{0}+ \ketbra{1}{1},\quad F_{21}= \ketbra{0}{1}+ \ketbra{1}{0}.
\end{align}
Note that $F_{j1}$ is unitary for all $j$, i.e. $F_{j1}^\dagger F_{j1}=\unit$. The total LOCC operation is then described by
\begin{align}\label{eq:ExampleLOCC}
\mathcal{E}_\mathrm{LOCC}(\ket{\psi})= \sum_{j} (\unit\otimes F_{j1}) (E_j \otimes \unit) \ketbra{\psi}{\psi} (E_j^\dagger \otimes \unit)(\unit \otimes F_{j1}^\dagger).
\end{align}
To calculate the outcome, consider
\begin{align}
\begin{split}
 (\unit\otimes F_{11}) (E_1 \otimes \unit) \ket{\psi}&= (E_1\otimes F_{11})\frac{1}{\sqrt{2}}\left(\ket{00} \otimes \ket{11} \right)\\
 &=\frac{1}{\sqrt{2}}\left( E_1\ket{0}\otimes F_{11} \ket{0}+E_1\ket{1}\otimes F_{11} \ket{1}\right)\\
 &=\frac{1}{\sqrt{2}}\left( \sqrt{\lambda_1}\ket{00}+\sqrt{\lambda_2}\ket{11}\right)\\
 &=\frac{1}{\sqrt{2}} \ket{\phi},
\end{split}
\end{align}
and
\begin{align}
\begin{split}
 (\unit\otimes F_{21}) (E_2 \otimes \unit) \ket{\psi}&= (E_2\otimes F_{21})\frac{1}{\sqrt{2}}\left(\ket{00} \otimes \ket{11} \right)\\
 &=\frac{1}{\sqrt{2}}\left( E_2\ket{0}\otimes F_{21} \ket{0}+E_2\ket{1}\otimes F_{21} \ket{1}\right)\\
 &=\frac{1}{\sqrt{2}}\left( \sqrt{\lambda_2}\ket{11}+\sqrt{\lambda_1}\ket{00}\right)\\
 &=\frac{1}{\sqrt{2}} \ket{\phi}.
\end{split}
\end{align}
By using this in Eq.~\eqref{eq:ExampleLOCC}, we find
\begin{align}
\mathcal{E}_\mathrm{LOCC}(\ket{\psi})= \frac{1}{2} \ketbra{\phi}{\phi} +\frac{1}{2} \ketbra{\phi}{\phi} =\ketbra{\phi}{\phi},
\end{align}
which we were aiming for. 
\end{example}

\subsection{Entanglement monotones} 
Above we argued that entanglement, i.e., any quantity which possibly measures entanglement, must be non-increasing under LOCC operations. Such functions are called entanglement monotones: 
\begin{definition}
A function $\mathcal{M}: \mathcal{H} \rightarrow \mathbb{R}$ is said to be an \underline{entanglement monotone} if it is non-increasing under LOCC operations. 
\end{definition}

\begin{definition}
A function $\mathcal{M}: \mathbb{R}^d \rightarrow \mathbb{R}$ is called \underline{Schur-concave} if $\mathcal{M}(x) \geq \mathcal{M}(y)$ for all $x,y\in\mathbb{R}^d$ satisfying $x \prec y$. 
\end{definition}

\begin{remarks}\leavevmode
\begin{enumerate}[1)]

\item Any entanglement measure must be an entanglement monotone. However, entanglement measures must also satisfy additional properties, which we will discuss in a bit.  

\item Entanglement monotones must satisfy $\mathcal{M}(\rho)=\mathcal{M}(\sigma)$ if $\rho$ and $\sigma$ are local unitarily equivalent. As a result, local unitarily equivalent states must have the same amount of entanglement.

\item  Any entanglement monotone for bipartite pure states must be a Schur-concave function. Conversely, a Schur-concave function provides an entanglement monotone for bipartite pure states. 

\item If $\ket{\psi}$ can be transformed into $\ket{\phi}$ via LOCC, i.e. $\lambda_\psi \prec \lambda_\phi$, then $\mathcal{M}(\lambda_\psi) \geq \mathcal{M}(\lambda_\phi)$.

\end{enumerate}
\end{remarks}

Suppose now, we are given a function $\mathcal{M}: \mathbb{R}^d \rightarrow \mathbb{R}$. The following theorem tells us how to find out whether it is Schur-concave:

\begin{theorem}
A function $\mathcal{M}: \mathbb{R}^d \rightarrow \mathbb{R}$ is Schur-concave if and only if it is symmetric under permutations of its arguments $\lambda_1,\dots,\lambda_d$, and 
\begin{align}\label{eq:Schurconcavity}
(\lambda_1-\lambda_2)\left( \frac{\partial \mathcal{M}}{\partial \lambda_1} - \frac{ \partial\mathcal{M}}{\partial \lambda_2}\right) \leq 0.
\end{align}
\end{theorem}
\begin{proof}
Exercise. 
\end{proof}

\subsection{Some entanglement monotones and measures} 
We are now set to study entanglement measures in more detail. First, recall that any entanglement measure must be an entanglement monotone. However, there are more requirements (axioms) for potential entanglement measures. While there is no agreement upon the complete list of axioms, we here define entanglement measures via the following list of axioms:

\begin{definition}\label{def:entmeasure}
A function $E:\mathcal{H} \rightarrow \mathbb{R}$ is called an \underline{entanglement measure} if it satisfies
\begin{enumerate}[i)]
\item Monotonicity: $E$ must be an entanglement monotone,
\item Convexity: $E( p \rho + (1-p) \sigma) \leq p E(\rho) + (1-p) E(\sigma)$ for all $\rho,\sigma$ on $\mathcal{H}_\mathrm{A}\otimes \mathcal{H}_\mathrm{B}$ and $0\leq p \leq 1$,
\item Additivity: $E(\rho^{\otimes n})=n E(\rho)$ for all $\rho$ on $\mathcal{H}_\mathrm{A}\otimes \mathcal{H}_\mathrm{B}$,
\item Subadditivity: $E(\rho \otimes \sigma)\leq E(\rho)+E(\sigma)$ for all $\rho,\sigma$ on $\mathcal{H}_\mathrm{A}\otimes \mathcal{H}_\mathrm{B}$.
\end{enumerate}
\end{definition}

\subsubsection*{Bipartite pure states} 
Let us recall, that for a bipartite pure state $\ket{\psi}\in\mathcal{H}_\mathrm{A}\otimes \mathcal{H}_\mathrm{B}$ the squared Schmidt coefficients $\lambda_j$ of $\ket{\psi}$ are the eigenvalues of the reduced state $\rho_\mathrm{A(B)}=\trp{B(A)}{\ketbra{\psi}{\psi}}$ [cf. Eq.~\eqref{eq:Schmidteigenvaluesrhoa}]. With this in mind, we can now state an entanglement measure for bipartite pure states:

\begin{theorem}
For bipartite pure states $\ket{\psi}\in\mathcal{H}_\mathrm{A}\otimes \mathcal{H}_\mathrm{B}$ the von Neumann entropy of the reduced density matrix is an entanglement measure, 
\begin{align}\label{eq:EntanglementEntropy}
E(\ket{\psi})=S(\rho_\mathrm{A(B)})=-\tr{\rho_\mathrm{A(B)} \log \rho_\mathrm{A(B)} }=-\sum_j \lambda_j \log \lambda_j,
\end{align}
called \ul{entanglement entropy} (or \ul{entropy of entanglement}).
\end{theorem}

\begin{proof}
\textit{i) Monotonicity}: Clearly, $E(\ket{\psi})$ is invariant under permutations of the Schmidt coefficients $\lambda_j$. Hence, we are left with checking Schur-concavity, which can be done using Eq.~\eqref{eq:Schurconcavity}:
\begin{align}
\begin{split}
(\lambda_1-\lambda_2)\left(\frac{\partial E}{\partial \lambda_1} -\frac{\partial E}{\partial \lambda_2} \right)&=(\lambda_1-\lambda_2)\left(-1-\log \lambda_1 + 1 + \log \lambda_2 \right)\\
&=(\lambda_1-\lambda_2)(\log\lambda_2-\log\lambda_1)\\
&\leq 0.
\end{split}
\end{align}
\textit{ii) Convexity}: Convexity is always satisfied for pure states. To see this, consider $\rho=\ketbra{\psi}{\psi}$ and $\sigma=\ketbra{\phi}{\phi}$, then $p \rho + (1-p) \sigma$ is pure (which we require, since we consider merely pure states) for all $0\leq p \leq 1$ if and only if $\ket{\psi}=\ket{\phi}$, such that $E(p \rho + (1-p) \sigma)=E(\rho)=E(\sigma)$. \\
\textit{iii) Additivity and iv) Subadditivity}: Consider the bipartite pure states $\ket{\psi}\in\mathcal{H}_\mathrm{A}\otimes \mathcal{H}_\mathrm{B}$ and $\ket{\phi}\in\mathcal{H}_\mathrm{C}\otimes \mathcal{H}_\mathrm{D}$, and the reduced density matrices $\rho_\mathrm{A}=\trp{B}{\ketbra{\psi}{\psi}}$ and $\rho_\mathrm{C}=\trp{D}{\ketbra{\phi}{\phi}}$. Then 
\begin{align}
\begin{split}
E(\ket{\psi} \otimes \ket{\phi})&= S( \trp{B,D}{ \ketbra{\psi}{\psi} \otimes \ketbra{\phi}{\phi}})\\
&=S(\rho_\mathrm{A} \otimes \rho_\mathrm{C})\\
&\overset{th.~\ref{th:basicpropvonNeumann}}{=} S(\rho_\mathrm{A} )+S(\rho_\mathrm{C})\\
&=E(\ket{\psi})+E(\ket{\phi}),
\end{split}
\end{align}
which proves both additivity and subadditivity.
\end{proof}

Another quantity frequently used in the literature is the so-called concurrence. The concurrence is \textit{not} an entanglement measure, but an entanglement monotone:

\begin{theorem}
For bipartite pure states $\ket{\psi}\in\mathcal{H}_\mathrm{A}\otimes \mathcal{H}_\mathrm{B}$ with reduced density operator $\rho_\mathrm{A(B)}$, the \underline{concurrence}
\begin{align}\label{eq:concurrencePure}
c(\ket{\psi})=\sqrt{2\left(1-\tr{\rho_\mathrm{A(B)}^2} \right)}
\end{align}
is an entanglement monotone. 
\end{theorem}

\begin{proof}
The concurrence is an entanglement monotone if it is Schur-concave. To prove Schur-concavity, let us write $c(\ket{\psi})$ in terms of the squared Schmidt coefficients of $\ket{\psi}$, i.e. the eigenvalues of $\rho_\mathrm{A(B)}$,
\begin{align}\label{eq:concurrenceEigval}
c(\ket{\psi})=\sqrt{2\left(1-\sum_j \lambda_j^2\right)}.
\end{align}
First, we see that $c(\ket{\psi})$ is invariant under permutations of the squared Schmidt coefficients $\lambda_j$. Hence, we are left with showing that $c(\ket{\psi})$ satisfies Eq.~\eqref{eq:Schurconcavity},
\begin{align}
\begin{split}
(\lambda_1-\lambda_2)\left( \frac{\partial c}{\partial \lambda_1} - \frac{\partial c}{\partial \lambda_2}   \right)&=(\lambda_1-\lambda_2)\left(-\frac{2 \lambda_1}{c} + \frac{2 \lambda_2}{c} \right)\\
&=-\frac{2}{c}(\lambda_1-\lambda_2)^2\\
&\leq 0,
\end{split}
\end{align}
which finishes the proof. 
\end{proof}

\begin{remarks}\leavevmode
\begin{enumerate}[1)]

\item Note that for both, the entanglement entropy and concurrence, we have $ E(\ket{\psi}),c(\ket{\psi}) \geq 0$, with equality if and only if $\ket{\psi}$ is separable (i.e. non-entangled). On the other hand, their upper bound differs, $E(\ket{\psi})\leq \log d$, and $c(\ket{\psi})\leq \sqrt{2(1-1/d)}$, with equality if and only if $\ket{\psi}$ is maximally entangled. (Recall that $d$ is the number of Schmidt coefficients)

\item For a two-qubit system $\mathcal{H}=\mathbb{C}^2 \otimes \mathbb{C}^2$ the concurrence is often defined as
\begin{align}\label{eq:concurrencequbit}
c(\ket{\psi})= |\bra{\psi^*} Y \otimes Y \ket{\psi} |,
\end{align}
with $Y$ the Pauli-$Y$ matrix from Eq.~\eqref{eq:Pauli}, and $\bra{\psi^*}$ the complex conjugate of $\bra{\psi}$, with the complex conjugation performed in the computational basis $\{\ket{0},\ket{1}\}$ (i.e. the eigenbasis of the Pauli-$Z$ matrix), such that for $\bra{\psi}=\sum_{j,k} \psi^*_{j,k} \bra{j,k}$, we have $\bra{\psi^*}=\sum_{j,k} \psi_{j,k} \bra{j,k}$. 

\item For a two-qubit state $\ket{\psi}=\sum_{j,k} \psi_{j,k} \ket{j,k}$, the concurrence~\eqref{eq:concurrencequbit} can also be written as
\begin{align}
c(\ket{\psi})=2\abs{ \psi_{00}\psi_{11}- \psi_{01}\psi_{10}}.
\end{align}

\end{enumerate}
\end{remarks}

\begin{proof}
2), 3): Exercise. 
\end{proof}

\subsubsection*{Bipartite mixed states} 
Finding an entanglement monotone for mixed states is not as simple as for pure states. The reason is that pure states can only describe quantum correlation, but mixed states can describe both quantum and classical correlations. However, there is a nice workaround known as the convex roof construction.  

Let's consider a bipartite mixed state $\rho=\sum_j p_j \ketbra{\psi_j}{\psi_j}$. Since it corresponds to a mixture of pure states $\ket{\psi_j}$ appearing with associated probabilities $p_j \geq 0$, it is suggestive to generalize an entanglement monotone $\mathcal{M}(\ket{\psi})$ for pure states by the average value $\sum_j p_j \mathcal{M}(\ket{\psi_j})$. However, this construction bears one problem: A density operator has no unique pure state decomposition [cf. theorem~\ref{th:unitaryfreedomdensity}]. Indeed, different pure state decompositions can lead to different average values of the entanglement monotone. This problem can be solved by taking the infimum over all pure state decompositions (i.e. the minimal average value), called convex roof, which must be an entanglement monotone. 
\begin{definition}
The \ul{convex roof} of an entanglement monotone $\mathcal{M}(\ket{\psi})$ for pure states is defined as
\begin{align}\label{eq:convexroof}
\mathcal{M}(\rho)=\inf_{\{p_j,\ket{\psi_j}\}_j} \sum_j p_j \mathcal{M}(\ket{\psi_j}),
\end{align}
with the infimum running over all pure state decompositions $\rho=\sum_j p_j \ketbra{\psi_j}{\psi_j}$. 
\end{definition}

With this at hand, we can directly apply this to the entanglement entropy, resulting in an entanglement monotone for bipartite mixed states:
\begin{definition}
The convex roof of the entanglement entropy, 
\begin{align}\label{eq:EOF}
E(\rho)=\inf_{\{p_j,\ket{\psi_j}\}_j} \sum_j p_j E(\ket{\psi_j}),
\end{align}
is called \ul{entanglement of formation}.
\end{definition}

The entanglement of formation is not additive, i.e. it does not satisfies item \textit{iii)} of our requirements for an entanglement measure from definition~\ref{def:entmeasure}. Hence, it is no entanglement measure according to our definition.

Note that solving the optimization problem in~\eqref{eq:convexroof} and~\eqref{eq:EOF} is in general a computationally hard task. Fortunately, solutions are known for some cases, in particular for the case of mixed two-qubit states. To this end, we now restrict to bipartite \textit{two-level} systems (i.e. two qubits).

\subsubsection*{Bipartite mixed two-level systems} 

\begin{theorem}\label{th:concmixedtwo}
Let $c(\rho)$ be the convex roof of the concurrence for a bipartite two-level system. Then 
\begin{align}\label{eq:concurrencetwolevel}
c(\rho)=\max(0,\lambda_1-\lambda_2-\lambda_3-\lambda_4),
\end{align}
where $\lambda_j$ are the eigenvalues in decreasing order of the Hermitian matrix $R=\sqrt{\sqrt{\rho}\tilde{\rho}\sqrt{\rho} }$, with $\tilde{\rho}=(Y\otimes Y)\rho^* (Y \otimes Y)$, $Y$ the Pauli-$Y$ matrix, and the complex conjugation of $\rho$ taken with respect to the computational basis (i.e. the eigenbasis of the Pauli-$Z$ matrix). 
\end{theorem}

\begin{proof}
See W. K. Wootters, \textit{Phys. Rev. Lett.} \textbf{80}, 2245 (1998), or Ref.~\cite{Mintert-BC-2009}. 
\end{proof}

By theorem~\ref{th:concmixedtwo}, we found a solution to the optimization problem of the convex roof construction $c(\rho)=\inf_{\{p_j,\ket{\psi_j}\}_j} \sum_j p_j c(\ket{\psi_j})$ of the concurrence. Now, note that any state $\ket{\psi} \in \mathbb{C}^2\otimes \mathbb{C}^2$ is described by two Schmidt coefficient, whose square must sum to unity, i.e. $\lambda_1+\lambda_2=1$. Hence, $c(\ket{\psi})=c(\lambda_1)$ only depends on a single Schmidt coefficient. Accordingly, we can find the inverse function $\lambda_1(c)$, which, using Eq.~\eqref{eq:concurrenceEigval}, reads
\begin{align}\label{eq:lambdaSSc}
\lambda_{1}(c)=\frac{1 + \sqrt{1-c^2}}{2}, \quad \lambda_{2}(c)=\frac{1-\sqrt{1-c^2}}{2}.
\end{align}
Similarly, since the entanglement entropy~\eqref{eq:EntanglementEntropy} is also a function of the Schmidt coefficients, for bipartite two-level systems, it can only depend on a single Schmidt coefficient given in~\eqref{eq:lambdaSSc}, i.e. $E(\ket{\psi})=E(\lambda_1(c))$.  As we show in the following, this allows us to find a solution for the entanglement of formation~\eqref{eq:EOF} of mixed bipartite two-level systems in terms of the solution~\eqref{eq:concurrencetwolevel} of the concurrence.

\begin{definition}\label{def:E}
The function $\mathcal{E}:\mathbb{R} \rightarrow \mathbb{R}$, $c\mapsto \mathcal{E}(c)$, is defined as
\begin{align}
\mathcal{E}(c)=-\frac{1+\sqrt{1-c^2}}{2} \log\left( \frac{1+\sqrt{1-c^2}}{2}\right) -\frac{1-\sqrt{1-c^2}}{2} \log \left(\frac{1-\sqrt{1-c^2}}{2}\right).
\end{align}
\end{definition}

Note that this function is simply the Shannon entropy of the Schmidt coefficients from Eq.~\eqref{eq:lambdaSSc}. With this definition we can now state the relation between the concurrence and the entanglement of formation for bipartite two-level systems. 

\begin{theorem}
For $\rho$ a mixed state of a bipartite two-level system with concurrence $c(\rho)$, the entanglement of formation is given by $E(\rho)=\mathcal{E}(c(\rho))$, with $\mathcal{E}$ from Def.~\ref{def:E}.
\end{theorem}

\begin{proof}
See S. A. Hill and W. K. Wootters, \textit{Phys. Rev. Lett.} \textbf{78}, 5022 (1997).
\end{proof}

\begin{example}
As an example, let us calculate the concurrence~\eqref{eq:concurrencetwolevel} for the bipartite mixed state $\rho=p \ketbra{\phi^+}{\phi^+} + (1-p) \unit/4$, with $\ket{\phi^+}=(\ket{00}+\ket{11})/\sqrt{2}$ the maximally entangled Bell state from Eq.~\eqref{eq:Bellstates}, and $0\leq p \leq 1$. Let us first calculate $\tilde{\rho}$. To this end, note that $\rho^*=\rho$, and $Y^2=\unit$, such that
\begin{align}
\begin{split}
\tilde{\rho}&=(Y\otimes Y)\rho^*(Y \otimes Y) \\
&=p (Y\otimes Y)\ketbra{\phi^+}{\phi^+}(Y\otimes Y) + \frac{1-p}{4} \unit.
\end{split}
\end{align}
Next, using $Y \ket{0}=\im \ket{1}$ and $Y \ket{1}=-\im \ket{0}$, we get $(Y\otimes Y) \ket{\phi^+}=-\ket{\phi^+}$, and, accordingly, 
\begin{align}
\begin{split}
\tilde{\rho}&=p \ketbra{\phi^+}{\phi^+} + \frac{1-p}{4} \unit\\
&=\rho.
\end{split}
\end{align}
With this, we find the Hermitian matrix $R=\sqrt{\sqrt{\rho} \tilde{\rho} \sqrt{\rho}}=\rho$, i.e. it is equal to the state $\rho$. Next, a short calculation reveals the eigenvalues of $R=\rho$ to be $\lambda_1=(1+3p)/4$, and $\lambda_2=\lambda_3=\lambda_4=(1-p)/4$. Accordingly, the concurrence~\eqref{eq:concurrencetwolevel} becomes 
\begin{align}
\begin{split}
c(\rho)&=\max(0,\lambda_1-\lambda_2-\lambda_3-\lambda_4)\\
&=\max\left( 0, \frac{3p-1}{2} \right)\\
&=\begin{cases}
0 & \text{for } 0\leq p \leq 1/3 \\
\frac{3p-1}{2} & \text{for } 1/3 < p \leq 1. 
\end{cases}
\end{split}
\end{align}
This is indeed and interesting result: Although $\rho$ describes a mixture of the maximally entangled state $\ket{\phi^+}$ and the maximally mixed state $\unit/4$ for all values $p>0$, it is not entangled in the range $0 \leq p \leq 1/3$. 
\end{example}

\subsection{Positive partial transpose and negativity} 
Recall that for mixed bipartite states, so far we found algebraically computable entanglement monotones (i.e. expressions of entanglement monotones involving e.g. no optimization procedure over an infinitely large set) merely for a Hilbert space dimension of $2\times 2$, i.e. for two-qubit systems. In the following we will formulate an algebraically computable entanglement monotone for bipartite mixed states for any finite Hilbert space dimension. To this end, we first consider a common method how to decide whether a a state is separable or not, and, based on this, we then formulate such an entanglement monotone.

First recall that a separable bipartite state reads $\rho=\sum_j p_j \rho_\mathrm{A}^j \otimes \rho_\mathrm{B}^j$ [see Eq.~\eqref{eq:separable}], with $p_j\geq 0$ and $\sum_j p_j=1$. Next, consider a positive map $\mathcal{E}:\mathcal{H}_\mathrm{A} \rightarrow \mathcal{H}_\mathrm{A} $ such that $\mathcal{E}(A) \geq 0$ for all $A \geq 0$ and $A$ on $\mathcal{H}_\mathrm{A}$. An important property of positive maps is that the extension $\mathcal{E}\otimes \unit$ to a composite system $\mathcal{H}_\mathrm{A} \otimes \mathcal{H}_\mathrm{B}$ is not necessarily positive (if the extension was positive, we would call $\mathcal{E}$ \ul{completely positive} [see Eq.~\eqref{eq:P1}]). Hence, if the extended map $\mathcal{E}\otimes \unit$ is not positive, then there exist states $\sigma \in \mathcal{D}(\mathcal{H}_\mathrm{A} \otimes \mathcal{H}_\mathrm{B})$, such that $(\mathcal{E} \otimes \unit) (\sigma) \ngeq 0$. This, however, cannot happen for product states $\rho=\sum_j p_j \rho_\mathrm{A}^j \otimes \rho_\mathrm{B}^j$, since, due to $\mathcal{E}(\rho_\mathrm{A}^j) \geq 0$,
\begin{align}
(\mathcal{E}\otimes \unit)(\rho)=\sum_j p_j \mathcal{E}(\rho_\mathrm{A}^j) \otimes \rho_\mathrm{B}^j
\end{align}
must be positive, i.e. $(\mathcal{E}\otimes \unit)(\rho)\geq 0$. However, note that the fact that $(\mathcal{E}\otimes \unit)(\sigma)\geq 0$ for a single map $\mathcal{E}$ does not necessarily imply that the state $\sigma$ is separable. 

\begin{theorem}
A bipartite state $\rho \in \mathcal{D}(\mathcal{H}_\mathrm{A} \otimes \mathcal{H}_\mathrm{B})$ is separable if and only if
\begin{align}
(\mathcal{E}\otimes \unit)(\rho)\geq 0
\end{align}
for all positive maps $\mathcal{E}: \mathcal{H}_\mathrm{A} \rightarrow \mathcal{H}_\mathrm{A}$.
\end{theorem}
\begin{proof}
See M. Horodecki, et al., \textit{Phys. Lett. A}, \textbf{223}, 1 (1996).
\end{proof}

As a direct result, in order to conduct that $\rho$ is entangled, it is sufficient to find a single positive map $\mathcal{E}$, such that $(\mathcal{E}\otimes \unit)(\rho)\ngeq 0$. It turns out that a large class of entangled states can be concluded as being entangled via the transposition $\mathcal{E}=T$, which is a positive map. 

\begin{definition}
The \ul{partial transpose} of $\rho\in\mathcal{D}(\mathcal{H}_\mathrm{A} \otimes \mathcal{H}_\mathrm{B})$ is defined as
\begin{align}
\rho^\mathrm{pt}=(T\otimes \unit)(\rho),
\end{align}
with the transposition taken in the computational basis.
\end{definition}

\begin{theorem}[\textbf{Peres-Horodecki criterion or PPT criterion}]
A bipartite state $\rho\in\mathcal{D}(\mathcal{H}_\mathrm{A} \otimes \mathcal{H}_\mathrm{B})$ acting on a Hilbert space with dimension $\mathrm{dim}(\mathcal{H}_\mathrm{A} \otimes \mathcal{H}_\mathrm{B})=2\times 2$ or $2\times 3$ is separable if and only if it has a positive partial transpose, i.e.
\begin{align}
\rho^\mathrm{pt} \geq 0.
\end{align}
\end{theorem}
\begin{proof}
See A. Peres, \textit{Phys. Rev. Lett.} \textbf{77}, 1413 (1996).
\end{proof}

Hence, for Hilbert spaces of dimension $2\times 2$ and $2\times 3$, the PPT criterion allows us to conclude whether a state is entangled or not. Moreover, the PPT criterion can be used to define an algebraically computable entanglement monotone for bipartite mixed states acting on a Hilbert space of any finite dimension.

\begin{theorem}
For $\rho\in\mathcal{D}(\mathcal{H}_\mathrm{A} \otimes \mathcal{H}_\mathrm{B})$ a bipartite mixed state, the \ul{negativity}
\begin{align}
\mathcal{N}(\rho)= \frac{\norm{\rho^\mathrm{pt}}_1-1}{2}
\end{align}
is an entanglement monotone, where $\norm{A}_1=\tr{\abs{A}}$ is the trace norm of an Hermitian operator $A$. 
\end{theorem}
\begin{proof}
See G. Vidal and R. F. Werner, \textit{Phys. Rev. A} \textbf{65}, 032314 (2002)
\end{proof}

\begin{remarks}\leavevmode
\begin{enumerate}[1)]

\item Note that while $\tr{\rho^\mathrm{pt}}=1$, we have that $\norm{\rho^\mathrm{pt}}_1 = \tr{\abs{\rho^\mathrm{pt}}} \geq 1$, with equality if and only if $\rho^\mathrm{pt} \geq 0$. Hence, $\mathcal{N}(\rho) \geq 0$, with equality if and only if $\rho^\mathrm{pt} \geq 0$. 

\item The negativity is an entanglement monotone for arbitrary finite Hilbert space dimensions, which can be computed with standard linear algebra packages. 

\item Note that for Hilbert space dimension larger than $2\times 3$, the negativity vanishes for some entangled states (since there exist entangled states with positive partial transpose for systems larger than $2\times 3$). Hence, it is no entanglement measure. 
\end{enumerate}
\end{remarks}

\end{document}